%% file: main.tex
\documentclass[11pt,a4paper]{article}
\usepackage[lmargin=1.0in,rmargin=1.0in,bottom=1.0in,top=1.0in,twoside=False]{geometry}
\usepackage[english]{babel}
\usepackage[shortlabels,]{enumitem}

\usepackage{tikz-cd}
\usetikzlibrary{decorations.pathmorphing}
\usepackage{fontaxes}
\usepackage{trimspaces}
\usepackage{nccfoots}
\usepackage{setspace}
\usepackage{inconsolata}
\usepackage{libertine}
\usepackage{amsmath,amssymb,amsthm}
\usepackage{mathtools}
\usepackage{comment}
\usepackage[all,defaultlines=4]{nowidow}
\usepackage{microtype}
\usepackage{mathrsfs}
\usepackage{mathtools}
\usepackage{marvosym}
\usepackage{thm-restate}
\usepackage{tikz}
\usepackage{todonotes}
\usetikzlibrary{calc}
\usetikzlibrary{shapes}
\definecolor{darkgreen}{rgb}{10,117,28}
\usepackage{wrapfig}
\usepackage{diagbox}
\usepackage{stackengine}
\usepackage{mathrsfs}

\newcommand{\CC}{\mathscr C}
\definecolor{blue}{rgb}{0.1,0.2,0.5}
\definecolor{brown}{rgb}{0.6,0.6,0.2}
\usepackage[ocgcolorlinks, linkcolor={blue}, citecolor={blue}]{hyperref}
\usepackage{cleveref}
\usepackage{subcaption}
\crefformat{page}{#2page~#1#3}%
\Crefformat{page}{#2Page~#1#3}%
\crefformat{equation}{#2(#1)#3}%
\Crefformat{equation}{#2(#1)#3}%
\crefformat{figure}{#2Figure~#1#3}%
\Crefformat{figure}{#2Figure~#1#3}%
\crefformat{section}{#2Section~#1#3}
\Crefformat{section}{#2Section~#1#3}
\crefformat{chapter}{#2Chapter~#1#3}
\Crefformat{chapter}{#2Chapter~#1#3}
\crefformat{chapter*}{#2Chapter~#1#3}
\Crefformat{chapter*}{#2Chapter~#1#3}
\crefformat{part}{#2Part~#1#3}
\Crefformat{part}{#2Part~#1#3}
\crefformat{enumi}{#2(#1)#3}
\Crefformat{enumi}{#2(#1)#3}

\newtheorem{theorem}{Theorem}[section]
\crefformat{theorem}{#2Theorem~#1#3}
\Crefformat{theorem}{#2Theorem~#1#3}
\newtheorem{lemma}[theorem]{Lemma}
\crefformat{lemma}{#2Lemma~#1#3}
\Crefformat{lemma}{#2Lemma~#1#3}

\crefformat{conjecture}{#2Conjecture~#1#3}
\Crefformat{conjecture}{#2Conjecture~#1#3}

\newcommand{\newtheoremwithcrefformat}[2]{%
  \newtheorem{#1}{#2}[section]%
  \crefformat{#1}{##2\MakeUppercase#1~##1##3}%
  \Crefformat{#1}{##2\MakeUppercase#1~##1##3}%
}

\newtheoremwithcrefformat{proposition}{Proposition}
\newtheoremwithcrefformat{observation}{Observation}
\newtheoremwithcrefformat{corollary}{Corollary}
\newtheoremwithcrefformat{claim}{Claim}
\newtheoremwithcrefformat{example}{Example}
\newtheoremwithcrefformat{remark}{Remark}
\newtheoremwithcrefformat{definition}{Definition}

\newenvironment{claimproof}{%
\begin{proof}[Proof of the claim]}{%
	\end{proof}
    }

\makeatletter
\def\ifenv#1{
   \def\@tempa{#1}%
   \ifx\@tempa\@currenvir
      \expandafter\@firstoftwo
    \else
      \expandafter\@secondoftwo
   \fi
}
\let\wfs@comment@comment\comment
\let\comment\@undefined
\newcommand{\untoto}{\let\toto\@undefined}
\usepackage{changes}
\let\wfs@changes@comment\comment
\let\comment\@undefined

\newcommand\comment{%
    \ifthenelse{\equal{\@currenvir}{comment}}
    {\wfs@comment@comment}
    {\wfs@changes@comment}%
}
\makeatother

\newcommand{\FO}{\ensuremath{\mathsf{FO}}}
\newcommand{\FOsep}{\ensuremath{\mathsf{FO}\hspace{0.6pt}\texttt{\normalfont{+}}\hspace{1pt}\conn}}
\newcommand{\MSO}{\ensuremath{\mathsf{MSO}}}
\newcommand{\CMSO}{\ensuremath{\mathsf{CMSO}}}
\newcommand{\FPT}{\ensuremath{\mathsf{FPT}}}
\newcommand{\AWstar}{\ensuremath{\mathsf{AW}[\star]}}
\newcommand{\FOMSO}{{\FO(\MSO(\preceq,A)\cup\Sigma)}}

\newcommand{\val}{\mathsf{val}}
\newcommand{\anc}{\mathsf{anc}}

\newcommand{\state}{\mathrm{state}}

\usepackage{anyfontsize}

\newcommand{\Cc}{\mathscr{C}}

\renewcommand{\phi}{\varphi}

\newcommand{\set}[1]{\ensuremath{\{#1\}}} 
\newcommand{\setof}[2]{\set{#1\mid#2}} 
\newcommand{\from}{\colon} 

\newcommand{\Lsf}{\mathsf{L}}
\newcommand{\Rsf}{\mathsf{R}}
\newcommand{\Ssf}{\mathsf{S}}

\newcommand{\Basic}{B}

\newcommand{\Ab}{\mathbb{A}}
\newcommand{\Bb}{\mathbb{B}}
\newcommand{\Cb}{\mathbb{C}}

\newcommand{\Gb}{\mathbb{G}}
\newcommand{\Hb}{\mathbb{H}}
\newcommand{\dom}[1]{\mathrm{dom}({#1})}
\newcommand{\lft}{\mathsf{left}}
\newcommand{\rgt}{\mathsf{right}}
\newcommand{\abstr}[1]{\llbracket #1\rrbracket}
\newcommand{\partto}{\rightharpoonup}

\newcommand{\str}[1]{\mathbf{#1}} 
\renewcommand{\subset}{\subseteq}

\renewcommand{\leq}{\leqslant}
\renewcommand{\geq}{\geqslant}
\renewcommand{\le}{\leqslant}
\renewcommand{\ge}{\geqslant}
\newcommand{\tup}[1]{\bar{#1}}

\newcommand{\N}{\mathbb{N}}

\usepackage{stmaryrd}

\renewcommand{\setminus}{-}

\newcommand{\Kb}{\mathbb{K}}
\newcommand{\Lb}{\mathbb{L}}
\newcommand{\Db}{\mathbb{D}}

\newcommand{\Oh}{\mathcal{O}}

\newcommand{\Tt}{\mathcal{T}}

\newcommand{\Sh}{\widehat{S}}

\newcommand{\conn}{\mathsf{conn}}
\newcommand{\lca}{\mathsf{lca}}
\newcommand{\dir}{\mathsf{dir}}

\newcommand{\bag}{\mathsf{bag}}
\newcommand{\cone}{\mathsf{cone}}
\newcommand{\cmp}{\mathsf{comp}}
\newcommand{\mrg}{\mathsf{mrg}}
\newcommand{\adh}{\mathsf{adh}}
\newcommand{\bgraph}{\mathsf{bgraph}}
\newcommand{\torso}{\mathsf{torso}}
\newcommand{\parent}{\mathsf{parent}}
\newcommand{\children}{\mathsf{children}}
\newcommand{\profile}{\mathsf{profile}}

\usepackage{accents}

\newcommand{\wh}[1]{\widehat{#1}}

\renewcommand{\preceq}{\preccurlyeq}
\renewcommand{\succeq}{\succcurlyeq}

\newcommand{\ERCagreement}{
\vspace{-6pt}

\noindent
{\begin{minipage}[t]{\textwidth}\small This paper is a part of project {\sc{BOBR}} that has received funding from the European Research Council (ERC) under the European Union's Horizon 2020 research and innovation programme (grant agreement No 948057) and part of the French-German Collaboration ANR/DFG Project UTMA supported by the German
Research Foundation (DFG) through grant agreement
No 446200270. \end{minipage}\hfill

\begin{minipage}{.32\textwidth}\includegraphics[width=\textwidth]{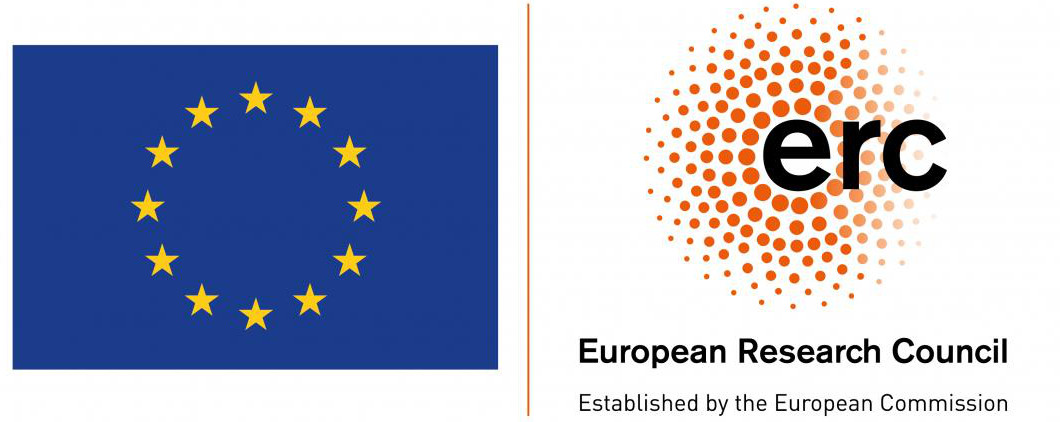}\includegraphics[width=0.8\textwidth]{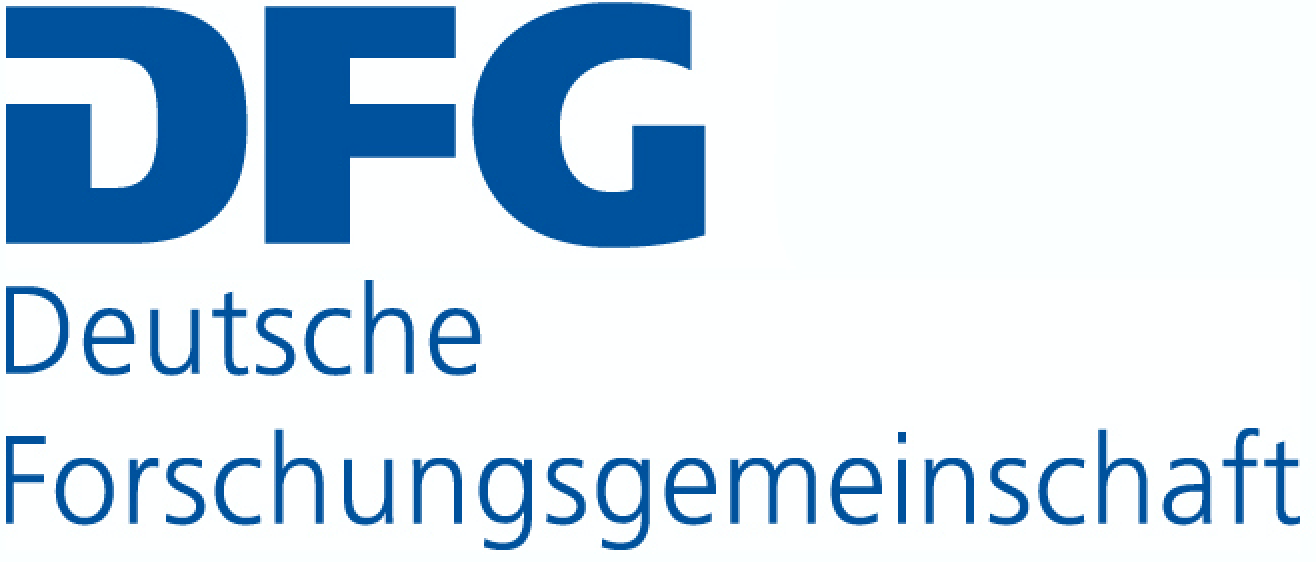}\end{minipage}\hfill}}

\makeatletter
\newcommand{\@abbrev}[3]{
  \def\c@a@def##1{
      \if ##1.
        \relax
      \else
        \@ifdefinable{\@nameuse{#1##1}}{\@namedef{#1##1}{#2##1}}
        \expandafter\c@a@def
      \fi
    }
  \c@a@def #3.
}
\@abbrev{bb}{\mathbb}{ABCDEFGHIJKLMNOPQRSTUVWXYZ}
\@abbrev{bf}{\mathbf}{ABCDEFGHIJKLMNOPQRSTUVWXYZabcdefghijklmnopqrstuvwxyz}
\@abbrev{bit}{\boldsymbol}{ABCDEFGHIJKLMNOPQRSTUVWXYZabcdefghijklmnopqrstuvwxyz}
\@abbrev{cal}{\mathcal}{ABCDEFGHIJKLMNOPQRSTUVWXYZ}
\@abbrev{frak}{\mathfrak}{ABCDEFGHIJKLMNOPQRSTUVWXYZabcdefghijklmnopqrstuvwxyz}
\@abbrev{rm}{\mathrm}{ABCDEFGHIJKLMNOPQRSTUVWXYZabcdefghijklmnopqrstuvwxyz}
\@abbrev{scr}{\mathscr}{ABCDEFGHIJKLMNOPQRSTUVWXYZ}
\@abbrev{sf}{\mathsf}{ABCDEFGHIJKLMNOPQRSTUVWXYZabcdefghijklmnopqrstuvwxyz}
\@abbrev{Alg}{\Alg}{ABCDEFGHIJKLMNOPQRSTUVWXYZ}
\@abbrev{Str}{\Str}{ABCDEFGHIJKLMNOPQRSTUVWXYZ}
\@abbrev{set}{\mathbb}{ABCDEFGHIJKLMNOPQRSTUVWXYZ}
\@abbrev{tup}{\tup}{ABCDEFGHIJKLMNOPQRSTUVWXYZabcdefghijklmnopqrstuvwxyz}
\@abbrev{tld}{\tilde}{ABCDEFGHIJKLMNOPQRSTUVWXYZabcdefghijklmnopqrstuvwxyz}
\makeatother

\usepackage{wasysym}
\usepackage{xspace}

\title{Algorithms and data structures for first-order logic with connectivity under vertex failures}

\author{
Micha\l{} Pilipczuk\thanks{University of Warsaw, Poland, \texttt{michal.pilipczuk@mimuw.edu.pl}}
\and
Nicole Schirrmacher \thanks{University of Bremen, Germany, \texttt{schirrmacher@uni-bremen.de}}
\and
Sebastian Siebertz\thanks{University of Bremen, Germany, \texttt{siebertz@uni-bremen.de}}
\and
Szymon Toru\'nczyk\thanks{University of Warsaw, Poland, \texttt{szymtor@mimuw.edu.pl}}
\and
Alexandre Vigny\thanks{University of Bremen, Germany, \texttt{vigny@uni-bremen.de}}
}

\date{}

\begin{document}
\maketitle

\input{abstract}
\vfill
\ERCagreement

\pagebreak
\input{stoc-intro}
\input{overview}
\input{preliminaries}

\input{queries}
\input{torso-queries}

\input{tradeoff}
\input{tree-structures}
\input{model-checking}

\input{first-levels}

\input{conclusions}

\paragraph*{Acknowledgements.} The authors thank Mikołaj Bojańczyk for helpful discussions on separator logic, as well as Anna Zych-Pawlewicz for pointers to the literature on connectivity oracles under vertex failures.

\bibliography{ref}

\appendix

\end{document}

%% file: abstract.tex

\thispagestyle{empty}

\begin{abstract}
  \noindent We introduce a new data structure for answering
  connectivity queries in undirected graphs subject to batched vertex
  failures. Precisely, given any graph $G$ and integer parameter $k$,
  we can in fixed-parameter time construct a data structure that can
   later be used to answer queries of the form: ``are vertices $s$
  and $t$ connected via a path that avoids vertices
  $u_1,\ldots, u_k$?'' in time $\smash{2^{2^{\Oh(k)}}}$. In the terminology of
  the literature on data structures, this gives the first
  deterministic data structure for
  connectivity under vertex failures where for every fixed number of
  failures, all operations can be performed in constant time.

  With the aim to understand the power and the limitations of our new
  techniques, we prove an algorithmic meta theorem for the recently
  introduced \emph{separator logic}, which extends first-order logic
  with atoms for connectivity under vertex failures. We prove that the
  model-checking problem for separator logic is fixed-parameter
  tractable on every class of graphs that exclude a fixed topological
  minor. We also show a weak converse.
  This implies that from the point of view of
  parameterized complexity, under standard complexity theoretical
  assumptions, the frontier of tractability of separator logic is
  almost exactly delimited by classes excluding a fixed topological
  minor.

  The backbone of our proof relies on a decomposition theorem of Cygan,
  Lokshtanov, Pilipczuk, Pilipczuk, and Saurabh~[SICOMP '19],
  which provides a tree decomposition of a given
  graph into bags that are unbreakable. Crucially, unbreakability
  allows to reduce separator logic to plain first-order logic within
  each bag individually. Guided by this observation, we design our
  model-checking algorithm using dynamic programming over the tree
  decomposition, where the transition at each bag amounts to running a
  suitable model-checking subprocedure for plain first-order
  logic. This approach is robust enough to provide also an extension
  to efficient enumeration of answers to a query expressed in
  separator logic.
\end{abstract}

\setcounter{page}{0}


%% file: stoc-intro.tex

\newcommand{\Ll}{{\cal L}}

\section{Introduction}\label{sec:stoc-intro}

\subsection{Connectivity under vertex failures}

In many applications, we do not need to answer a query once but rather
need to repeatedly answer queries over dynamically changing data. A
prime example are databases, where a database is repeatedly modified
and queried by the users.  In most cases, the modifications to the
database can be expected to be small compared to its size, which
suggests that in a preprocessing step, we can set up a data structure
that is efficiently updated after each modification and that allows
for efficient querying.  In 1997, Frigioni and Italiano introduced the
\emph{dynamic subgraph model}~\cite{frigioni1997dynamically}, which
allows to conveniently model such dynamic situations in the particular
case of connectivity in network infrastructures that are subject to
node or link failures.  This model, which has been studied intensively
in the data structures community under the name of
{\em{connectivity oracles under vertex failures}}, is as
follows. We are given a
graph $G$ and an integer $k$. We imagine that $G$ is a network in
which at every moment, at most $k$ vertices (or edges) are inactive
(are subject to {\em{failures}}). The goal is to construct a data
structure that supports the following two operations. First, one can
{\em{update}} $G$ by resetting the set of failed vertices to a given
set of size at most $k$. Second, one can {\em{query}} $G$ by asking,
for a given pair of vertices $s$ and~$t$, whether $s$ and~$t$ can be
connected by a path that avoids the failed vertices.

Following the introduction of the problem by Frigioni and
Italiano~\cite{frigioni1997dynamically} and the more general problem
of batched vertex or edge failures by P\u{a}tra\c{s}cu and
Thorup~\cite{PatrascuT07}\footnote{Strictly speaking,
  in~\cite{PatrascuT07} P\u{a}tra\c{s}cu and Thorup considered only
  edge failures. To the best of our knowledge, the first to
  systematically investigate (batched) vertex failures on general
  graphs were Duan and Pettie in~\cite{DuanP10}.}, there has been a long line of work on data
structures for connectivity oracles under vertex and edge failures.
We refer to a remarkably comprehensive literature overview in the work
of Duan and Pettie~\cite{DuanP20}, which also provides the currently
best bounds for the problem as far as combinatorial and deterministic
data structures are concerned. Their data structure can be initialized
in time $\Oh(|G|\|G\|\cdot \log |G|)$, takes $\Oh(k\|G\|\log |G|)$
space, and supports updates in $\Oh(k^3\log^3 |G|)$ time and queries
in $\Oh(k)$ time. Here~$\|G\|$ is the vertex count of $G$ and $\|G\|$ is the joint number of edges and
vertices in $G$.  These bounds can be slightly improved at the cost of
allowing randomization, however, the polylogarithmic dependency on the
size of the graph in the update time persists. As noted
in~\cite{DuanP20}, from known results it is possible to derive data
structures with constant time complexity of operations for $k\leq 3$
(or $k\leq 4$ for edge failures), but, citing their words, ``scaling
these solutions up, even to an arbitrarily large constant $k$, becomes
prohibitively complex, even in the simpler case of edge failures.''

Recently, van den Brand and Saranurak~\cite{BrandS19} proposed a very
different approach to the problem, using which they obtained a
randomized data structure that can be initialized in time
$\Oh(|G|^\omega)$, takes $\Oh(|G|^2\log |G|)$ space, and supports
updates in $\Oh(k^\omega)$ time and queries in $\Oh(k^2)$ time. In particular, the update and the query time are constant for constant $k$;
to the best of our knowledge, this is the only data structure that has
this property known so far. The approach is algebraic and, simplifying
it substantially, boils down to storing a matrix of counts of walks
between pairs of vertices and extracting answers to queries using
algebraic operations on those counts. To manipulate the counts
efficiently, one needs to work on their short hashes, for instance in
the form of elements from a fixed finite field. This introduces
randomization to the picture and avoiding it within this methodology
seems difficult. We remark that the data structure of van den Brand
and Saranurak~\cite{BrandS19} works even for the more general problem
of reachability in directed graphs under vertex failures, and can
handle arc insertions.

As the first main contribution of this paper, we prove the following
theorem.  Note that querying whether two vertices are connected by a
path avoiding failed vertices is equivalent to an update followed by a
connectivity query in the terminology of previous
works~\cite{BrandS19,DuanP10,DuanP20,PatrascuT07}.


\begin{theorem}\label{thm:general-queries}
  Given a graph $G$ and an integer $k$, one can in time
  $2^{2^{\Oh(k)}}\cdot |G|^2 \|G\|$ construct a data structure that
  may answer the following queries in time $2^{2^{\Oh(k)}}$: given
  $s,t\in V(G)$ and $k$ vertices $u_1,\ldots,u_k$ of~$G$, are $s$ and
  $t$ connected in $G$ by a path that avoids the vertices
  $u_1,\ldots, u_k$? The space usage of the data structure is
  $2^{2^{\Oh(k)}}\cdot \|G\|$.
\end{theorem}

Note that in \cref{thm:general-queries}, for every fixed $k$, all
operations are supported in constant time (which is doubly-exponential
in $k$). Also, the data structure is entirely deterministic and purely
combinatorial. To the best of our knowledge, this is the first
deterministic data structure with constant time queries. Note
also that for a fixed $k$, the space usage of our data structure is
linear in the size of the graph, as opposed to the quadratic
dependency in the result of van den Brand and
Saranurak~\cite{BrandS19}.

Our approach to proving \cref{thm:general-queries} is completely new
compared to the previous approaches
\mbox{\cite{BrandS19,DuanP10,DuanP20,PatrascuT07}}. Our key
combinatorial observation is that connectivity queries over a constant
number~$k$ of vertex failures can be evaluated in constant time
provided the underlying graph~$G$ is sufficiently well-connected. The
requirement on the well-connectedness of $G$ depends on the number~$k$. Obviously, there is no guarantee that the input graph $G$
satisfies this property. We therefore use a decomposition theorem of
Cygan, Lokshtanov, Pilipczuk, Pilipczuk, and Saurabh~\cite{CyganLPPS19} that (roughly) states the following
(see \cref{thm:strong-unbreakability} for a precise statement). For
every graph $G$ and fixed number $k$, there is a tree decomposition
where every intersection of adjacent bags has bounded size and every bag
is well-connected for parameter~$k$. Moreover, such a tree
decomposition can be computed in time $2^{\Oh(k^2)}|G|^2\|G\|$. The
data structure of \cref{thm:general-queries} is constructed around the tree decomposition $\Tt$ provided by the theorem of Cygan et
al.~\cite{CyganLPPS19}. Intuitively speaking, whether $s$ and $t$ are connected via a path that
avoids the vertices $u_1,\ldots, u_k$ can be decided using dynamic programming over $\Tt$. We employ known
techniques developed for answering queries on trees in constant time, for instance a deterministic variant of Simon's
factorization~\cite{Colcombet07}, to be able to compute the outcome of this dynamic programming efficiently for any query given on input.
We provide a more detailed overview of the proof of
\cref{thm:general-queries} in \cref{sec:overview}.  The full proof is
provided in \cref{sec:general-queries}.

The time complexity guarantees provided by \cref{thm:general-queries}
can be improved by allowing larger space usage. We give a variant of
\cref{thm:general-queries} (\cref{thm:tradeoff}) that offers queries
in time $k^{\Oh(1)}$ after initialization in time
$2^{\Oh(k\log k)}\cdot |G|^{\Oh(1)}$, but assumes space usage
$k^{\Oh(1)}\cdot |G|^2$. The main idea here is to use a later variant
of the decomposition of Cygan et al.~\cite{CyganLPPS19}, proved
in~\cite{CyganKLPPSW21}, that provides quantitatively better bounds on
the adhesion and unbreakability of the decomposition, at the cost of
relaxing a technical condition on the unbreakability (see the
difference between strong and weak unbreakability, explained in
\cref{sec:prelims}). The need of dealing with this relaxation turns
out to be a nuisance that, nonetheless, can be overcome in our
approach.


\subsection{Algorithmic meta theorems}

Once an algorithmic technique has been applied to solve a specific
problem, it is very desirable to understand the power and the
limitations of that technique, that is, to understand which other
problems can be solved with that technique and which problems
cannot. An approach to this very general goal is to prove an
\emph{algorithmic meta theorem}: a theorem that establishes
tractability for a whole class of problems,
possibly on certain restricted input instances. As logic allows to
conveniently define classes of algorithmic problems, namely, the class
of all problems that can be expressed in the logic under
consideration, algorithmic meta theorems are often formulated as
model-checking problems for a logic. In the {\em{model-checking
    problem}} for a logic $\Ll$ on a class of graphs~$\CC$ we are
given a graph $G\in\CC$ and a sentence $\varphi\in \Ll$, and the task
is to decide whether $\varphi$ holds in $G$.  The archetypal result of
this form is a result of Courcelle stating that every formula of
monadic-second order logic ($\MSO_2$) can be evaluated in linear time
on all graphs whose treewidth is bounded by some fixed constant~\cite{Courcelle90}.
This
result has been extended to various other logics $\Ll$ and graph
classes $\CC$.

This approach to algorithmic meta theorems can be made precise in the
language of parameterized complexity, as follows.  Say that~$\Ll$ is
\emph{tractable} on~$\CC$ if the model-checking problem for~$\Cc$
and~$\Ll$ is {\em{fixed-parameter tractable (fpt)}}: it can be solved
in time $f(\varphi)\cdot \|G\|^c$, for some computable function $f$
and a constant $c$, both depending only on $\CC$.  Results about
tractability of a logic $\Ll$ on a class $\Cc$ imply the tractability
of a vast array of problems --- namely all problems that can be
expressed in the logic $\Ll$ --- on a given class of graphs.  For
instance, if model-checking \emph{first-order logic} ($\FO$) is fpt on
a class of graphs~$\CC$, then in particular the dominating set problem
is fpt on $\CC$, since the existence of a dominating set of size $k$
can be expressed by a first-order sentence with $k+1$ quantifiers that
range over the vertices of the graph.  Similarly, the independent set
problem is then also fpt on $\CC$.  On the other hand, if the more
powerful logic $\MSO_2$ is tractable on a class of graphs $\CC$, then
the $3$-colorability problem, or the hamiltonicity problem are
polynomial-time solvable on $\CC$, since both those problems can be
expressed using an $\MSO_2$ sentence, whose quantifiers that range
over \emph{sets} of vertices and edges of the graph.

There has been a long line of work on fixed-parameter tractability of
model-checking $\FO$, as it is arguably the most fundamental logic.
This culminated in the work of Grohe, Kreutzer and
Siebertz~\cite{GroheKS17}, who proved that this problem is
fixed-parameter tractable on every class that is {\em{nowhere
    dense}}. Those classes include for example the class of planar
graphs, or every class with bounded maximum degree, or of bounded
genus.  Without going into details, nowhere denseness is a general
notion of uniform sparseness in graphs, and it is broader than most
well-studied concepts of sparseness considered in structural graph
theory, such as excluding a fixed (topological) minor.  As observed by
Dvo\v{r}\'ak et al.~\cite{DvorakKT13}, the result of Grohe et al.\ is
tight in the following sense: whenever $\Cc$ is not nowhere dense and
is closed under taking subgraphs, then $\FO$ model-checking on $\Cc$
is as hard as on general graphs, that is, $\AWstar$-hard. This means
that as far as subgraph-closed classes are concerned, sparsity ---
described formally through the notion of nowhere denseness --- exactly
delimits the area of tractability of $\FO$ and the limits of the
locality method for $\FO$ model-checking.

The aforementioned classic theorem of Courcelle~\cite{Courcelle90} states
that the model-checking problem for $\MSO_2$ is fixed-parameter
tractable on a class $\Cc$ if $\Cc$ has bounded treewidth. The proof
translates the given sentence into a suitable tree automaton working
on a tree decomposition of the input graph; this approach brings a
wide range of automata-based tools to the study of $\MSO_2$ on
graphs. Again, as shown by Kreutzer and
Tazari~\cite{kreutzer2010lower} and Ganian et
al.~\cite{ganian2014lower}, under certain complexity assumptions this
is (almost) the most one could hope for: on subgraph-closed classes
whose treewidth is poly-logarithmically unbounded, model-checking of
$\MSO_2$ becomes intractable from the parameterized perspective. So
for $\MSO_2$, bounded treewidth (almost) delimits the frontier of
tractability, at least for subgraph-closed classes, and automata-based techniques exactly explain what can be done algorithmically.

With the aim to provide a meta theorem that captures the essence of
the techniques developed to answer connectivity under vertex failures,
we search for a logic that can express these queries. It is easy to
see that $\MSO_2$ is strong enough, but in fact, $\MSO_2$ is much
stronger and as explained above, it is intractable beyond graphs of
bounded treewidth. On the other hand, $\FO$ can only express local
properties and in particular, it cannot even express the very simple
query of whether two vertices are in the same connected component. This
naturally leads to an extension of $\FO$ that may use connectivity
under vertex failures as atomic formulas. Interestingly, such a logic
was only very recently introduced independently by
Boja\'nczyk~\cite{Bojanczyk21} and by Schirrmacher et
al.~\cite{SchraderSV21} under the name \emph{separator logic}, and
denoted $\FOsep$. This logic extends $\FO$ by predicates
$\conn_k(s,t,u_1,\ldots,u_k)$ (one for each $k\ge 0$) with the
following semantics: if $s,t,u_1,\ldots,u_k$ are vertices of a graph
$G$, then $\conn_k(s,t,u_1,\ldots,u_k)$ holds if and only if there is
a path that connects $s$ with $t$ and does not pass through any of the
vertices $u_1,\ldots,u_k$.  Thus, separator logic involves very basic
connectivity queries that have a global character, unlike plain~$\FO$,
which is local. On the other hand, while $\conn_k(s,t,u_1,\ldots,u_k)$
predicates are expressible in $\MSO_2$, it is no surprise that the
expressive power of $\FOsep$ is strictly weaker than that of
$\MSO_2$. For example, $\MSO_2$ can express bipartiteness, while
$\FOsep$ cannot~\cite[Theorem~3.11]{SchraderSV21}. Separator logic can
still express many interesting algorithmic properties, such as the
feedback vertex set problem and many recently studied elimination
distance problems~\cite{bulian2016graph}.

As the second main result of the paper, we prove an algorithmic meta
theorem for separator logic. We show that the frontier of tractability
of $\FOsep$ is almost exactly delimited by classes of graphs that
exclude a fixed topological minor. More precisely, we prove the
following complementary results.

\begin{restatable}{theorem}{mainub}\label{thm:main-ub}
  Let $\Cc$ be a class of graphs that exclude a fixed graph as a
  topological minor. Then given $G\in \Cc$ and an $\FOsep$ sentence
  $\varphi$, one can decide whether $\varphi$ holds in $G$ in time
  $f(\varphi)\cdot \|G\|^3$, where $f$ is a computable function
  depending on $\CC$.
\end{restatable}

For the lower bound, we will need a technical assumption.  We say that
a class $\Cc$ admits \emph{efficient encoding of topological minors}
if for every graph $H$ there exists $G\in \Cc$ such that $H$ is a
topological minor of $G$, and, given $H$, such $G$ together with a
suitable topological minor model can be computed in time polynomial in
$|H|$. Note that in particular such $G$ can be only polynomially
larger than $H$.

\begin{restatable}{theorem}{mainlb}\label{thm:main-lb}
  Let $\Cc$ be a subgraph-closed class of graphs that admits efficient
  encoding of topological minors. Then the model-checking of\,
  $\FOsep$ on $\Cc$ is $\AWstar$-hard, even for the fragment that uses
  only $\conn_k$ predicates with $k\le 2$.
\end{restatable}

On the one hand, these results show that separator logic is a natural
logic whose expressive power lies strictly between $\FO$ and $\MSO_2$,
and which corresponds, in the sense described above, to a natural
property of classes of graphs that lies between bounded treewidth and
nowhere denseness, see~\cref{fig:corresp}.
%
%
\begin{figure}

  \begin{center}
    \begin{tikzcd}
      \FO \arrow[r, "\subsetneq"]
      \arrow[leftrightsquigarrow]{d}{\cite{GroheKS17}}
      & \FOsep \arrow[r, "\subsetneq"] \arrow[leftrightsquigarrow]{d}{(\ast)}
      & \MSO_2\arrow[leftrightsquigarrow]{d}{\cite{Courcelle90,kreutzer2010lower}}\\
      \parbox{2cm}{\centering
      \textit{nowhere dense}}

      &
      \parbox{2cm}{\centering
      \textit{excluded top.~minor}}
      \arrow[l, swap, "\supsetneq"]
      &
      \parbox{2cm}{\centering
      \textit{bounded treewidth}}
      \arrow[l, swap, "\supsetneq"]
    \end{tikzcd}
  \end{center}
  \caption{Correspondences between logics and properties of graph classes. An arrow $\cal L\leftrightsquigarrow \cal P$
  between a logic $\cal L$ and a property $\cal P$ of graph classes denotes a correspondence:
  the logic $\cal L$ is tractable on a graph class $\CC$ if and, under additional assumptions, only if, the class $\CC$ has the property $\cal P$.
  Our result is the central correspondence $(\ast)$.}
  \label{fig:corresp}
\end{figure}
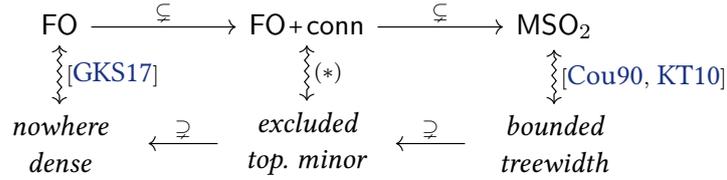
On the other hand, the proof of \cref{thm:main-ub} takes the technique
of dynamic programming over tree decompositions with unbreakable parts
to the limit. The tractability result generalizes several recent
algorithmic results for concrete problems 
expressible in
separator logic~\cite{agrawal2021fpt,bulian2017fixed,LindermayrSV20}.

The proof of \cref{thm:main-lb} is easy (see \cref{sec:overview-mc}),
so we focus on sketching the proof of \cref{thm:main-ub}, which
is the main technical contribution of the second part of the paper.
%
Let us fix
an $\FOsep$ sentence~$\varphi$, and let $k+2$ be the maximum arity of
connectivity predicates appearing in $\varphi$. Again, the key observation is
that $\conn_k$ predicates can be rewritten to
plain $\FO$ provided the graph is well-connected, and the reason for this is the same combinatorial observations that underlies the proof of \cref{thm:general-queries}.
Coming back to our dynamic programming on the tree decomposition,
we choose to present it using
a new framework based on automata, as this brings us closer
to the classic understanding of logic on tree-decomposable
graphs. Intuitively, the automaton processes a given tree in a
bottom-up manner, but we assume that on the children of every node
there is an additional structure of a graph. We say that such trees are \emph{augmented} with graphs.
When the automaton processes a node~$x$,
it chooses a transition based on an $\FO$ query executed on the
graph on the children of $x$, where each child is labeled with the
state computed for it before in the run. Then testing whether a
formula $\varphi\in \FOsep$ holds on a graph $G$ is reduced to deciding
whether an automaton constructed from $\varphi$ accepts such an
{\em{augmented tree}}, constructed from the tree decomposition
provided by the theorem of Cygan et al.~\cite{CyganLPPS19}.

The benefit of this approach is that other questions related to the
logic $\FOsep$ can be reduced in the same way to questions about tree
automata over augmented trees.  We showcase this by proving that given
an $\FOsep$ query $\varphi(\bar x)$ with free variables and a graph
$G$ from a fixed topological-minor-free class $\Cc$, the answers to
$\varphi(\bar x)$ on $G$ can be enumerated with constant delay after
fpt preprocessing (see \cref{thm:main-enumeration} for a precise
statement). Similarly, the formula $\varphi(\bar x)$ can be queried in
constant time after fpt preprocessing (see
\cref{thm:main-query-answering}).  In general, we believe that the
framework of automata over augmented trees may be of independent
interest, as it seems to be a convenient model for understanding
dynamic programming procedures over tree-decomposable structures.

Similar fpt query-answering and enumeration algorithms for
plain~\FO~have been obtained for classes with bounded expansion
\cite{DBLP:journals/lmcs/KazanaS19,10.1145/3375395.3387660} and
nowhere dense classes \cite{DBLP:conf/pods/SchweikardtSV18} and for
$\MSO_2$~on trees and graphs of bounded
treewidth~\cite{Bagan06,KazanaS13,AmarilliBMN19}.  Our general
approach based on augmented trees allows us to generalize those results
to trees that are augmented with graphs coming from a nowhere dense
class, and to a certain logic combining the power of \MSO~on trees and
\FO~on graphs.

\medskip
Let us comment on the novelty of using the decomposition theorem of
Cygan et al.~\cite{CyganLPPS19}.  The result of Cygan et
al.~\cite{CyganLPPS19} has been used several times for various graph
problems~\cite{CyganKLPPSW21,CyganLPPS19,Lokshtanov0S20,LoksthanovPPS21,SaurabhZ18}. Our
application is the most similar to (and in fact, draws inspiration
from) the work of Lokshtanov et al.~\cite{LokshtanovR0Z18}, who proved
the following statement: for every $\CMSO_2$ sentence~$\varphi$ and
$k\in \N$, model-checking $\varphi$ on general graphs can be
reduced to model-checking $\varphi$ on $(q,k)$-unbreakable graphs,
where $q$ is a constant depending on $\varphi$ and $k$. As $\FOsep$ is
subsumed by $\CMSO_2$, this result can be almost applied to establish
\cref{thm:main-ub}. The caveat is that the unbreakable graph output by
the reduction needs to admit efficient $\FO$ model-checking so that
the considered $\FOsep$ sentence can be decided in fpt time after
rewriting it to plain $\FO$. In essence, we show that it is possible
to guarantee this provided the input graph excludes a fixed
topological minor. We remark that the work of Lokshtanov et
al.~\cite{LokshtanovR0Z18} does not use the decomposition theorem of
Cygan et al.~\cite{CyganLPPS19} directly, but relies on its conceptual
predecessor, the {\em{recursive understanding}}
technique~\cite{ChitnisCHPP16,KawarabayashiT11}.

\paragraph*{Organization.}
In \cref{sec:overview} we give an overview of the proofs of
our main results: \cref{thm:general-queries},
\cref{thm:main-ub}, and \cref{thm:main-lb}.  In the main body of the paper, after formal
preliminaries, we present the proof of \cref{thm:general-queries} in
\cref{sec:general-queries} and
\cref{sec:torso-queries}. \cref{sec:tradeoff} presents
\cref{thm:tradeoff}, a variant of \cref{thm:general-queries} that
offers faster queries.  The proof of \cref{thm:main-ub} is presented
in \cref{sec:automata} and~\ref{sec:model-checking}, devoted
respectively to introducing the framework of automata over augmented
trees and applying them to $\FOsep$ on classes excluding a fixed
topological minor. \cref{thm:main-lb} is proved in \cref{sec:lbs}.

%% file: overview.tex

\section{Overview}\label{sec:overview}

In this section we provide an overview of the proofs of our main results: \cref{thm:general-queries}, \cref{thm:main-ub}, and \cref{thm:main-lb}.
We start with \cref{thm:general-queries}, as this proof exposes all main ideas in a purely algorithmic setting. Then we discuss how
the same methodology can be used for \cref{thm:main-ub}. In this
introduction we assume familiarity with basic terminology of tree
decompositions; see \cref{sec:prelims} for details.

\subsection[Data structure for connectivity under vertex failures:
proof of Theorem~\ref{thm:general-queries}]{Data structure for
  connectivity under vertex failures: proof of
  \cref{thm:general-queries}}

\paragraph*{Unbreakability.}
As mentioned in \cref{sec:stoc-intro}, the central idea of this work
is to tackle problems involving connectivity predicates using a decomposition into well-connected --- or, more formally, {\em{unbreakable}} --- parts. We
first need to recall a few definitions.

A \emph{separation} in a graph $G$ is a pair $(A,B)$ of vertex subsets
such that $A\cup B=V(G)$ and there are no edges in $G$ between
$A\setminus B$ and $B\setminus A$. The {\em{order}} of the separation
is the cardinality of the {\em{separator}}~$A\cap B$. A vertex subset
$X$ is {\em{$(q,k)$-unbreakable}} in a graph $G$ if for every
separation $(A,B)$ in~$G$ of order at most $k$ in~$G$, either
$|A\cap X|\leq q$ or $|B\cap X|\leq q$. So intuitively speaking, a
separation of order $k$ cannot break $X$ in a balanced way: one of the
sides must contain at most $q$ vertices of $X$.

The notion of unbreakability has been implicitly introduced in the
context of parameterized algorithms by Kawarabayashi and Thorup in
their work on the {\sc{$k$-Way Cut}} problem~\cite{KawarabayashiT11}.
This work has brought about the method of {\em{recursive
    understanding}}, which was then made explicit by Chitnis, Cygan,
Hajiaghayi, Pilipczuk, and Pilipczuk in the technique of
{\em{randomized contractions}}~\cite{ChitnisCHPP16}. Intuitively,
recursive understanding is a Divide\&Conquer scheme using which one
can reduce problems on general graphs to problems on suitably
unbreakable graphs, provided certain technical conditions are
satisfied. The technique was also applied in the context of
model-checking: Lokshtanov, Ramanujan, Saurabh, and
Zehavi~\cite{LokshtanovR0Z18} proved that the problem of deciding a
property expressed in $\CMSO_2$ in fpt time on general graphs can be
reduced to the same problem on suitably unbreakable graphs. This
result was used to provide algorithms for computing elimination
distance to certain graph
classes~\cite{DBLP:conf/stacs/AgrawalKP0021,DBLP:conf/lics/FominGT21,JansenK21},
which is very much related to separator logic
(see~\cite[Example~3.5]{SchraderSV21}).

In this work, we will not use recursive understanding or randomized
contractions {\em{per se}}, but their conceptual successor: the
decomposition into unbreakable parts. Precisely, the following theorem
was proved by Cygan, Lokshtanov, Pilipczuk, Pilipczuk, and
Saurabh~\cite{CyganLPPS19}.

\begin{theorem}[\cite{CyganLPPS19}]\label{thm:ov:strong-unbreakability}
  There is a function $q(k)\in 2^{\Oh(k)}$ such that given a graph $G$
  and $k\in \N$,~one can in time $2^{\Oh(k^2)}\cdot |G|^2\|G\|$
  construct a (rooted) tree decomposition $\Tt=(T,\bag)$ of $G$
  satisfying the~following.
  \begin{itemize}[nosep]
  \item All adhesions in $\Tt$ are of size at most $q(k)$; and
  \item every bag of $\Tt$, say at node $x$, is
    $(q(k),k)$-unbreakable in the subgraph of $G$ induced by the
    union of bags at all descendants\footnote{We follow the
      convention that every node is its own ancestor and descendant.}
    of $x$.
  \end{itemize}
\end{theorem}

We remark that a different variant of
\cref{thm:ov:strong-unbreakability} was later proved by Cygan, Komosa,
Lokshtanov, Pilipczuk, Pilipczuk, and
Wahlstr\"om~\cite{CyganKLPPSW21}. This variant provides much stronger
unbreakability guarantees --- $q(k)=k$ --- at the cost of relaxing the
unbreakability to hold only in the whole graph $G$. See
\cref{sec:prelims} for a discussion. The variant
of~\cite{CyganKLPPSW21} can be also used in our context and this leads
to improving some quantitative bounds, however for simplicity, we focus
on \cref{thm:ov:strong-unbreakability} in this overview.

A typical usage of \cref{thm:ov:strong-unbreakability} is to
apply bottom-up dynamic programming on the obtained tree
decomposition $\Tt=(T,\bag)$. The subproblem for a node $x$ of $T$
corresponds to finding partial solutions in the subgraph of $G$
induced by the union of bags at descendants of $x$, for every possible
behavior of such a partial solution on the adhesion connecting $x$
with its parent. That this adhesion has size at most $q(k)$ gives an
upper bound on the number of different behaviors. To solve such a
subproblem one needs to aggregate the solutions computed for the
children of $x$ by solving an auxiliary problem on the subgraph
induced by the bag of $x$. In various problems of interest, the
unbreakability of the bag becomes helpful in solving the auxiliary
problem. This general methodology has been successfully used to give
multiple fpt algorithms for cut
problems~\cite{CyganLPPS19,CyganKLPPSW21,SaurabhZ18}, most prominently
for {\sc{Minimum Bisection}}~\cite{CyganLPPS19}. More recent uses
include a parameterized approximation scheme for the {\sc{Min
    $k$-Cut}} problem~\cite{Lokshtanov0S20} and a fixed-parameter
algorithm for {\sc{Graph Isomorphism}} parameterized by the size of
the excluded minor~\cite{LoksthanovPPS21}. On a high level, we apply
the same methodology here.

\paragraph*{The case of totally unbreakable graphs.} Let us come back
to the problem of evaluating connectivity queries: we are given a
graph $G$ and after some preprocessing of $G$, we would like to be
able to quickly evaluate for given vertices $s,t,u_1,\ldots,u_k$
whether $\conn_k(s,t,u_1,\ldots,u_k)$ holds. Consider first the
following special case: the whole $G$ is $(q,k)$-unbreakable for some
parameter $q$ (or more formally, $V(G)$ is $(q,k)$-unbreakable in
$G$). The key observation is that then, whether
$\conn_k(s,t,u_1,\ldots,u_k)$ holds can be easily decided in time
polynomial in $q$ and $k$ as follows.

Run a breadth-first search from $s$ in $G-\{u_1,\ldots,u_k\}$, but
terminate the search once $q+1$ different vertices have been reached.
If the search reached $t$, then for sure $\conn_k(s,t,u_1,\ldots,u_k)$
holds. If the search finished before reaching $q+1$ different vertices
and it did not reach $t$, then for sure $\conn_k(s,t,u_1,\ldots,u_k)$
does not hold. In the remaining case, perform a symmetric procedure
starting from $t$. The key observation is that if both breadth-first
searches --- from $s$ and from $t$ --- got terminated after reaching
$q+1$ different vertices, then the $(q,k)$-unbreakability of $G$
implies that $\{u_1,\ldots,u_k\}$ do not separate $s$ from $t$. So
then $\conn_k(s,t,u_1,\ldots,u_k)$ holds. It is straightforward to
implement this procedure in time polynomial in $q$ and $k$ given the
standard encoding of $G$ through adjacency lists.

In the general case, we cannot expect the given graph $G$ to be
$(q,k)$-unbreakable for any parameter~$q$ bounded in terms of $k$.
However, we can use \cref{thm:ov:strong-unbreakability} to compute a
tree decomposition $\Tt$ of $G$ into parts that are
$(q,k)$-unbreakable for $q\in 2^{\Oh(k)}$. The idea is that then, a
query $\conn_k(s,t,u_1,\ldots,u_k)$ can be evaluated by bottom-up
dynamic programming on $\Tt$. By suitably precomputing enough
auxiliary information about $\Tt$, this evaluation can be performed in
time depending only on $q$ and $k$.

\paragraph*{Evaluating a query on the decomposition.} Consider then
the following setting: we are given the decomposition $\Tt=(T,\bag)$
provided by \cref{thm:ov:strong-unbreakability}, and for given
$s,t,u_1,\ldots,u_k$ we would like to decide whether
$\conn_k(s,t,u_1,\ldots,u_k)$ holds. So far let us not optimize the
complexity: the goal is only to design a general algorithmic
mechanism which will be later implemented efficiently using
appropriate data structures.


For every node $x$ of $T$, let $\adh(x)$ be the adhesion between $x$
and its parent (or $\emptyset$ if $x$ is the root), let $\cone(x)$ be
the union of the bags at the descendants of $x$, and let
$\cmp(x)=\cone(x)\setminus \adh(x)$. Call a node $x$ {\em{affected}}
if $\cmp(x)$ contains a vertex of $\{s,t,u_1,\ldots,u_k\}$, and
{\em{unaffected}} otherwise. Further, let
$D(x)=\adh(x)\cup (\cmp(x)\cap \{s,t\})$; note that $|D(x)|\leq q+2$.
Our goal is to compute the following {\em{(connectivity) profile}} for
every node $x$ of $T$:
$$\profile(x)=\left\{\{a,b\}\in\binom{D(x)}{2}\quad\Big|\quad a\textrm{ and }b\textrm{ are connected in }G[\cone(x)\setminus \{u_1,\ldots,u_k\}]\right\}.$$
We remark that the tree decomposition $\Tt$ can be chosen so that it
satisfies the following basic connectivity property: for every node
$x$ and vertices $a,b\in \adh(x)$, there is a path connecting~$a$
and~$b$ whose all internal vertices belong to $\cmp(x)$. Thus, for
every unaffected node $x$, we have that
$\profile(x)=\binom{D(x)\setminus \{u_1,\ldots,u_k\}}{2}$.

It is easy to see that the profile of $x$ is uniquely determined by
the profiles of the children of $x$ and the subgraph induced by the
bag of $x$. The following lemma shows that this can be done
efficiently.

\begin{lemma}[\cref{lem:aggregate}, informal statement]\label{lem:ov:aggregate}
  Let $x$ be a node of $T$. Suppose that for each affected child~$z$
  of $x$ we are given $\profile(z)$. Then, subject to suitable
  preprocessing of $G$, one can compute $\profile(x)$ in time
  polynomial in $q$. \end{lemma}

Note that every node has at most $k+2$ affected children, hence the
input to the algorithm of \cref{lem:ov:aggregate} is of size
polynomial in $q$.

Let us sketch the proof of \cref{lem:ov:aggregate}. In essence, we
apply the same strategy as the one from the case when the whole
graph $G$ is unbreakable. Let $\bgraph(x)$ be the
graph obtained from $G[\bag(x)]$ by turning the adhesion $\adh(z)$
into a clique, for every child $z$ of $x$. Note that $\bgraph(x)$ does
not depend on the query, and thus can be precomputed upon
initialization. As $\bag(x)$ is
$(q,k)$-unbreakable in $G[\cone(x)]$, to test whether two vertices
$a,b\in \bag(x)$ are connected in
$G[\cone(x)\setminus \{u_1,\ldots,u_k\}]$ it suffices to apply
breadth-first searches from $a$ and $b$ in $\bgraph(x)$ that are
terminated after reaching $q+1$ different vertices, except that these
searches are forbidden to use vertices from
$\{u_1,\ldots,u_k\}\cap \bag(x)$ as well as those edges originating
from turning adhesions of affected children into cliques that are not
in respective profiles. Provided $\bgraph(x)$ is precomputed and a
list of affected children together with their profiles is given on
input, it is easy to implement this algorithm so that it runs in time
polynomial in $q$. Thus we can decide for every pair
$\{a,b\}\in \binom{\adh(x)}{2}$ whether it should be included in
$\profile(x)$. A similar reasoning can be applied to pairs in
$\binom{D(x)}{2}$ that include $s$ or $t$.

By applying \cref{lem:ov:aggregate} bottom-up, we can compute all the
profiles in time $q^{\Oh(1)}\cdot |T|\leq q^{\Oh(1)}\cdot |G|$. Then
we can read whether $\conn_k(s,t,u_1,\ldots,u_k)$ holds by checking
whether $\{s,t\}$ belongs to the profile of the root of $T$. Hence,
the query $\conn_k(s,t,u_1,\ldots,u_k)$ can be evaluated in time
$q^{\Oh(1)}\cdot |G|$.

\newcommand{\hm}{\mathsf{top}}

\paragraph*{Data structure.} Our final goal is to enrich the
decomposition $\Tt$ with some additional information so that the
mechanism presented above can be executed in time depending only on
$q$.

For a vertex $v$ of $G$, let $\hm(v)$ be the unique top-most node
of~$T$ whose bag contains $v$. Let
$\Sh\coloneqq \{s,t,u_1,\ldots,u_k\}$ and let
$X=\{\hm(v)\colon v\in \Sh\}$. 

\begin{figure}[h]
  \centering
  \includegraphics[scale=0.5,page=9]{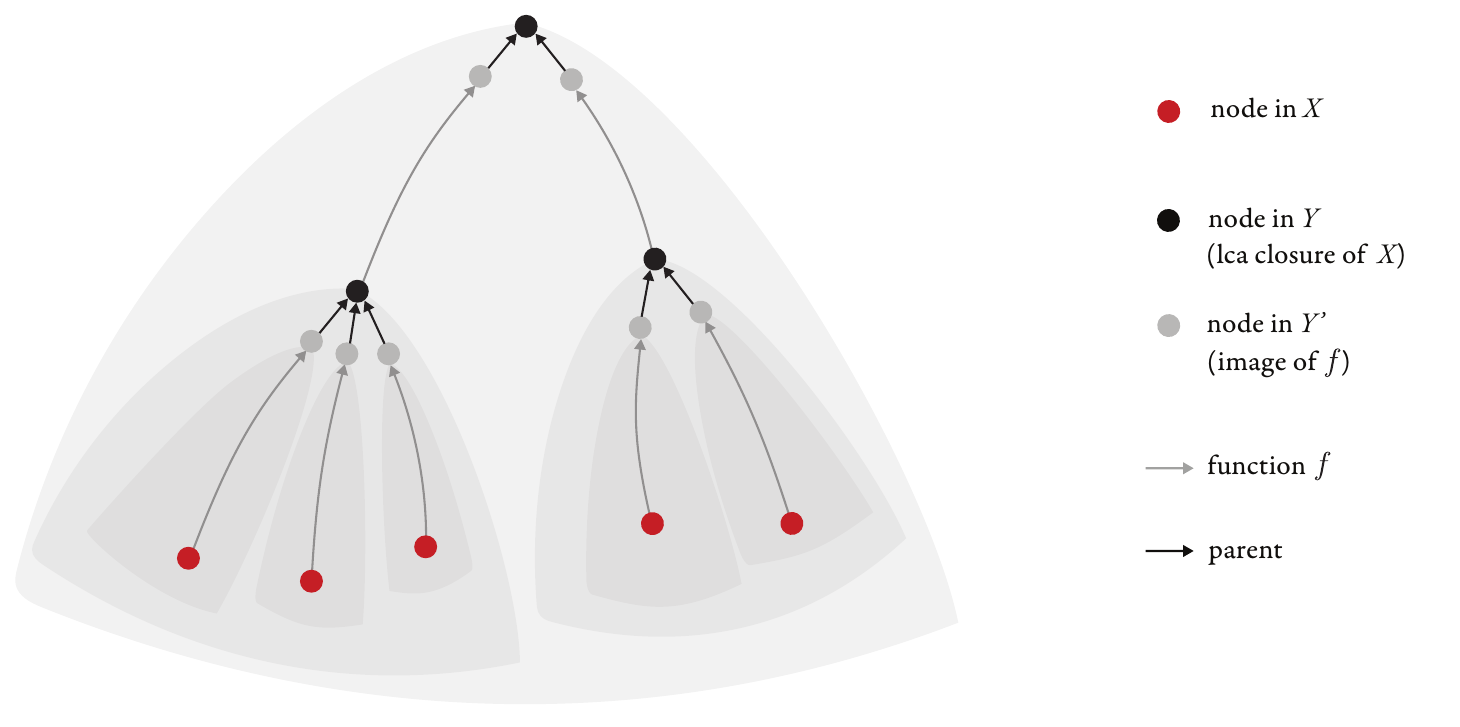}
  \caption{The sets $X,Y$ and $Y'$. Only the nodes from $Y\cup Y'$ are 
  marked in the figure, wheras 
  the shaded areas indicate subtrees consisting of 
  nodes which are not marked in the figure. Every affected node lies on a path from some $y\in Y$ to $f(y)\in Y'$.
  }
  \label{fig:lca-closure0}
 \end{figure}

Note that $|X|\leq |\Sh|\leq k+2$ and
that a node of~$T$ is affected if and only if it is an ancestor of a
node from $X$. Let now $Y$ be the {\em{lowest common ancestor
    closure}} of $X$: the set comprising $X$ and all lowest common
ancestors of pairs of vertices of~$X$. It is well-known that then
$|Y|\leq 2|X|-1\leq 2k+3$. We also add the root of $T$ to $Y$ in case
it is not already present; thus $|Y|\leq 2k+4$. Observe that if for
every vertex $v$ we store a pointer to $\hm(v)$, and we set up the
data structure for lowest common ancestor queries of Harel and
Tarjan~\cite{HarelT84} on $T$, then for given $s,t,u_1,\ldots,u_k$ the
set $Y$ can be computed in time polynomial in $k$.
For every node $y$ in $Y$ let $f(y)$ denote the closest ancestor of $y$
whose parent belongs to $Y$, or $y$ if no such ancestor exists
(see \cref{fig:lca-closure0}).
Note that every affected node lies on a path joining $y$ with $f(y)$,
for some $y\in Y$.
Let $Y'=\setof{f(y)}{y\in Y}$.
Then $Y'$ is the set of all nodes that are affected children of 
a node in $Y$. Note that every element in $Y'$ has exactly one 
preimage under $f$, that is, there is an function $f^{-1}\from Y'\to Y$
such that $f(f^{-1}(y'))=y'$ for all $y'\in Y'$.
In particular, $|Y'|\le |Y|$.
Using a
slight modification of the data structure of Harel and
Tarjan~\cite{HarelT84}, we can compute $Y'$ as well as the function $f^{-1}\from Y'\to Y$ in time polynomial in~$k$.

The idea is that when evaluating query $\conn_k(s,t,u_1,\ldots,u_k)$,
we can compute the profiles only for the nodes of $Y\cup Y'$, instead of all
affected nodes of $T$. Note that the root of $T$ was explicitly added to
$Y$, so we will still be able to read the answer to
$\conn_k(s,t,u_1,\ldots,u_k)$ from the profile of the root.

We process the nodes of $Y\cup Y'$ in a bottom-up manner. Consider then any element 
 $x\in Y\cup Y'$; our task is to compute its profile based on the profiles
of its strict descendants belonging to $Y\cup Y'$. 
We distinguish two cases. 
In the first case, $x\in Y$, so we may compute $\profile(x)$ using the algorithm of
\cref{lem:ov:aggregate}. 
The second case is of an element  $x\in Y'$; then $x=f(y)$ where $y=f^{-1}(x)$.
Observe that the path from $y$ to $x$ consists of affected nodes,
whose siblings are unaffected nodes. 
We now show how $\profile(x)$ can be
 computed from $\profile(y)$.

 For two nodes
$c,d$ of $T$, where $c$ is an ancestor of $d$, let $\torso(c,d)$ be
the set of all pairs $\{a,b\}\in \binom{\adh(c)\cup \adh(d)}{2}$ such
that in $G$ there is a path from $a$ to $b$ whose all internal
vertices belong to $\cmp(c)\setminus \cone(d)$. Then to compute
$\profile(x)$ from $\profile(y)$, construct an
auxiliary graph with vertices 
$(\adh(x)\cup \adh(y)$ and edges $\profile(y)\cup \torso(x,y)$.
Now
observe that $\profile(x)$ is the reachability relation in this graph
restricted to $\adh(x)\setminus \{u_1,\ldots,u_k\}$. The key point
implying the correctness of this observation is that there are no
vertices from $X$ in $\cmp(x)\setminus \cmp(y)$.

Observe now that the values of $\torso(c,d)$ for $(c,d)$ ranging over
ancestor/descendants pairs in $T$ are independent of the query.
Therefore, upon initialization we can compute a data structure that
can be queried for those valued efficiently. The statement below
describes two possible implementations.

\begin{lemma}[\cref{lem:torso-ds-trivial} and \cref{lem:torso-ds},
  combined and simplified]\label{lem:ov:torso-ds}
  Given $\Tt$, one can
  set up data structures that can answer $\torso(c,d)$ queries with
  the following specifications:
  \begin{enumerate}[nosep]
  \item initialization time $q^{\Oh(1)}\cdot |G|^2\|G\|$, memory usage
    $q^{\Oh(1)}\cdot |G|^2$, query time $q^{\Oh(1)}$; or
  \item initialization time $2^{\Oh(q^2)}\cdot |G|$, memory usage
    $2^{\Oh(q^2)}\cdot |G|$, query time $2^{\Oh(q^2)}$.
  \end{enumerate}
\end{lemma}

The first point of \cref{lem:ov:torso-ds} is actually straightforward:
just compute and store all the $\Oh(|T|^2)=\Oh(|G|^2)$ answers to the
torso queries, each in time $q^{\Oh(1)}\cdot \|G\|$ using $q^{\Oh(1)}$
applications of breadth-first search. The second point, which trades
exponential dependency on $q$ for obtaining linear memory usage and
initialization time, is more interesting. Before we discuss it, let us
observe that we may now complete the proof of
\cref{thm:general-queries}. Indeed, using the data structure from the
second point of \cref{lem:ov:torso-ds}, we can compute $\profile(x)$
from $\profile(y)$ for each $x\in Y'$ and $y=f^{-1}(x)$. By
performing this procedure for all the at most $2(2k+4)$ nodes of $Y\cup Y'$ in a
bottom-up manner, we eventually compute the profile of the root of
$T$, from which the answer to the query can be read. The running time
is dominated by $q^{\Oh(1)}$ calls to the data structure of the second
point of \cref{lem:ov:aggregate}, each taking
$2^{\Oh(q^2)}=2^{2^{\Oh(k)}}$ time.

Finally, let us discuss the proof of the second point of
\cref{lem:ov:aggregate}. We reduce this problem to the problem of
evaluating {\em{product queries}} in a semigroup-labeled tree,
defined as follows. Suppose $T$ is a rooted tree whose edges are
labeled with elements of a finite semigroup $S$. The task is to set
up a data structure that for given nodes $x,y$, where $x$ is an
ancestor of $y$, outputs the product of the elements of~$S$ on the
path in $T$ from $x$ to $y$. The problem of evaluating torso queries
can be reduced to this abstract setting by considering a suitable
semigroup of {\em{bi-interface graphs}}, which represent connectivity
between corresponding adhesions; see~\cite{BojanczykP16} for the
origin of this concept. This semigroup has size~$2^{\Oh(q^2)}$. So it
then suffices to use the following result.

\begin{theorem}[\cref{thm:colcombet-ds}, simplified]\label{thm:ov:colcombet-ds}
  There is a data structure for the problem of evaluating product
  queries in a tree $T$ labeled with elements of a finite semigroup
  $S$ that achieves query time $|S|^{\Oh(1)}$, memory usage
  $|S|^{\Oh(1)}\cdot |T|$, and initialization time
  $|S|^{\Oh(1)}\cdot |T|$.
\end{theorem}

\cref{thm:ov:colcombet-ds} can be easily derived from the
deterministic variant of Simon's factorization, due to
Colcombet~\cite{Colcombet07}. However, in \cref{sec:torso-queries} we
also present a direct and arguably simpler proof that may be of
independent interest.

\medskip

This finishes the proof of \cref{thm:general-queries}. Note that in
the argumentation there are two modules that can be replaced with
other solutions:
\begin{itemize}[nosep]
\item The decomposition of \cref{thm:ov:strong-unbreakability} can be
  replaced with the more recent variant from~\cite{CyganKLPPSW21}. This allows us to have $q=k$, implying that
  all the $\smash{2^{2^{\Oh(k)}}}$ factors are replaced with $\smash{2^{\Oh(k^2)}}$.
  The cost is increasing the polynomial factor of the initialization
  time to an unspecified term $\smash{|G|^{\Oh(1)}}$ and several technical
  complications in the analysis, which nevertheless can be overcome.
\item Instead of using the data structure of the second point of
  \cref{lem:ov:torso-ds}, we can use the first point. This allows us
  to have polynomial dependency on $q$ in the query time at the
  cost of quadratic dependence on $|G|$ in the memory usage.
\end{itemize}
In particular, if both these replacements are performed, we obtain a data structure with initialization time $2^{\Oh(k\log k)}\cdot |G|^{\Oh(1)}$, memory usage $k^{\Oh(1)}\cdot |G|^2$, and query time $k^{\Oh(1)}$. See \cref{thm:tradeoff}.

\subsection{Model-checking $\FOsep$:
\Cref{thm:main-ub,thm:main-lb}
}\label{sec:overview-mc}
We now turn to the proofs of \cref{thm:main-ub} and \cref{thm:main-lb}.
In those theorems, we fix a class $\CC$
and consider the model-checking problem for the logic $\FOsep$,
that is, the problem of deciding whether a given $\FOsep$ sentence holds in a given graph $G\in\CC$.

\paragraph*{Lower bound.}
First let us explain why we require $\CC$ to exclude a fixed topological minor
in \cref{thm:main-ub}.
Clearly, the model-checking problem for {\FOsep} generalizes the model-checking problem for {\FO}, so $\CC$ should be in particular a class for which it is known that
$\FO$ model-checking is $\FPT$.
This is the case not only for classes that exclude a topological minor but also for all nowhere dense classes~\cite{GroheKS17}. So let us see why it is not enough to merely assume that $\CC$ is  nowhere dense.

For a graph $G$ and $k\in \N$, the {\em $k$-subdivision} $G^{(k)}$ of $G$ is obtained from $G$ by replacing each edge by a path of length $k+1$.
It follows from the definition of nowhere denseness that the class $\CC$ of all graphs of the form $G^{(n)}$, where $G$ is an arbitrary graph on $n$ vertices, is nowhere dense.
However, a formula of $\FOsep$ can easily recover $G$ from $G^{(n)}$, at least when $G$ has no vertices of degree $2$. Indeed, the original vertices of $G$ are precisely the vertices of $G^{(n)}$ that have degree at least $3$, which can be expressed by an $\FO$ formula $\phi_V(x)$. Furthermore, two such vertices $u,v$ are adjacent in $G$ if and only if there is some vertex $z$ of degree $2$ in $G^{(n)}$ such that
for every vertex $w$ of $G^{(n)}$, $\conn(z,w,u,v)$ holds only if $w$ has degree $2$ in $G^{(n)}$. This condition can be written using an $\FOsep$ formula $\phi_E(u,v)$.
Using this observation, any $\FO$ sentence $\phi$ can be replaced by an $\FOsep$ sentence $\phi'$, so that for every graph $G$ as above, $\phi$ holds in $G$ if and only if $\phi'$ holds in $G^{(n)}$. Essentially, $\phi'$ is obtained from $\phi$ by replacing every atomic formula $E(x,y)$ (denoting adjacency) by $\phi_E(x,y)$, and guarding all quantifiers by $\phi_V(x)$.
Therefore, if we could efficiently model-check $\phi'$ on the nowhere dense class $\CC$, we could as well model-check $\phi$ on the class of all graphs (of size at least $3$ and minimum degree at least $3$), as follows:
given an arbitrary graph $G$, first construct the graph $G^{(n)}\in \CC$ (in time polynomial in $n=|V(G)|$)
and then test $\phi'$ on $G^{(n)}$. Hence, if $\FOsep$ model-checking was fpt on $\CC$, then $\FO$ model-checking would be fpt on the class of all graphs, implying $\FPT=\AWstar$. So nowhere dense graph classes are too general for the statement of \cref{thm:main-ub} to hold. Intuitively, the notion of nowhere denseness is not preserved under contracting long paths, which can be simulated in the logic \FOsep.

The same simple argument proves \cref{thm:main-lb}.
Indeed, the argument works not only for the graph~$G^{(n)}$ obtained as the $n$-subdivision of $G$ but also for any subdivision $G'$ (obtained by replacing each edge independently by any number of vertices) of $G$,
as long as $G'$ can be computed from $G$ in polynomial time. The condition that $\CC$ is subgraph-closed and admits efficient encoding of topological minors guarantees precisely that for any given graph $G$ we can construct, in polynomial time, a graph~$G'$ which is a subdivision of $G$ and belongs to $\CC$. This yields \cref{thm:main-lb}
(see \cref{sec:lbs} for details).

\paragraph*{Upper bound.}
We now give some details concerning the proof of \cref{thm:main-ub}, which is the second main contribution of the paper.

The starting observation is that for a
$(q,k)$-unbreakable graph $G$, the query $\conn_k(u,v,x_1,\ldots,x_k)$ can be expressed in plain \FO. Indeed, this query fails if and only if there is a set $A$ of at most $q$ vertices which contains exactly one of the vertices $u$ and $v$, and such that all neighbors of vertices of $A$ outside of $A$ are contained in $\set{x_1,\ldots,x_k}$. The existence of such a set of $q$ vertices can be expressed using $q$ existential quantifiers followed by a universal quantifier.
So $\FOsep$ with connectivity queries of arity at most $k+2$
 is equally expressive as plain $\FO$ on $(q,k)$-unbreakable graphs, as long as $q$ is a constant.

Now, if the graph $G$ is not $(q,k)$-unbreakable, we work with the
tree decomposition into $(q,k)$-unbreakable parts given by \cref{thm:strong-unbreakability}, where $q=q(k)$ is a constant depending on $k$.
The next observation is that the constant-time querying algorithm used in the proof of \cref{thm:general-queries} and explained above, can be seen as a formula on a suitably defined logical structure. Indeed,
if we look into the data structure answering the queries $\conn_k(u,v,x_1,\ldots,x_k)$, we notice that each such query triggers a constant number
of operations asking about the least common ancestor of two nodes in the tree decomposition, or asking about the membership of a vertex to a bag.
Additionally, the algorithm may ask, for two vertices $u,v$ belonging to a bag $x$,
whether $u$ and $v$ are adjacent in the original graph, or whether they are adjacent in the bag graph $\bgraph(x)$. Finally, the algorithm also performs torso queries.
We may therefore construct a logical structure, corresponding to the data structure, such that this basic functionality
is expressible by a first-order formula in this structure. Then the query
$\conn_k(u,v,x_1,\ldots,x_k)$ becomes expressible using a plain {\FO} formula in the logical structure. It will then remain to show that we can model check {\FO} formulas in fpt on the resulting logical structure.


It is convenient to abstract away the details of our logical structure, and represent it in the following form:
It is a tree $T$ (possibly vertex-labeled), together with additional edges (possibly labeled) between some nodes which are siblings in the tree; see \cref{fig:augmented}. We call such a structure an \emph{augmented tree}.
In such a tree, the set of children of any node $v$ carries a structure of a graph.
In the particular case of our construction, those graphs will precisely correspond to bag graphs.
Recall that a bag graph is obtained from the subgraph induced by a bag by turning all adhesions towards children into cliques. It is not difficult to show (see \cref{lem:bag graph}) that if the original graph $G$ is $H$-topological-minor free, then the resulting bag graphs
are $H'$-topological-minor-free for some $H'$ depending on $H$ and on $k$.
Hence, we obtain a tree augmented with $H'$-topological-minor-free graphs.
And the original query $\conn_k(u,v,x_1,\ldots,x_k)$ can be now represented as an
{\FO} formula in the augmented tree, where the {\FO} formula may use the ancestor relation~$\le$ on the tree nodes, as well as the adjacency relation between the siblings.
Showing this basically amounts to revisiting the proof of \cref{lem:ov:aggregate} and arguing that the atomic algorithmic operations can now be simulated by {\FO} formulas in the augmented tree.

Let us note here that the remainder of the argument would go through even if we assumed that the resulting tree is augmented with graphs from a fixed nowhere dense class $\CC'$
which is not necessarily topological-minor-free.
However, if the original graph $G$ merely belongs to a nowhere dense class $\CC$,
then the bag graphs, and hence the graphs in the augmented tree, no longer need to belong to any fixed nowhere dense class. In fact, they may contain arbitrarily large cliques.

Continuing with the proof, now it remains to show that we can solve {\FO} model-checking on trees augmented with $H'$-topological-minor-free graphs in fixed-parameter time,
where the {\FO} formulas may use the ancestor relation in the tree as well as the edge relation among the siblings.
 For this, we lift the usual, automata-based techniques for model-checking on trees, to the case of augmented trees. We define automata that process augmented trees in a bottom-up fashion, labeling the nodes of the tree with states from a finite set of states, starting from the leaves and proceeding towards the root.
 At any node~$x$, to determine the state of the automaton at $x$, we consider the graph induced on the children of $x$ in the augmented tree, vertex-labeled by the states computed earlier. Then the state at $x$ is determined by evaluating a fixed collection of first-order sentences on this labeled graph. In this way, all nodes of the tree are labeled by states, and the automaton accepts the augmented tree depending on the state at the root.

 By construction, it can be decided in fpt time whether or not
 a given such automaton accepts a given  tree that is augmented with graphs
 that come from, say, a nowhere dense class. Additionally, using standard techniques,
 we show that for every first-order sentence $\phi$ on augmented trees there is such an automaton that determines whether $\phi$ holds in a given augmented tree.
 Therefore, {\FO} model-checking is fpt on trees augmented with graphs from a nowhere dense class; this in particular applies to classes with excluded topological minors.

 To summarize, to  model-check a sentence $\phi$ of {\FOsep} on a $H$-topological-minor-free graph $G$, we perform the following reasoning:
 \begin{enumerate}[nosep]
   \item Using \cref{thm:ov:strong-unbreakability}, compute a tree decomposition $\cal T$ of $G$ with $(q,k)$-unbreakable bags, where $k+2$ is the maximal arity of the connectivity predicates in $\phi$, and $q=q(k)$.
   \item Turn 
   $\cal T$ into a tree $\cal T'$ augmented with $H'$-topological-minor free graphs, where $H'$ depends on $H$ and $k$.
   \item In the augmented tree $\cal T'$, the connectivity predicate $\conn_k(u,v,x_1,\ldots,x_k)$ can be expressed using a plain {\FO} formula. This follows by analyzing how  connectivity predicates are evaluated in constant time, and observing that the atomic operations can be performed using {\FO} formulas in the augmented trees.
   \item Therefore, also $\phi$ can be expressed as a plain {\FO} formula $\phi'$, which is evaluated in 
   $\cal T'$.
   \item Convert $\phi'$ into an automaton $\cal A$ processing augmented trees,
   using standard automata-based techniques.
   \item
   Run $\cal A$ on $\cal T'$ in a bottom-up fashion. This involves performing {\FO} queries on the graphs induced on the children of any given node, and uses one of the known algorithms for {\FO} model-checking on $H'$-topological-minor-free graphs.
 \end{enumerate}

\cref{sec:automata} describes the basic toolbox of automata on augmented trees,
explaining in detail the last two steps above, abstracting away from the specific application, and providing further results on augmented trees, e.g. concerning query enumeration.
In fact, with the same methods we can show that not only $\FO$ model-checking is fpt on trees augmented with graphs from a nowhere dense class,
but also a more general logic $\FOMSO$ combining the power of $\FO$ on the augmenting graphs with $\MSO$ on trees can be solved in fpt on augmented trees.
We then show in \cref{sec:model-checking} how to express the predicates $\conn_k$ on augmented trees using this more general logic $\FOMSO$.

%% file: preliminaries.tex

\section{Preliminaries}\label{sec:prelims}

Throughout this paper we assume the standard word RAM model with
machine words of length $\Oh(\log n)$, where $n$ is the total size of
the input. Recall that in this model we allow using arrays indexed by
integers in the range $\{1,\ldots,n^{\Oh(1)}\}$, which is sufficient
for the existence of data structures for least common ancestor queries
in trees in constant time~\cite{HarelT84}.

\subsection{Graphs}

\paragraph*{Basic notation.} We consider finite, undirected graphs
without loops. For a graph $G$, we write $|G|$ for~$|V(G)|$ and
$\|G\|$ for $|V(G)|+|E(G)|$. For a pair of vertices $u,v\in V(G)$, a
{\em{$u$-$v$ path}} is a path with endpoints~$u$ and~$v$. We say that
$u$ and $v$ are {\em{connected}} in $G$ if there is a $u$-$v$ path in
$G$; otherwise~$u$ and~$v$ are {\em{disconnected}}. The graph $G$ is
connected if all pairs of its vertices are connected. The
{\em{interior}} of a path consists of all its vertices apart from the
endpoints, and two paths are {\em{internally vertex-disjoint}} if each
of them is disjoint with the interior of the other. For a vertex
subset $X\subseteq V(G)$ we write $G[X]$ for the subgraph of $G$
induced by $X$.

\paragraph*{Trees and ancestors.}
A (rooted) tree is an acyclic and connected graph $T$ with a
distinguished root vertex~$r$. We write $\parent(x)$ for the parent of
a node $x$ of $T$, and $\children(x)$ is the set of children of $x$
in~$T$. By convention we set $\parent(r)=\bot$. We call a vertex
$x\in V(T)$ an ancestor of a vertex $y\in V(T)$, written
$x\preceq_T y$, or simply~$x\preceq y$ if $T$ is clear from the
context, if $x$ lies on the unique path between $y$ and the root $r$.
Note that we consider every node to be an ancestor of itself.  For
nodes $x,y\in V(T)$, we write $\lca(x,y)$ for the least common
ancestor of $x$ and $y$ in $T$. Note that $x$ is an ancestor of $y$ if
and only if $\lca(x,y)=x$. If $x$ is a strict ancestor of $y$ (that
is, $x\preceq y$ and $x\neq y$), then we define the {\em{$y$-directed
    child of $x$}} to be the unique child $y'$ of $x$ that is an
ancestor of $y$; we denote it by~$\dir(x,y)$.

\paragraph*{Topological minors.} A {\em{topological minor model}} of a
graph $H$ in a graph $G$ is an injective mapping $\eta$ that maps
vertices of $H$ to vertices of $G$ and edges of $H$ to pairwise
internally vertex-disjoint paths in $G$ so that for every
$uv\in E(H)$, the endpoints of $\eta(uv)$ are $\eta(u)$ and
$\eta(v)$. We say that $G$ is {\em{$H$-topological-minor-free}} if
there is no topological minor model of $H$ in $G$. A class of graphs
$\Cc$ is {\em{topological-minor-free}} if there exists a graph $H$
such that every member of $\Cc$ is $H$-topological-minor-free.

\paragraph{Tree decompositions.} A {\em{tree decomposition}} of a
graph $G$ is a pair $\Tt=(T,\bag)$, where $T$ is a rooted tree and
$\bag\colon V(T)\to 2^{V(G)}$ is a mapping that assigns to each node
$x$ of $T$ its {\em{bag}} $\bag(x)$, which is a subset of vertices of
$G$. The following conditions have to be satisfied:
\begin{enumerate}[(T1)]
\item For every $u\in V(G)$, the set of nodes $x\in V(T)$ satisfying
  $u\in \bag(x)$ induces a connected and nonempty subtree of $T$.
\item For every edge $uv\in E(G)$, there exists a node $x\in V(T)$
  such that $\{u,v\}\subseteq \bag(x)$.
\end{enumerate}

Recall that if $r$ is the root of $T$, then by convention
$\parent(r)=\bot$ and $\bag(\bot)=\emptyset$.

With a tree decomposition $\Tt=(T,\bag)$ we associate a few auxiliary
mappings. For a node $x\in V(T)$, we define the following:
\begin{itemize}
\item The {\em{adhesion}} of $x$ is
 $$\adh(x)\coloneqq \bag(\parent(x))\cap \bag(x).$$
 \item The {\em{margin}} of $x$ is
 $$\mrg(x)\coloneqq \bag(x)\setminus \adh(x).$$
\item The {\em{cone at $x$}} is
 $$\cone(x)\coloneqq \bigcup_{y\succeq_T x} \bag(y).$$
\item The {\em{component at $x$}} is
 $$\cmp(x)\coloneqq \cone(x)\setminus \adh(x)= \bigcup_{y\succeq_T x} \mrg(y).$$
\end{itemize}

It is easy to see that the margins $\{\mrg(x)\colon x\in V(T)\}$ are
pairwise disjoint and cover the whole vertex set of $G$. Whenever we
use the notation introduced above, the tree decomposition $\Tt$ will
be always clear from the context.

The {\em{adhesion}} of a tree decomposition $\Tt=(T,\bag)$ is defined
as the largest size of an adhesion, that is,
$\max_{x\in V(T)}|\adh(x)|$.

We will commonly assume (without loss of generality) that the tree
decompositions that we work with satisfy some basic connectivity
conditions, subsumed in the following definition. A tree decomposition
$\Tt=(T,\bag)$ is {\em{regular}} if for every non-root node
$x\in V(T)$ we have that
\begin{enumerate}[(R1)]
\item\label{reg:mrg} the margin $\mrg(x)$ is nonempty;
\item\label{reg:con} the graph $G[\cmp(x)]$ is connected; and
\item\label{reg:nei} every vertex of $\adh(x)$ has a neighbor in $\cmp(x)$.
\end{enumerate}

Note that the last two conditions imply that for every node $x$ and
every pair of distinct vertices $u,v\in \adh(x)$ there is a $u$-$v$
path in $G[\cone(x)]$ whose interior is contained in $\cmp(x)$.

\paragraph*{Unbreakability.} A {\em{separation}} in a graph $G$ is a
pair of vertex subsets $A,B\subseteq V(G)$ such that $A\cup B=V(G)$
and there is no edge with one endpoint in $A\setminus B$ and the other
in $B\setminus A$. The {\em{order}} of the separation $(A,B)$ is the
size of its {\em{separator}} $A\cap B$.

For $q,k\in \N$, a vertex subset $X$ in a graph $G$ is
{\em{$(q,k)$-unbreakable}} if for every separation $(A,B)$ of~$G$ of
order at most $k$, we have
\[|A\cap X|\leq q\qquad\textrm{or}\qquad |B\cap X|\leq q.\]

The notion of unbreakability can be lifted to tree decompositions by
requiring it from every individual bag. There are, however, two
different ways of understanding this.

\begin{definition}
  Fix $q,k\in \N$.  A tree decomposition $\Tt=(T,\bag)$ of a graph $G$
  is
\begin{itemize}
\item {\em{weakly $(q,k)$-unbreakable}} if for every $x\in V(T)$,
  $\bag(x)$ is $(q,k)$-unbreakable in $G$; and
\item {\em{strongly $(q,k)$-unbreakable}} if for every $x\in V(T)$,
  $\bag(x)$ is $(q,k)$-unbreakable in $G[\cone(x)]$.
\end{itemize}
\end{definition}

It is straightforward to see that strong $(q,k)$-unbreakability
implies weak $(q,k)$-unbreakability, but the converse is not true.

We recall known results about the existence and computability of
unbreakable tree decompositions (\cref{thm:strong-unbreakability} was already stated as 

\begin{theorem}[\cite{CyganKLPPSW21}]\label{thm:weak-unbreakability}
  For every graph $G$ and $k\in \N$ there exists a weakly
  $(k,k)$-unbreakable tree decomposition of $G$ of adhesion at most
  $k$. Moreover, given $G$ and~$k$, such a tree decomposition can be
  computed in time $2^{\Oh(k\log k)}\cdot |G|^{\Oh(1)}$.
\end{theorem}

\begin{theorem}[\cite{CyganLPPS19}]\label{thm:strong-unbreakability}
  There is a function $q(k)\in 2^{\Oh(k)}$ such that for every graph
  $G$ and $k\in \N$ there exists a strongly $(q(k),k)$-unbreakable
  tree decomposition of $G$ of adhesion at most $q(k)$.  Moreover,
  given~$G$ and~$k$, such a tree decomposition can be computed in time
  $2^{\Oh(k^2)}\cdot |G|^2\cdot \|G\|$.
\end{theorem}

In this paper we will mostly use \cref{thm:strong-unbreakability}, as
working with strong unbreakability yields a cleaner presentation of
arguments. However, in \cref{sec:tradeoff} we show that weak
unbreakability can be used in the context of
\cref{thm:general-queries}, yielding a variant of the theorem that
offers better time complexity for queries.

Given any weakly or strongly $(q,k)$-unbreakable decomposition we can,
if necessary, refine it so that it becomes regular.  Hence, we may
assume that the tree decompositions constructed by the algorithms of
\cref{thm:weak-unbreakability} and \cref{thm:strong-unbreakability}
are regular. More precisely, this refinement can be achieved by
applying a folklore refinement procedure for instance presented in the
proof of~\cite[Lemma~2.8]{BojanczykP16}, which takes linear
time. This refinement does not increase the sizes of adhesions nor
spoils unbreakability, because every new bag is a subset of an old
bag. Note here that~\cite{BojanczykP16} allows tree decompositions to
be rooted forests rather than rooted trees. However, a rooted forest
can be turned into a rooted tree by adding a fresh root with an empty
bag, whose children are all the former roots. Note here that
properties mentioned in the definition of regularity are not required
from this new root.

%

\paragraph*{Bag graphs.}
Let $\Tt=(T,\bag)$ be a regular tree decomposition of a graph $G$.
With a node $x$ of $G$ we associate the {\em{bag graph}} $\bgraph(x)$,
which is the graph obtained from $G[\cone(x)]$ by contracting, for
each $y\in \children(x)$, the set $\cmp(y)$ into a single vertex,
identified with $y$. More precisely, $\bgraph(x)$ is a graph defined
as follows:
\begin{itemize}
\item The vertex set of $\bgraph(x)$ is the disjoint union of
  $\bag(x)$ and the set of children of $x$ in $T$.
\item For $u,v\in \bag(x)$, $u$ and $v$ are adjacent in $\bgraph(x)$
  if and only if they are adjacent in $G$.
\item For $u\in \bag(x)$ and $y\in \children(x)$, $u$ and $y$ are
  adjacent in $\bgraph(x)$ if and only if there exists $w\in \cmp(y)$
  that is adjacent to $u$ in $G$.
 \item No two distinct $y,y'\in \children(x)$ are adjacent in $\bgraph(x)$.
\end{itemize}
Since $\Tt$ is regular, for every $y\in \children(x)$, the
neighborhood of $y$ in $\bgraph(x)$ is equal to $\adh(y)$.

We observe that the sizes of all bag graphs roughly sum up to the size
of the graph $G$.

\begin{lemma}\label{lem:total-size}
  Let $\Tt=(T,\bag)$ be a regular tree decomposition of a graph $G$,
  say of adhesion $a$. Then
 $$\sum_{x\in V(T)} |\bgraph(x)|\leq (a+2)|G|\qquad\textrm{and}\qquad \sum_{x\in V(T)} \|\bgraph(x)\|\leq \|G\|+(a+2)^2|G|.$$
\end{lemma}
\begin{proof}
  For the bound on $\sum_{x\in V(T)} |\bgraph(x)|$, observe the vertex
  set of $\bgraph(x)$ is the disjoint union of $\adh(x)$, $\mrg(x)$,
  and $\children(x)$. Every adhesion has size at most $a$, every node
  of $T$ has at most one parent, and every vertex of $G$ is in exactly
  one margin. Finally, $|T|\leq |G|$, hence
 $$\sum_{x\in V(T)} |\bgraph(x)|\leq a|T|+|G|+|T|\leq (a+2)|G|.$$

 For the bound on $\sum_{x\in V(T)} \|\bgraph(x)\|$, call an edge
 $uv\in E(\bgraph(x))$ {\em{vital}} in the graph $\bgraph(x)$ if
 $u,v\in \mrg(x)$.  Note that every edge of $G$ is vital in at most
 one bag graph $\bgraph(x)$, hence the sum over all $x\in V(T)$ of the
 number of vital edges in $\bgraph(x)$ is at most $|E(G)|$. On the
 other hand, the non-vital edges of $\bgraph(x)$ are those incident to
 the vertices of $\adh(x)$ and those incident to the children of
 $x$. The total number of non-vital edges of the first kind is at most
 $a\cdot |\bgraph(x)|$, as $|\adh(x)|\leq a$, and the total number of
 edges of the second kind is at most $a\cdot |\children(x)|$, as every
 child of $x$ has degree at most $a$ in $\bgraph(x)$. We conclude that
 \begin{eqnarray*}
   \sum_{x\in V(T)} \|\bgraph(x)\| & \leq & \sum_{x\in V(T)} |\bgraph(x)|+|E(G)|+a\cdot \sum_{x\in V(T)} \left(|\bgraph(x)|+|\children(x)|\right) \\
                                   & \leq & (a+2)|G|+|E(G)|+a(a+2)|G| +a|T|\\
                                   & \leq & \|G\|+(a+2)^2|G|.
 \end{eqnarray*}
 This concludes the proof.
\end{proof}

We observe that topological-minor-freeness of $G$ carries over to bag
graphs.

\begin{lemma}\label{lem:bag graph}
  Let $\Tt=(T,\bag)$ be a regular tree decomposition of a graph $G$,
  say of adhesion $a$. Suppose further that $G$ is
  $K_t$-topological-minor-free. Then for each $x\in V(T)$, the graph
  $\bgraph(x)$ is \mbox{$K_{t'}$-topological}-minor-free, where
  $t'=\max(t,a+2)$.
\end{lemma}
\begin{proof}
  Suppose there is a topological minor model $\eta$ of $K_{t'}$ in
  $\bgraph(x)$. Observe that no child of $x$ can be in the image of
  the vertex set of $K_{t'}$ under $\eta$, for these vertices have
  degree at most $a$ in $\bgraph(x)$ while $t'\geq a+2$. Therefore,
  the only way any $y\in \children(x)$ can participate in the image of
  $\eta$ is that~$y$ is an internal vertex of a path $\eta(e)$, for
  some edge $e$ of $K_{t'}$.
  We construct a topological minor model $\eta'$ of $K_{t'}$ in
  $G[\cone(x)]$ as follows:

 \begin{itemize}
 \item For each $u\in V(K_{t'})$, set $\eta'(u)=\eta(u)$.
 \item For each $e\in E(K_{t'})$, construct $\eta'(e)$ from $\eta(e)$
   by replacing, for each $y\in \children(x)$, the subpath $u-y-v$,
   where $u,v\in \bag(x)$ are the neighbors of $y$ on $\eta(e)$, with
   any $u$-$v$ path in $G[\cone(y)]$ whose interior is entirely
   contained in $\cmp(y)$. Such a path exists by the assumption that
   $\Tt$ is regular.
 \end{itemize}

 Since the paths in $\{\eta(e)\colon e\in E(K_{t'})\}$ are pairwise
 internally vertex-disjoint and $\cmp(y)$ and $\cmp(y')$ are disjoint
 for distinct $y,y'\in \children(x)$, it follows that the paths in
 $\{\eta'(e)\colon e\in E(K_{t'})\}$ are also pairwise internally
 vertex-disjoint. Thus, $\eta'$ is a topological minor model of
 $K_{t'}$ in $G[\cone(x)]$, hence also in $G$; a contradiction.
\end{proof}

Finally, we need an observation about how strong unbreakability of a
tree decomposition influences unbreakability properties of bag graphs.
For this, it will be convenient to work with another type of a bag
graph, defined as follows. Suppose $\Tt=(T,\bag)$ is a regular tree
decomposition of a graph $G$, $x$ is a node of $T$, and
$S\subseteq \cone(x)$. The {\em{$S$-restricted bag graph}} of $x$,
denoted $\bgraph_S(x)$, is the graph obtained from $\bgraph(x)$ by the
following operations:
\begin{itemize}
 \item Remove every vertex of $S\cap \bag(x)$.
 \item For every $y\in \children(x)$, remove $y$ and for every pair of
   vertices $\{u,v\}\in \binom{\adh(y)\setminus S}{2}$, add an edge
   $uv$ provided it was not already there and in $G[\cone(y)]$ there
   is a $u$-$v$ path that is disjoint with $S$.
\end{itemize}

For example, as $\Tt$ is regular, $\bgraph_\emptyset(x)$ is the graph
obtained from $\bgraph(x)$ by removing all $y\in \children(x)$ and,
for each of them, turning $\adh(y)$ into a clique.

\begin{lemma}\label{lem:short-BFS-combinatorial}
  Let $\Tt=(T,\bag)$ be a regular and strongly $(q,k)$-unbreakable
  tree decomposition of a graph~$G$. Let $x\in V(T)$ and
  $S\subseteq \cone(x)$ be such that $|S|\leq k$. Then for any
  $u,v\in \bag(x)\setminus S$, we have the following.
  \begin{itemize}
  \item Vertices $u$ and $v$ are connected in $G[\cone(x)]-S$ if and
    only if they are connected in $\bgraph_S(x)$.
  \item If $u$ and $v$ are disconnected in $\bgraph_S(x)$, and
    $C_u,C_v$ are the connected components of $\bgraph_S(x)$
    containing $u$ and $v$, respectively, then either $|V(C_u)|\leq q$
    or $|V(C_v)|\leq q$.
  \end{itemize}
\end{lemma}
\begin{proof}
  The first assertion follows directly from the construction, so we
  focus on the second one.

  Let $L$ consist of all vertices reachable in $G[\cone(x)]-S$ from
  $u$, and let $R=\cone(x)\setminus (L\cup S)$. Thus,
  $(L\cup S,R\cup S)$ is a separation in $G[\cone(x)]$ of order
  $|S|\leq k$. From the first assertion it follows that
  $V(C_u)\subseteq L$ and $V(C_v)\subseteq R$. As
  $V(C_u),V(C_v)\subseteq \bag(x)$, from the strong
  $(q,k)$-unbreakability of~$\Tt$ we conclude that either
  $|V(C_u)|\leq q$ or $|V(C_v)|\leq q$.
\end{proof}

\paragraph{Connectivity queries.}
Suppose $G$ is a graph, $k$ an integer, $s$ and $t$ are vertices of
$G$, and $S$ is a set of $k$ vertices $u_1,\ldots,u_k$ of $G$.  We say
that $\conn_k(s,t,u_1,\ldots,u_k)$ holds in $G$ if and only if in $G$
there is an $s$-$t$ path that does not pass through any vertex of
$S$. Note that in particular, $\conn_k(s,t,u_1,\ldots,u_k)$ does not
hold if $\{s,t\}\cap S\neq \emptyset$. We may also write
$\conn(s,t,S)$.


\subsection{Logic}

\paragraph{Signatures.}
We will only consider finite relational signatures $\Sigma$ consisting
of unary (arity~$1$) and binary (arity~$2$) relation symbols.  A
signature consisting only of unary relation symbols will be called an
alphabet. A $\Sigma$-structure $\str A$ consists of a finite universe
$V(\str A)$ and an interpretation $R(\str A)\subseteq V(\str A)^m$ of
each $m$-ary relation symbol $R\in\Sigma$. When there is no ambiguity,
we will not distinguish between relation symbols and their
interpretations.

Graphs are represented as $\Sigma$-structures where the universe is
the vertex set and $\Sigma$ consists of one binary relation symbol
$E(\cdot,\cdot)$, interpreted as the edge relation.


\paragraph{First-order logic (\FO).}
We assume an infinite supply $\textsc{Var}_1$ of first-order variables
and an infinite supply $\textsc{Var}_2$ of second-order set
variables. First-order variables will always be denoted by small
letters $x,y,z,\ldots$, while second-order variables will be denoted
by capital letters $X,Y,Z,\ldots$.

For a fixed signature $\Sigma$, formulas of first-order logic ($\FO$)
are constructed from atomic formulas of the form $x=y$, where
$x,y\in \textsc{Var}_1$, and $R(x_1,\ldots, x_m)$, where $R\in \Sigma$
is an $m$-ary relation symbol and $x_1,\ldots, x_m\in \textsc{Var}_1$,
by applying the Boolean operators~$\neg$,~$\wedge$~and~$\vee$, and
existential and universal quantification~$\exists x$ and $\forall x$.

A variable~$x$ not in the scope of a quantifier is a free variable (we
will not consider formulas with free set variables). A formula without
free variables is a sentence.
We write $\phi(\tup x)$ to indicate that the free variables of a
formula $\phi$ are contained in the set of variables~$\tup x$. A
valuation of $\tup x$ in a set $A$ is a function
$\tup a\from\tup x\to A$. Let $A^{\tup x}$ denote the set of all
valuations of $\tup x$ in $A$.

The satisfaction relation between $\Sigma$-structures and formulas is
defined as usual by structural induction on the formula. When $\str A$
is a $\Sigma$-structure, $\phi(\bar x)$ is a formula with free
variables contained in~$\bar x$, and $\bar a\in V(\str A)^{\bar x}$ is
a valuation of $\bar x$, we write
$(\str A,\bar a)\models \phi(\bar x)$ or $\str A\models\phi(\bar a)$
to denote that $\phi$ holds in~$\str A$ when the variables are
evaluated by $\bar a$. We let
$\phi(\str A):=\{\tup a\in A^{\tup x}~|~\str A\models \varphi(\tup
a)\}$.

\paragraph{Separator logic (\FOsep).}
Assume that $\Sigma$ contains a distinguished binary relation symbol
$E(\cdot,\cdot)$ that will always be interpreted as the edge relation
of a graph, that is, as an irreflexive and symmetric relation. For a
$\Sigma$-structure $\str A$, write $G(\str A)$ for the graph
$(V(\str A), E(\str A))$.  \emph{Separator logic} ({\FOsep}) is
first-order logic extended by the predicates
$\conn_k(x,y,z_1,\ldots,z_k)$, for all $k\ge 0$, where
${x,y,z_1,\ldots,z_k\in \textsc{Var}_1}$. The satisfaction relation
between $\Sigma$-structures and $\FOsep$ formulas is as for
first-order logic, where the atomic formula
$\conn_k(x,y,z_1,\ldots,z_k)$ holds in a $\Sigma$-structure $\str A$
with a valuation $\tup a$ of the variables $x,y,z_1,\ldots,z_k$ to
elements $s,t,u_1,\ldots,u_k\in V(\str A)$, if: \\
${G(\str A)\models \conn_k(s,t,u_1,\ldots,u_k)}$.

\paragraph{Monadic second-order logic (\MSO).}

Monadic second-order logic ($\MSO$) extends $\FO$ by existential and
universal quantification~$\exists X$ and~$\forall Y$ over second-order
set variables. We will consider $\MSO$ formulas over the signature
$A\cup \{\preceq\}$, where $A$ is a finite alphabet and $\preceq$ is a
binary relation symbol that will always be interpreted as the ancestor
relation in a tree. By $\MSO(\preceq,A)$ we denote the set of $\MSO$
formulas over such signatures.
Hence, formulas of $\MSO(\preceq,A)$ will only be evaluated over
rooted trees (represented by the ancestor relation $\preceq$)
decorated with predicates from $A$, and the semantics are defined as
usual by structural induction on the formula.

\paragraph{First-order logic with {\MSO} atoms.} We will also consider
{\FO} formulas with restricted {\MSO} atoms.  Fix a signature $\Sigma$
and an alphabet $A$.  The logic $\FO(\MSO(\preccurlyeq, A), \Sigma)$
denotes first-order logic over signature $\Sigma$ where in addition to
the standard syntax, one can use formulas of $\MSO(\preceq,A)$ as
atomic formulas. Semantics are defined as expected.

%% file: queries.tex

\section[Connectivity queries in general graphs: proof of Theorem 1.3]{Connectivity queries in general graphs: proof of \cref{thm:general-queries}}\label{sec:general-queries}

In this section we prove \cref{thm:general-queries}.
The first step of the construction of the data structure is to apply
the algorithm of \cref{thm:strong-unbreakability} to compute, in time
$2^{\Oh(k^2)}\cdot |G|^2\cdot \|G\|$, a strongly $(q,k)$-unbreakable
tree decomposition $\Tt=(T,\bag)$ of~$G$ of adhesion at most $q$,
where $q=q(k)=2^{\Oh(k)}$. As argued, we may assume that $\Tt$ is
regular. We fix the parameters $k$ and $q$ for the remainder of this
section.

\subsection{Basic ingredients}

Before we proceed to the main construction, we need to set up some
auxiliary data structures that enable us to navigate and query the
tree decomposition $\Tt$.

We will need the following lemma for navigation in $T$. Recall that we
write $\lca(x,y)$ for the least common ancestor of $x$ and $y$ in $T$
and $\dir(x,y)$ (if $x$ is a strict ancestor of $y$) for the unique
child $y'$ of $x$ that is an ancestor of $y$. The proof of the
following lemma relies on the classic data structure for
$\lca(\cdot,\cdot)$ queries of Harel and Tarjan~\cite{HarelT84}.

\begin{lemma}\label{lem:navigation}
  Given a tree $T$, one can compute in time $\Oh(|T|)$ a data
  structure that may answer the $\lca(x,y)$ and $\dir(x,y)$ queries in
  time $\Oh(1)$.
\end{lemma} \begin{proof}
  For the $\lca(\cdot,\cdot)$ queries we may
  use the data structure of Harel and Tarjan~\cite{HarelT84}, which
  can be initialized in time $\Oh(|T|)$ and supports queries in
  constant time. We now show that we can use this data structure as
  well for the $\dir(\cdot,\cdot)$ queries.

  For every node $x$ of $T$, arbitrarily enumerate the children of $x$
  with numbers $1,\ldots,d(x)$, where~$d(x)$ is the number of children
  of $x$. By applying a linear-time preprocessing, we may assume that
  every~$x$ is supplied with an array of length $d(x)$ containing
  pointers to its children, in this order.  Construct a tree~$T'$ from
  $T$ as follows:
  \begin{itemize}
  \item Replace every node $x$ with a path consisting of nodes
    $(x,0),(x,1),\ldots,(x,d(x))$, where $(x,i)$ is the parent of
    $(x,i-1)$ for each $i\in \{1,\ldots,d(x)\}$.
  \item For each $i\in \{1,\ldots,d(x)\}$, if $y$ is the $i$th child
    of $x$ in $T$, then we make $(y,d(y))$ a child of $(x,i)$.
  \end{itemize}

  See \cref{fig:tree-unraveling} for an illustration. Note that the number of nodes of $T'$
  is upper bounded by the total number of nodes and edges of $T$,
  hence $|T'|<2|T|$.  \medskip

  \begin{figure}[ht]
      \centering
       \includegraphics[page=8,scale=3]{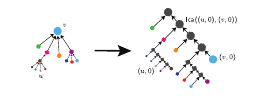}
       \caption{Construction of $T'$ from $T$ in the proof of Lemma~\ref{lem:navigation}. A node $v$ of degree $d$ is replaced by a path $(v,d),\ldots,(v,1)$ of $d$ nodes of degree  $2$, followed by a leaf node $(v,0)$.
       If $u$ is a descendant of $v$ in~$T$, then 
       $\dir(v,u)$ can be deduced from $\lca((u,0),(v,0))$ in $T'$.
       }

    \label{fig:tree}\label{fig:tree-unraveling}
  \end{figure}

  The following claim is clear.

  \begin{claim}
    If $x$ is a strict ancestor of $y$ in $T$, then the least common
    ancestor of $(x,0)$ and $(y,0)$ in $T'$ is the node $(x,i)$ such
    that $\dir(x,y)$ is the $i$th child of $x$ in $T$.
  \end{claim}

  Hence, if we set up the data structure of Harel and Tarjan for the
  tree $T'$, then a $\dir(x,y)$ query in~$T$ can be answered as
  follows: compute the least common ancestor of $(x,0)$ and $(y,0)$ in
  $T'$, read its second coordinate $i$, and output the $i$th child of
  $x$ in $T$. This takes time $\Oh(1)$. \end{proof}

%
%

Next, we set up a data structure to encode and quickly evaluate
connectivity in subgraphs that are ``between'' the adhesion of two
nodes. These subgraphs are defined by the following notion of
\emph{torsos}.

\paragraph*{Torsos.} For a graph $H$ and a subset of its vertices $X$,
the {\em{torso}} of $X$ in $H$ is the graph $\torso_H(X)$ on vertex
set $X$ where a pair of vertices $u,v\in X$ is adjacent if and only if
in $H$ there exists a $u$-$v$ path whose interior is disjoint with
$X$. For a pair of nodes $x$ and $y$ of $T$, where $x$ is an ancestor
of $y$, we define
$$\torso(x,y)\coloneqq \torso_{G[\cone(x)\setminus \cmp(y)]}(\adh(x)\cup \adh(y)).$$

Note that $\torso(x,y)$ is a graph on the vertex set
$\adh(x)\cup \adh(y)$, which comprises of at most $2q$ vertices.
Intuitively, $\torso(x,y)$ encodes the connectivity pattern between
the elements of $\adh(x)$ and of $\adh(y)$ within the subgraph induced
by the context in $\Tt$ between $x$ and $y$.


We will need a data structure for evaluating the torsos quickly.

\begin{lemma}\label{lem:torso-ds}
  One can compute in time $2^{\Oh(q^2)}\cdot |G|^2\cdot \|G\|$ a data
  structure that can answer the following query in time
  $2^{\Oh(q^2)}$: for given $x,y\in V(T)$, where $x$ is an ancestor of
  $y$, output the graph $\torso(x,y)$.  The data structure takes space
  $2^{\Oh(q^2)}\cdot |G|$.
\end{lemma}

The proof of \cref{lem:torso-ds} is postponed to
\Cref{sec:torso-queries}. Here is a quick sketch. By
appropriately coloring the vertices of the adhesions with colors from
$\{1,\ldots,q\}$, we can reinterpret each graph $\torso(x,y)$ as an
element of a semigroup of size $2^{\Oh(q^2)}$, in the same way as was
done in~\cite{BojanczykP16}. Then $\torso(x,y)$ can be obtained as
the semigroup product of the graphs $\torso(z_{i-1},z_{i})$ for
$i=1,\ldots,\ell$, where $x=z_0-z_1-\ldots-z_\ell=y$ is the path from
$x$ to $y$ in~$T$. We design a data structure for path product queries in semigroup-labeled
trees. We believe that this data structure is a general tool of
independent interest, hence we actually give two different proofs. The
first one is based on a classic approach using deterministic Simon's
factorization of Colcombet~\cite{Colcombet07}, and can be viewed as a
reinterpretation of ideas from~\cite{Colcombet07}. The second is a
self-contained, simpler argument that seems new.


\subsection{Data structure}\label{sec:ds-description}

With the basic ingredients in place, we may proceed with the proof of
\cref{thm:general-queries}. Let us describe the data structure we are
going to use. It consists of:
\begin{itemize}
\item The tree decomposition $\Tt$, constructed using the algorithm of
  \cref{thm:strong-unbreakability}, with underlying tree~$T$. Each
  vertex $u\in V(G)$ has a pointer to the unique node $x\in V(T)$ such
  that $u\in \mrg(x)$.
\item For each $x\in V(T)$, the $\emptyset$-restricted bag graph
  $\bgraph_\emptyset(x)$.
\item The data structure for $\lca(x,y)$ and $\dir(x,y)$ queries,
  provided by \cref{lem:navigation}. Call it $\Lb$.
\item The data structure for $\torso(x,y)$ queries, provided by
  \cref{lem:torso-ds}. Call it $\Db$.
\end{itemize}

As for initialization, the time needed to compute $\Tt$ is
$2^{\Oh(k^2)}\cdot |G|^2\cdot \|G\|$. It is straightforward to see
that for each particular $x\in V(T)$, the graph $\bgraph_\emptyset(x)$
can be computed in time
$\Oh(q^{\Oh(1)}\cdot \|G\|)=2^{\Oh(k)}\cdot \|G\|$, hence the
construction of those graphs takes time
$2^{\Oh(k)}\cdot |G|\cdot \|G\|$ in total. The time needed to
construct $\Db$ is
$2^{\Oh(q^2)}\cdot |G|^2\cdot \|G\|=2^{2^{\Oh(k)}}\cdot |G|^2\cdot
\|G\|$. Finally, the time needed to construct $\Lb$ is
$\Oh(|V(T)|)\leq \Oh(|G|)$. All in all, the total initialization time
is $2^{2^{\Oh(k)}}\cdot |G|^2\cdot \|G\|$, as promised.

We now estimate the space usage of the data structure. The data
structures $\Lb$ and $\Db$ take space $\Oh(|T|)\leq \Oh(|G|)$ and
$2^{\Oh(q^2)}\cdot |T|\leq 2^{2^{\Oh(k)}}\cdot |G|$, respectively, by
\cref{lem:navigation} and \cref{lem:torso-ds}. By
\cref{lem:total-size}, we have
$\sum_{x\in V(T)} |\bag(x)|\leq \Oh(q|G|)\leq 2^{\Oh(k)}\cdot |G|$,
hence the total size of the description of the tree decomposition
$\Tt$ is $2^{\Oh(k)}\cdot |G|$ as well. Finally, by
\cref{lem:total-size} we have
$\sum_{x\in V(T)} \|\bgraph(x)\|\leq \Oh(q^2|G|)$. Recall that
$\bgraph_\emptyset(x)$ is obtained from $\bgraph(x)$ by removing all
the children of $x$ and turning their neighborhoods into cliques. As
each of these neighborhoods has size at most $q$ and
$\sum_{x\in V(T)}|\children(x)|\leq |T|\leq |G|$, we conclude that
these operations add at most $\binom{q}{2}\cdot |G|$ edges to graphs
$\bgraph_\emptyset(x)$ in total. It follows that
$\sum_{x\in V(T)} \|\bgraph(x)\|\leq \Oh(q^2|G|)$, so storing the
graphs $\bgraph_\emptyset(x)$ takes $\Oh(q^2|G|)=2^{\Oh(k)}\cdot |G|$
space. This concludes the analysis of the space complexity.

Before continuing, we need to clarify how exactly each graph
$\bgraph_\emptyset(x)$, for $x\in V(T)$, is represented in the data
structure.
The edges of $G$ can be partitioned into {\em{original edges}} ---
those existing in $G[\bag(x)]$ --- and {\em{adhesion edges}} --- edges
$uv$ that are not present in $G[\bag(x)]$, but for which there exists
$y\in \children(x)$ with $\{u,v\}\subseteq \adh(y)$. In the latter
case we say that $y$ {\em{supports}} the adhesion edge~$uv$. The edge
set of $\bgraph_\emptyset(x)$ is stored as usual using neighborhood
lists: every vertex is supplied with a list of edges incident to
it. However, for every edge on this list we remember whether it is an
original or an adhesion edge, and in the latter case we also remember
the list of all children of $x$ supporting the edge. It is easy to see
that this representation of $\bgraph_\emptyset(x)$ can be computed and
stored within the time and space complexity claimed above.

\subsection{Implementation of queries}
\label{sec:implementation of queries}
We are left with implementing connectivity queries using the data
structure described above. Suppose we are given a query $\conn(u,v,S)$
to answer, where $u,v\in V(G)$ and $S$ is a vertex subset of size at
most $k$.  Let $\Sh\coloneqq S\cup \{u,v\}$. Note that
$|\Sh|\leq k+2$.

For $x\in V(T)$, let
$$D(x)\coloneqq \adh(x)\cup (\{u,v\}\cap \cone(x)).$$

The {\em{profile}} of $x$ is the set
$\profile(x)\subseteq\binom{D(x)}{2}$ consisting of all pairs
$\{a,b\}\in \binom{D(x)}{2}$ satisfying the following condition: in
$G[\cone(x)]$ there is an $a$-$b$ path that does not pass through any
vertex of $S$.

Let $r$ be the root of $T$. Then $D(r)=\{u,v\}$ and $\profile(r)$ is
non-empty if and only if \mbox{$G\models \conn(u,v,S)$}. Therefore, it
suffices to compute $\profile(r)$. We will do this by computing the
profiles of a selection of~$\Oh(k)$ nodes of $T$ in a bottom-up
manner.

Precisely, let
$$X\coloneqq \set{r}\cup \setof{x\in V(T)}{\mrg(x)\cap \Sh\neq \emptyset}.$$

Note that $|X|\leq k+3$ and $X$ can be computed in time $\Oh(k)$,
because every vertex of $G$ has a pointer to the node of $T$ whose
margin contains it. Further, let
$$Y\coloneqq \textrm{least common ancestor closure of }X;$$
that is, $Y$ comprises of $X$ and all least common ancestors of pairs of nodes from $X$. Then
\[|Y|\leq 2|X|-1\leq 2k+5\] and $Y$ is closed under taking least
common ancestors. Note that $Y$ can be computed in time $\Oh(k^2)$
using the data structure $\Lb$. Our goal is to compute
\[\profile(x)\quad \textrm{for each}\quad x\in Y\]
in a bottom-up manner. Each set $\profile(x)$ is stored as a list;
note that its size is at most $\binom{q+2}{2}=\Oh(q^2)$.

The next two lemmas provide the main ingredients for the
computation. Here, we say that a node $x$ of $T$ is {\em{affected}} if
$\cmp(x)$ contains a vertex of $\Sh$, and $x$ is {\em{unaffected}}
otherwise.


\begin{lemma}\label{lem:warp}
  Suppose $z$ is an ancestor of $y$ in $T$, and $\Sh$ is disjoint with
  $\cmp(z)\setminus \cone(y)$. Then given $\profile(y)$, one can
  compute $\profile(z)$ using one query to the data structure $\Db$
  and $q^{\Oh(1)}$ additional time.
\end{lemma}
\begin{proof}
  Using a query to $\Db$, compute $\torso(z,y)$. Construct a graph $J$
  on vertex set $D(y)\cup \adh(z)$ by defining the edge set of $J$ to
  be the union of $\profile(y)$ and the edge set of
  $\torso(z,y)$. Note that $|V(J)|\leq 2q+2$ and $D(z)\subseteq
  V(J)$. The following claim is clear from the construction and the
  assumption that $\Sh$ is disjoint with $\cmp(z)\setminus \cone(y)$.

 \begin{claim}\label{cl:warp-J}
   For every pair of vertices $a,b\in V(J)$, there is an $a$-\,$b$
   path in $G[\cone(z)]$ that does not pass through any vertex of $S$
   if and only if there is such a path in the graph $J$.
 \end{claim}

 Therefore, for every pair $\{a,b\}\in \binom{D(z)}{2}$, we can find
 out whether $\{a,b\}\in \profile(z)$ using a breadth-first search in
 $J$.  Since $J$ had $\Oh(q)$ vertices and hence $\Oh(q^2)$ edges,
 this takes $\Oh(q^2)$ time per pair, hence $\Oh(q^4)$ time in total.
\end{proof}

\begin{lemma}\label{lem:aggregate}
  Let $x$ be a node of $T$. Suppose we are given the set $Z$
  comprising of all affected children of~$x$ and, for each $z\in Z$,
  we are given $\profile(z)$. Then one can compute $\profile(x)$ in
  time $q^{\Oh(1)}$.
\end{lemma}
Note that this statement holds in particular when $x$ has no affected
children.
\begin{proof}
  Note that since the components at the children of $Z$ are pairwise
  disjoint, we have $$|Z|\leq |\Sh|\leq k+2.$$

  Recall that the data structure contains the $\emptyset$-restricted
  bag graph $\bgraph_\emptyset(x)$, represented as described at the
  end of \cref{sec:ds-description}. We now show that we can emulate
  the neighborhood lists of the $S$-restricted bag graph
  $\bgraph_S(x)$ using this representation.

  \begin{claim}\label{cl:emulation}
    Given $a\in \bag(x)\setminus S$, one can in time $q^{\Oh(1)}$
    either conclude that the degree of $a$ in $\bgraph_S(x)$ is larger
    than $q$, or list all the neighbors of $a$ in $\bgraph_S(x)$.
  \end{claim}
  \begin{claimproof}
    Note that $\bgraph_S(x)$ is a subgraph of
    $\bgraph_\emptyset(x)$. Therefore, it suffices to iterate over the
    neighborhood list of $a$ in $\bgraph_\emptyset(x)$ and for each
    edge $ab$ on this list, decide whether it is present in
    $\bgraph_S(x)$.

    Note that the edge $ab$ is not present in $\bgraph_S(x)$ if and
    only if at least one of the following conditions holds:
    \begin{itemize}
    \item $b\in S$; or
    \item $ab$ is an adhesion edge and for every node $z$ supporting
      $ab$, we have $z\in Z$ and $\{a,b\}\notin \profile(z)$.
    \end{itemize}

    We observe that the second condition can be checked in time
    $q^{\Oh(1)}$. Indeed, since $|Z|\leq k+2$, if there are more than
    $k+2$ nodes supporting $ab$, then the condition must be
    false. Otherwise we can iterate through the list of all nodes $z$
    supporting $ab$ (stored along with $ab$ in the representation of
    $\bgraph_\emptyset(x)$), and for each of them check whether it
    belongs to $Z$ and satisfies $\{a,b\}\notin
    \profile(z)$. Implemented in a straightforward way, this takes
    time polynomial in $q$.

    Hence, for every particular edge $ab$ present in
    $\bgraph_\emptyset(x)$ we can decide whether it is present also in
    $\bgraph_S(x)$ in time $q^{\Oh(1)}$. The next observation is that
    only for at most $(k+2)\binom{q}{2}$ edges this check will yield a
    negative answer. Indeed, at most $|S|\leq k$ edges $ab$ may not be
    present in $\bgraph_S(x)$ due to $b\in S$, and every $z\in Z$
    supports at most $\binom{|\adh(z)|}{2}\leq \binom{q}{2}$ edges,
    resulting in at most
    $|Z|\cdot \binom{q}{2}\leq k+(k+2)\binom{q}{2}$ edges $ab$ not
    being present in $\bgraph_S(x)$ due to the second
    condition. Therefore, if the degree of $a$ in
    $\bgraph_\emptyset(G)$ is larger than $q+k+(k+2)\binom{q}{2}$,
    then its degree in $\bgraph_S(G)$ must be larger than $q$, and we
    may report this outcome without analyzing the neighborhood list of
    $a$. Otherwise, we process all the $\Oh(q^2k)$ edges incident to
    $a$ in $\bgraph_\emptyset(G)$, and for each of them we decide
    whether it is present in $\bgraph_S(G)$ in time $q^{\Oh(1)}$. This
    results in listing all edges incident to $a$ in $\bgraph_S(G)$ in
    time $q^{\Oh(1)}$.
  \end{claimproof}

  The next claim is the key observation: connectivity in
  $\bgraph_S(x)$ can be decided in time $q^{\Oh(1)}$, and the crux is
  to use \cref{lem:short-BFS-combinatorial}.

 \begin{claim}\label{cl:short-bfs}
   Given $a,b\in \bag(x)\setminus S$, one can decide whether $a$ and
   $b$ are connected in $\bgraph_S(x)$ in time $q^{\Oh(1)}$.
 \end{claim}
 \pagebreak
 \begin{claimproof}
   By \cref{lem:short-BFS-combinatorial}, the following algorithm
   solves this problem:
   \begin{itemize}
   \item Run a breadth-first search from $a$ in $\bgraph_S(x)$,
     terminating the search whenever any of the following conditions
     is met (by {\em{reaching}} a vertex we mean that it is added to
     the bfs queue):
     \begin{itemize}
     \item $b$ has been reached;
     \item a vertex of degree larger than $q$ in $\bgraph_S(x)$ is
       popped from the queue;
     \item in total, the search has reached $q+1$ vertices.
     \end{itemize}
     If the search terminated due to the queue being empty before any
     of the conditions above is met, conclude that $a$ and $b$ are not
     connected in $\bgraph_S(x)$.
   \item Do the same with the roles of $a$ and $b$ swapped.
   \item If neither of the steps above reached the conclusion that $a$
     and $b$ are not connected in $\bgraph_S(x)$, then conclude that
     $a$ and $b$ are connected in $\bgraph_S(x)$.
   \end{itemize}
   Note here that the second termination condition stated above ---
   when a vertex of degree larger than $q$ is popped from the queue
   --- is also correct for the following reason: such a vertex
   witnesses that the connected component of $\bgraph_S(x)$ on which
   the search is employed has more than $q$ vertices.

   Now it is straightforward to see that the algorithm presented above
   can be implemented in time~$q^{\Oh(1)}$ with the help of
   \cref{cl:emulation}. Namely, when in a breadth-first search a
   vertex $w$ is popped from the queue, we need to iterate through all
   the neighbors of $w$ and add them to the queue provided they were
   not reached before. By \cref{cl:emulation}, we can either list
   those neighbors in time $q^{\Oh(1)}$ or conclude that the search
   can be immediately terminated. Since the search will not reach more
   than $q+1$ vertices in total, for each of those neighbors we can
   check whether it was already reached in time $\Oh(q)$. At most
   $q+1$ vertices will be processed in this way, resulting in the
   total time complexity of $q^{\Oh(1)}$.
 \end{claimproof}

 Now observe that for every pair
 $\{a,b\}\in \binom{D(x)\cap \bag(x)}{2}$, $\{a,b\}$ should be
 included in $\profile(x)$ if and only if $a$ and $b$ are connected in
 $\bgraph_S(x)$; this follows from
 \cref{lem:short-BFS-combinatorial}. We can use \cref{cl:short-bfs} to
 decide this in time $q^{\Oh(1)}$ per pair, so also also in time
 $q^{\Oh(1)}$ for all such pairs in total.

 We are left with making the same decision for pairs
 $\{a,b\}\in \binom{D(x)}{2}$ where either $a$ or $b$ is contained in
 $D(x)\setminus \bag(x)$. Noting that
 $D(x)\setminus \bag(x)\subseteq \{u,v\}$, by symmetry suppose that
 $a=u$ and $u\notin \bag(x)$. Let $z$ be the child of $x$ satisfying
 $u\in \cmp(z)$. Note that $z=\dir(x,y)$ where $y$ is such that
 $u\in \mrg(y)$, hence $z$ can be computed in time $\Oh(1)$ using a
 query to $\Lb$.

 Suppose first that $b\in \bag(x)$. Then $\{u,b\}$ should be included
 in $\profile(x)$ if and only if there exists $w\in D(z)$ such that
 $\{u,w\}\in \profile(z)$ and $w,b$ are connected in
 $\bgraph_S(x)$. This condition can be checked in time $q^{\Oh(1)}$
 for each particular $w\in \adh(z)$ using the algorithm of
 \cref{cl:short-bfs}, hence every pair $\{u,b\}$ as above can be
 decided in time $q^{\Oh(1)}$. There are $\Oh(q)$ vertices
 $b\in D(x)\cap \bag(x)$ to check, hence these pairs can be decided in
 total time $q^{\Oh(1)}$.

 Finally, we are left with the case when $b=v$ and $v\notin
 \bag(x)$. Let $z'$ be the child of $x$ such that $v\in \cmp(z')$;
 similarly as before, we can find $z'$ in time $\Oh(1)$ using
 $\Lb$. Then $\{u,v\}$ should be included in $\profile(x)$ if and only
 if at least one of the following conditions is satisfied:
 \begin{itemize}
 \item $z=z'$ and $\{u,v\}\in \profile(z)$.
 \item There exist $w\in \adh(z)$ and $w'\in \adh(z')$ such that
   $\{u,w\}\in \profile(z)$, $\{w',v\}\in \profile(z')$, and $w,w'$
   are connected in $\bgraph_S(x)$.
 \end{itemize}

 The first condition can be checked in time $\Oh(q^2)$, while the
 second can be checked in time $q^{\Oh(1)}$ by applying the algorithm
 of \cref{cl:short-bfs} for every pair
 $(w,w')\in \adh(z)\times \adh(z')$.

 Altogether, we have decided which pairs of $\binom{D(x)}{2}$ should be
 included in $\profile(x)$ in time $q^{\Oh(1)}$ and the proof is
 complete.
\end{proof}

We may now implement the $\conn(u,v,S)$ query using \cref{lem:warp}
and~\ref{lem:aggregate}. Let $T_Y$ be the tree on node set $Y$ with
the ancestor order inherited from $T$; that is, $x$ is the parent of
$y$ in $T_Y$ if $x$ is a strict ancestor of $y$ in $T$ and the
interior of the $x$-$y$ path in $T$ is disjoint with $Y$. Note that
$T_Y$ can be constructed in time $\Oh(k^2)$ using the data structure
$\Lb$.

Now, we compute $\profile(x)$ for all $x\in Y$ in a bottom-up manner
over $T_Y$. Let $y_1,\ldots,y_\ell$ be the children of $x$ in $T_Y$
(possibly $\ell=0$ if $x$ is a leaf of $T_Y$). For each
$i\in \{1,\ldots,\ell\}$, let
$$z_i\coloneqq \dir(x,y_i).$$

Note that each $z_i$ can be found in time $\Oh(1)$ using a query to
$\Lb$, hence computing $z_1,\ldots,z_\ell$ takes time~$\Oh(k)$. The
following claim follows immediately from taking $Y$ to be the least
common ancestor closure of~$X$.

\begin{claim}
  For each $i\in \{1,\ldots,\ell\}$, $\Sh$ is disjoint with
  $\cmp(z_i)\setminus \cone(y_i)$.
\end{claim}

Note that the profiles $\profile(y_i)$ for $i\in \{1,\ldots,\ell\}$ have
already been computed before. Therefore, we may apply \cref{lem:warp}
$\ell$ times to compute the profiles $\profile(z_i)$ for
$i\in \{1,\ldots,\ell\}$, each using one query to $\Db$ and
$q^{\Oh(1)}$ additional time.

Finally, observe that $z_1,\ldots,z_\ell$ are precisely the affected
children of $x$. Therefore, we may apply \cref{lem:aggregate} to
compute $\profile(x)$ in time $q^{\Oh(1)}$. Remember that this works in
particular when $x$ is a leaf of $T_Y$ and has no affected children.

At the end we compute $\profile(r)$ for $r$ being the root of $T$. As
we argued, $G\models\conn(u,v,S)$ if and only if $\profile(r)$ is
non-empty, hence we may now answer the query.

As for the time complexity, the whole procedure uses $\Oh(k)$ queries
to $\Db$ and $q^{\Oh(1)}$ additional time. Since every query to $\Db$
takes time $2^{\Oh(q^2)}$, by \cref{lem:torso-ds}, and $q=2^{\Oh(k)}$,
the total running time is bounded by a doubly-exponential function of
$k$, as promised. This concludes the proof of
\cref{thm:general-queries}.

%% file: torso-queries.tex

\section{Torso queries}\label{sec:torso-queries}

In this section we give a proof of \cref{lem:torso-ds}. As mentioned,
the main idea is to reinterpret torso queries as queries about
products along paths in a tree labeled with elements of a fixed finite
semigroup, and then design a data structure for this more general kind
of queries. We start by introducing the setting of semigroup-labeled
trees and showing how torso queries fit into this setting. Then we
design a data structure for path product queries in semigroup-labeled
trees. We believe that this data structure is a general tool of
independent interest, hence we actually give two different proofs. The
first one is based on a classic approach using deterministic Simon's
factorization of Colcombet~\cite{Colcombet07}, and can be viewed as a
reinterpretation of ideas from~\cite{Colcombet07}. The second is a
self-contained, simpler argument that seems new.

\paragraph*{Semigroup-labeled trees and path product queries.}
Recall that a semigroup $S$ consists of a set of elements, called the
{\em{support}} and also denoted by $S$, and a binary associative
operation on the support called the {\em{product}} and denoted
$\cdot$. A semigroup is finite if it has finite support.

For a semigroup $S$, an {\em{$S$-labeled tree}} is a rooted tree $T$ together with a map that with each edge~$e$ of~$T$ associates an element $\sigma(e)\in S$. Suppose $x,y$ are nodes of $T$, where $x$ is a strict ancestor of $y$. Then we define
\[\sigma(x,y)\coloneqq \sigma(e_1)\cdot \sigma(e_2)\cdot\ldots\cdot\sigma(e_\ell),\]
where $e_1,\ldots,e_\ell$ are the consecutive edges of the path from
$x$ to $y$ in $T$ (that is, $e_1$ is incident with $x$ and $e_\ell$ is
incident with $y$).

The following theorem is the main result of this section: it provides
an efficient data structure for path product queries, parameterized by
the size of the semigroup.

\begin{theorem}\label{thm:colcombet-ds}
  Fix a finite semigroup $S$ and suppose we are given an $S$-labeled
  tree $T$. Then one can in time $|S|^{\Oh(1)}\cdot |T|$ construct a
  data structure that can answer the following query in time
  $|S|^{\Oh(1)}$ and using $|S|^{\Oh(1)}$ operations in $S$: given
  nodes $x,y\in V(T)$, where $x$ is a strict ancestor of $y$, output
  $\sigma(x,y)$. The data structure uses $|S|^{\Oh(1)}\cdot |T|$
  memory and stores $|S|^{\Oh(1)}\cdot |T|$ elements of $S$.
\end{theorem}

Later we give two proofs of \cref{thm:colcombet-ds}, but let us see
now how torso queries can be understood through path product queries
in a semigroup-labeled tree. That is, we are going to prove
\cref{lem:torso-ds} assuming \cref{thm:colcombet-ds}.

\paragraph*{Torso queries as path product queries.} Assume the setting
from the statement of \cref{lem:torso-ds}. We interpret torso queries
as path product queries over trees labeled with a finite semigroup of
size $2^{\Oh(q^2)}$. For this, we shall use the semigroup of
{\em{bi-interface graphs}}, introduced for a similar purpose
in~\cite{BojanczykP16}.

Let $r\in \N$. A {\em{bi-interface graph}} of arity $r$ consists of a
graph $G$ and two partial mappings $\lft,\rgt\colon V(G)\partto [r]$
that are injective and satisfy the following property: if
$u\in \dom{\lft}\cap \dom{\rgt}$, then $\lft(u)=\rgt(u)$. The domains
of the mappings $\dom{\lft}$ and $\dom{\rgt}$ are called the {\em{left
    interface}} and the {\em{right interface}}, respectively.

A bi-interface graph is {\em{basic}} if its vertex set is a subset of
$[r]\times \{\Lsf,\Rsf,\Ssf\}$,
$\dom{\lft}\setminus \dom{\rgt}\subseteq [r]\times \{\Lsf\}$,
$\dom{\rgt}\setminus \dom{\lft}\subseteq [r]\times \{\Rsf\}$,
$\dom{\lft}\cap \dom{\rgt}\subseteq [r]\times \{\Ssf\}$, and
$V(G)=\dom{\lft} \cup \dom{\rgt}$. By $\Basic_r$ we denote the set of
all basic bi-interface graphs of arity $r$. Note that
$$|\Basic_r|\leq 2^{\Oh(r^2)},$$
as this is an upper bound on the total number of different graphs on
vertex set $[r]\times \{\Lsf,\Rsf,\Ssf\}$.

Next, we define the {\em{abstraction operation}}, which takes a
bi-interface graph $\Gb=(G,\lft,\rgt)$ and outputs the basic
bi-interface graph $\abstr{\Gb}$, of same arity, defined as follows:
\begin{itemize}
\item First, for a vertex $u\in \dom{\lft}\cup \dom{\rgt}$, define
  $\iota(u)$ as
 $$ \iota(u) :=
 \begin{cases}
  (\lft(u),\Lsf) &\text{if } u\in \dom\lft\setminus \dom{\rgt},\\
  (\rgt(u),\Rsf) &\text{if } u\in \dom\rgt\setminus \dom{\lft}, \text{and}\\
  (\lft(u),\Ssf)=(\rgt(u),\Ssf) &\text{if } u\in \dom\lft\cap \dom{\rgt}.
 \end{cases}
 $$

\item The vertex set of $\abstr{\Gb}$ is
  $\{\iota(u)\colon u\in \dom\lft \cup \dom{\rgt}\}$.
\item In $\abstr{\Gb}$ there is an edge between $\iota(u)$ and
  $\iota(v)$ if and only if in $G$ there is a path from $u$ to $v$
  that does not pass through any other vertex of
  $\dom{\lft}\cup \dom{\rgt}$.
\item The interface mappings $\lft',\rgt'$ of $\abstr{\Gb}$ are
  naturally inherited from $\Gb$: $\lft'(\iota(u))=\lft(u)$ for
  $u\in \dom{\lft}$ and $\rgt'(\iota(u))=\rgt(u)$ for
  $u\in \dom{\rgt}$. \end{itemize}

Thus, intuitively speaking, $\abstr{\Gb}$ is obtained by applying the
torso operation in $\Gb$ with respect to the union of the interfaces,
and reindexing the vertices according to the labels assigned to the
interfaces.

Let $\Gb=(G,\lft_G,\rgt_G)$ and $\Hb=(H,\lft_H,\rgt_H)$ be two
bi-interface graphs of arity $r$. The {\em{composition}} of $\Gb$ and
$\Hb$ is the bi-interface graph $\Gb\cdot \Hb=(J,\lft_G,\rgt_H)$,
where $J$ is obtained from the disjoint union of $G$ and $H$ by fusing
a vertex $u\in \dom{\rgt_G}$ with a vertex $v\in \dom{\lft_H}$
whenever $\rgt_G(u)=\lft_H(v)$. Thus, bi-interface graphs of arity $r$
with the composition operation form an (infinite) semigroup; call it
$S_r$. We can now endow the basic bi-interface graphs of arity $r$
with composition operation $\odot$ defined as follows: for basic
bi-interface graphs $\Ab$ and $\Bb$, we set
\[\Ab\odot\Bb\coloneqq \abstr{\Ab\cdot \Bb}.\]

Thus, basic bi-interface graphs of arity $r$ endowed with $\odot$ form
a semigroup of size $2^{\Oh(r^2)}$; call this semigroup $B_r$. It is
straightforward to check that the abstraction operation
$\abstr{\cdot}$ is a semigroup homomorphism from $S_r$ to $B_r$, for
every $r\in \N$.

With the definitional layer introduced, we can state a lemma that
encapsulates the translation from torso queries to queries over
$B_q$-labeled trees. This lemma will be used not only in this section,
but also later on for the purpose of viewing computation of torsos in
algebraic terms.

\begin{lemma}\label{lem:translation}
  Suppose we are given a tree decomposition $\Tt=(T,\bag)$ of a graph
  $G$ of adhesion $q$. Then one can in time $q^{\Oh(1)}\cdot |T|\|G\|$
  compute injective colorings
  $\{\lambda_x\colon \adh(x)\to [2q]\colon x\in V(T)\}$ and a mapping
  $\sigma\colon E(T)\to B_{2q}$ such that the following properties are
  satisfied:
 \begin{itemize}
 \item If $x$ is the parent of $y$ in $T$, $u\in \adh(x)$, and
   $v\in \adh(y)$, then
  $$\lambda_x(u)=\lambda_y(v)\qquad\textrm{if and only if}\qquad u=v.$$
  \item If $x$ is a strict ancestor of $y$ in $T$, and we denote
  $$\Gb(x,y)=(G[\cone(x)\setminus \cmp(y)],\lambda_x,\lambda_y),$$
  then $\Gb(x,y)$ is a correctly defined bi-interface graph and $$\abstr{\Gb(x,y)}=\sigma(x,y).$$
 \end{itemize}
\end{lemma}

\begin{proof}
  Compute the colorings $\lambda_x$ in a top-down manner as follows:
 \begin{itemize}
 \item If $r$ is the root of $T$, then $\adh(r)$ is empty, hence set
   $\lambda_r$ to be the empty coloring.
 \item If $y$ is a child of $x$ and $\lambda_x$ has already been
   defined, then construct $\lambda_y$ as follows: for each
   $u\in \adh(x)\cap \adh(y)$ set
   $\lambda_y(u)\coloneqq \lambda_x(u)$, and then color the vertices
   of $\adh(y)\setminus \adh(x)$ using colors from
   $[2q]\setminus \lambda_x(\adh(x))$ in an arbitrary injective
   manner. Note that since $|\adh(x)|\leq q$, we have
   $|[2q]\setminus \lambda_x(\adh(x))|\geq q$, so such an injective
   coloring always exists.
 \end{itemize}
 
 That the colorings $\{\lambda_x\colon x\in V(T)\}$ satisfy the first
 point from the lemma statement follows immediately from the
 construction. Clearly, they can be computed in time $\Oh(q|T|)$ by
 performing a depth-first search in $T$.

 Next, for an edge $xy\in E(T)$, where $x$ is the parent of $y$, we construct
 \[\sigma(xy)\coloneqq \abstr{\Gb(x,y)},\]
 where $\Gb(x,y)$ is defined as in the lemma statement. Computing
 $\sigma(xy)$ boils down to running a breadth-first search in
 $\Gb(x,y)$ from each vertex of $\adh(x)\cup \adh(y)$; this can be
 done in time $q^{\Oh(1)}\cdot \|G\|$ per node $y$, so in time
 $q^{\Oh(1)}\cdot |T|\|G\|$ in total.

 It remains to note that if $x$ is a strict ancestor of $y$ and
 $x=z_0-z_1-\ldots-z_\ell=y$ is the path from~$x$ to~$y$ in $T$, then
 $$\Gb(x,y)=\Gb(z_0,z_1)\cdot \Gb(z_1,z_2)\cdot\ldots\cdot \Gb(z_{\ell-1},z_\ell),$$
 from which it follows that
 \[\abstr{\Gb(x,y)}=\sigma(z_0z_1)\odot \sigma(z_1z_2)\odot \ldots
   \odot \sigma(z_{\ell-1}z_\ell)=\sigma(x,y).\hfill\qedhere\]
\end{proof}

Now \cref{lem:torso-ds} follows directly from combining
\cref{lem:translation} and~\cref{thm:colcombet-ds}. More precisely,
using \cref{lem:translation} we compute the mapping $\sigma$ and
colorings $\{\lambda_x\colon x\in V(T)\}$, and then we set up the data
structure for $\sigma(x,y)$ queries provided by
\cref{thm:colcombet-ds}. To answer a query about $\torso(x,y)$, we
first query this data structure to obtain $\sigma(x,y)$, and then we
translate the names of vertices from $[q]\times \{\Lsf,\Rsf,\Ssf\}$
back to $\adh(x)\cup \adh(y)$ by reverting the colorings $\lambda_x$
and $\lambda_y$.

\bigskip

We now proceed to \cref{thm:colcombet-ds}. In both proofs we use the
following easy corollary of \cref{lem:navigation}.

\begin{lemma}\label{lem:partition navigation}
  Given a tree $T$ and a set of nodes $L\subseteq V(T)$, one can in
  time $\Oh(|T|)$ construct a data structure that may
  answer the following query in time $\Oh(1)$: given nodes $x$ and
  $y$, where $x$ is an ancestor of $y$ in $T$, find the topmost node
  on the $x$-$y$ path in $T$ that belongs to~$L$, or report that no
  such node exists. \end{lemma}

\begin{proof}
  For every $x\in V(T)$, we store whether $x\in L$. Further, let
  $\anc(x)$ be the closest (deepest) strict ancestor of $x$ that
  belongs to $L$, or $\bot$ if no such ancestor exists. Note that the
  values of $\anc$ can be computed in time $\Oh(|T|)$ using a
  depth-first search in $T$.

  Next, let $T_L$ be the rooted tree on node set $L\cup \{\bot\}$
  constructed from $T$ as follows:
 \begin{itemize}
 \item Restrict the node set to $L$ keeping the ancestor relation;
   that is, for $x,y\in L$ we have that $x$ is an ancestor of $y$ in
   $T_L$ if and only if $x$ is an ancestor of $y$ in $T$. Thus, $T_L$
   is now a rooted forest.
 \item Add a fresh root $\bot$ to $T_L$ and make all former roots of
   $T_L$ into children of $\bot$.
 \end{itemize}
 Again, $T_L$ can be computed (through its parent relation) in time
 $\Oh(|T|)$ by applying a depth-first search.

 Apply \cref{lem:navigation} to the tree $T_L$, thus constructing a
 suitable data structure $\Lb$ in time $\Oh(|T|)$. It can be now
 easily seen that whenever $x$ is an ancestor of $y$ in $T$, the
 answer to the query from the lemma statement is equal to
 $$\dir_{\Lb}(\anc(x),\anc(y)),$$
 where $\dir_{\Lb}(\cdot,\cdot)$ is a query to the data structure
 $\Lb$. In case $\anc(x)=\anc(y)$, we report $y$ provided $y\in L$,
 and otherwise we report that there is no node on the $x$-$y$ path
 that belongs to $L$. \end{proof}

 \input{app-warp.tex}

\paragraph*{Proof of \cref{thm:colcombet-ds} using deterministic Simon's factorization.}
We now give a proof that follows closely the combinatorial idea
presented by Colcombet in~\cite[Lemma~3]{Colcombet07}. See also the
work of Kazana and Segoufin~\cite{KazanaS13} that applies this idea in
the context of enumeration of $\MSO$ queries on trees. We need a few
definitions, which are taken directly from~\cite{Colcombet07}.

Fix a finite semigroup $S$. Let $P$ be an oriented path whose edges
are labeled by elements of $S$; let $\sigma\colon E(P)\to S$ be this
labeling. As $P$ is oriented, there is a natural order $<$ on $V(P)$
signifying which vertex is earlier. For two nodes $x<y$ of $P$, we
denote $\sigma(x,y)=\sigma(e_1)\cdot \ldots \cdot \sigma(e_\ell)$,
where $e_1,\ldots,e_\ell$ are the consecutive edges of the subpath of
$P$ from $x$ to $y$.

For $h\in \N$, a {\em{split of height $h$}} of $P$ is a labeling
$\lambda\colon V(P)\to [h]$. In such a split $\lambda$, two nodes
$x<y$ of~$P$ are called {\em{$i$-neighbors}} if
$\lambda(x)=\lambda(y)=i$ and $\lambda(z)\geq i$ for all $x<z<y$.
Clearly, being $i$-neighbors is an equivalence relation on
$\lambda^{-1}(i)$. Then the split $h$ is called {\em{forward Ramsey}}
if for all vertices $x,y,x',y'$ that are pairwise $i$-neighbors and
satisfy $x<y$ and $x'<y'$, we have
$$\sigma(x,y)\cdot \sigma(x',y')=\sigma(x,y).$$

Now let $T$ be a rooted tree with edges labeled by elements of $S$
through a labeling $\sigma\colon E(T)\to S$. Again, a split of height
$h$ of $T$ is a labeling $\lambda\colon V(T)\to [h]$. We say that such
a split is {\em{forward Ramsey}} if it is forward Ramsey when
restricting it to any root-to-leaf path in $T$. Here, the path is
oriented from the root to the leaf.

Colcombet's theorem states that in every tree there is always a
forward Ramsey split of height bounded by $|S|$. The following
statement follows easily from the results presented
in~\cite{Colcombet07}.

\begin{theorem}[\cite{Colcombet07}]\label{thm:colcombet-split}
  Let $T$ be a rooted tree with edges labeled by elements of a finite
  semigroup $S$. Then there is a forward Ramsey split of $T$ of height
  at most $|S|$. Moreover, such a split can be computed in time
  $|S|^{\Oh(1)}\cdot |T|$.
\end{theorem}

To be more precise, Colcombet states his result in terms of the
existence of a forward Ramsey split on a path that can be computed by
a deterministic automaton that processes the path in the given order.
As in~\cite{Colcombet07} and multiple other applications, the tree
variant stated above can be obtained by running the automaton on every
root-to-leaf path of the tree. For computing the split within the
promised time complexity, it suffices to construct the automaton in
time $|S|^{\Oh(1)}$ and run it on the root-to-leaf paths in~$T$ using
depth-first search.

With \cref{thm:colcombet-split} recalled, we can proceed to the proof
of \cref{lem:torso-ds}. Let $\lambda$ be a split of height
$h\coloneqq |S|$ for the input tree $T$, computed using the algorithm
of \cref{thm:colcombet-split}. With every $x\in V(T)$ we will store
the value $\lambda(x)$. Further, letting
$L_i\coloneqq \lambda^{-1}(i)$ for each $i\in [h]$, we apply
\cref{lem:partition navigation} to $T$ and the node subset $L_i$, thus
obtaining a suitable data structure $\Kb_i$. Finally, we apply
\cref{lem:navigation} to compute a suitable data structure $\Lb$ for
$\lca(\cdot,\cdot)$ and $\dir(\cdot,\cdot)$ queries. All of these take
time $|S|^{\Oh(1)}\cdot |T|$ in total.

Recall that within $\Kb_i$, with every $x\in V(T)$ we store a pointer
$\anc_i(x)$ to the closest (deepest) strict ancestor of $x$ that
belongs to $L_i$; if there is no such strict ancestor, $\bot$ is
stored instead. In case $\anc_i(x)$ is well-defined, alongside it we
also store the value \[\val_i(x)\coloneqq \sigma(\anc_i(x),x)\in S.\]

Again, using a depth-first search it is straightforward to compute all
these values in time $|S|^{\Oh(1)}\cdot |T|$ and using
$|S|^{\Oh(1)}\cdot |T|$ operations in $S$.

%

This concludes the definition of the data structure. Clearly, the time
complexity of the construction algorithm, as well as the space
complexity, are as promised. We are left with implementing the
$\sigma(x,y)$ queries within the promised time complexity.

Let $x$ and $y$ be a pair of nodes of $T$, where $x$ is a strict
ancestor of $y$. Let $R(x,y)$ be the set of all nodes that lie on the
unique path from $x$ to $y$ in $T$, including $x$ but excluding $y$.
Further, let
\[r(x,y)\coloneqq \min \{\lambda(z)\colon z\in R(x,y)\}.\]

We are going to give an algorithm that computes $\sigma(x,y)$ in time
$|S|^{\Oh(1)}$ and using $|S|^{\Oh(1)}$ operations in~$S$, and makes
at most one recursive call to compute $\sigma(x,z)$ for some $z$ that
is a strict ancestor of $y$ and a strict descendant of $x$. We will
have a guarantee that if such a call is invoked, then $r(x,z)>r(x,y)$.
Since the co-domain of $r(\cdot,\cdot)$ is $[h]$, this means that the
depth of the recursion is at most $h=|S|$, so the whole algorithm will
run within the promised complexity.

First, for every $i\in [h]$ we would like to test whether $R(x,y)$
contains a node $z$ with $\lambda(z)=i$, and if so, find the topmost
such node. This can be done using one call to the data structure
$\Kb_i$, applied to~$x$ and the parent of $y$. Using this procedure
for each $i\in [h]$, we can determine the value of $r(x,y)$ in
time~$\Oh(|S|)$, and also find the suitable node $z$: the topmost node
satisfying $\lambda(z)=r(x,y)$.

Next, we would like to determine whether there exists a node
$z'\in R(z,y)$, $z'\neq z$, such that $\lambda(z')=r(x,y)$, and if so,
to find the topmost such $z'$. Again, this can be done using one call
to the data structure $\Kb_{r(x,y)}$, applied to $z$ and the parent of
$y$.

Finally, let
\[z''\coloneqq \anc_{r(x,y)}(y).\]

With these definitions in place, it can be readily seen that
$\sigma(x,y)$ can be expressed as
\begin{equation}\label{eq:split}
  \sigma(x,y)=\sigma(x,z)\cdot \sigma(z,z')\cdot \sigma(z',z'')\cdot
  \sigma(z'',y), \end{equation} where the terms
\begin{itemize}[nosep]
 \item $\sigma(x,z)$ ought to be omitted if $x=z$;
 \item $\sigma(z,z')\cdot\sigma(z',z'')$ ought to be omitted if $z'$ does not exist (note that this implies that $z=z''$); and
 \item $\sigma(z',z'')$ ought to be omitted if $z'=z''$.
\end{itemize}

\smallskip Now observe that we can compute the consecutive factors on
the right hand side of~\eqref{eq:split} as follows:
\begin{itemize}
\item We have $\sigma(z'',y)=\val_i(y)$, and this value is stored
  directly in the data structure.
\item Similarly, provided $z'$ exists, $\sigma(z,z')=\val_i(z')$.
  Furthermore, if $z'\neq z''$, then $z,z',z''$ are all
  $r(x,y)$-neighbors on the root-to-$y$ path in $T$, hence
  $\sigma(z,z')\cdot \sigma(z',z'') = \sigma(z,z')$ by the
  forward-Ramseyanity of the split.
 \item Provided $x\neq z$, $\sigma(x,z)$ can be computed by a recursive call.
   Here, observe that by the definition of $z$, no node of $R(x,z)$ is
   assigned value $r(x,y)$ under $\lambda$, hence $r(x,z)>r(x,y)$, as
   was promised. \end{itemize}

 Therefore, using~\eqref{eq:split} we can compute $\sigma(x,y)$ using
 one call to $\sigma(x,z)$ with $r(x,z)>r(x,y)$ and, additionally,
 $|S|^{\Oh(1)}$ time and $|S|^{\Oh(1)}$ operations in $S$. This
 concludes the proof of \cref{thm:colcombet-ds}.

%% file: app-warp.tex

\paragraph*{Direct proof of \cref{thm:colcombet-ds}.} We start with a
direct proof of \cref{thm:colcombet-ds}. We will reuse the ideas underlying this proof in \cref{sec:automata}, when considering query enumeration and answering on graphs that exclude a topological minor.
The following proof generalizes an idea from~\cite[Sec. 2.2]{10.1007/978-3-642-02737-6_1} in the context of paths, to the case of trees.

For a finite set $X$, the
{\em{functional semigroup}} $F_X$ is the semigroup of all functions
$f\colon X\to X$ with composition as the semigroup operation. Here,
the composition of $f,g\in F_X$ is the function
$f;g\coloneqq \lambda a.g(f(a))$.
We chose to compose functions in this order for consistency with the
alternative proof of \cref{thm:colcombet-ds}.

We first prove \cref{thm:colcombet-ds} for functional semigroups. This is
captured by the following statement.

\begin{lemma}\label{lem:functional-warp}
  Fix a finite set $X$ and suppose we are given an $F_X$-labeled tree
  $T$. Then one can in time $|X|^{\Oh(1)}\cdot |T|$ construct a data
  structure that can answer the following query in time
  $|X|^{\Oh(1)}$: given nodes $x,y\in V(T)$, where $x$ is a strict
  ancestor of $y$, output $\sigma(x,y)$. The data structure uses
  $|X|^{\Oh(1)}\cdot |T|$ memory.
\end{lemma}

Note that $|F_X|=|X|^{|X|}$, so the complexity guarantees offered by
\cref{lem:functional-warp} for the special case of functional
semigroups are exponentially better than that promised by
\cref{thm:colcombet-ds}. Unfortunately, the semigroup of basic
bi-interface graphs is not a functional semigroup, so we cannot use
\cref{lem:functional-warp} directly in the proof of
\cref{lem:torso-ds}.

\begin{proof}[Proof of \cref{lem:functional-warp}]
  For a pair of nodes $v,w$ of $T$, where $v$ is an ancestor of $w$,
  write $f_{vw}\from X\to X$
  for
   the composition
  \[\sigma(v_1v_2);\sigma({v_1v_2});\ldots;\sigma({v_{m-1}v_m}),\]
  where $v_1\ldots v_m$ is the path starting at $v_1=v$ and ending at $v_m=w$
  (in particular if $v=w$ then $f_{vw}$ is the identity on $X$). Then $f_{vw}=\sigma(v,w)$ and $f_{vw}=\sigma(vw)$ when $v$ is the parent of $w$.
  Our goal is to compute in time $|X|^{\Oh(1)}\cdot |T|$
  a data structure that  given two nodes $v$ and $w$ of $T$, where $v$ is an ancestor of $w$, outputs $f_{vw}$ in time $|X|^{\Oh(1)}$.

  Let $s\coloneqq |X|$.  We first construct bijections \renewcommand{\xi}{c}
  $\xi_v\colon X\to [s]$ for all $v\in V(T)$ so that the following
  property is satisfied: whenever $v$ is the parent of $w$ in $T$, we
  have
 \begin{equation}\label{eq:beaver}
   c_w(f_{vw}(x)) \le \xi_v(x) \qquad\textrm{for all }x\in X.
 \end{equation}
We will say that $x\in X$ \emph{has color} $i$ at a node $v$ if $c_v(x)=i$.
Therefore, \eqref{eq:beaver} says that when going from a node $v$ to its child $w$,
the color of an element $x$ may either remain unchanged, or decrease,
when applying the function $f_{vw}$.
 This can be achieved in a top-down manner as described below, and depicted in \cref{fig:paths}.
 \begin{figure}[h!]
  \centering
   \includegraphics[page=2]{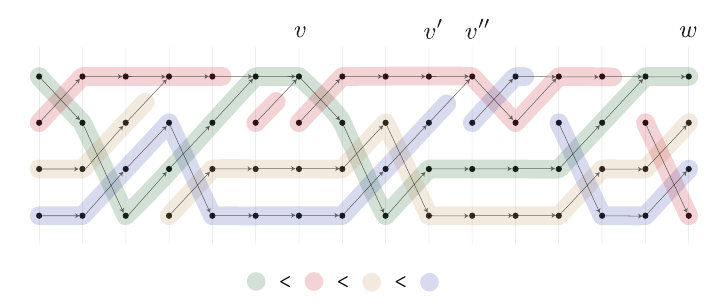}
   \caption{\textit{Construction of the colorings.} The
   figure depicts the colorings $\xi_v\from X\to[s]$ along a single branch $\pi$ in the tree $T$. Here $|X|=s=4$, and $[s]=\set{1,2,3,4}$ is represented by the four colors depicted above. The nodes along the branch $\pi$ correspond to the columns in the figure,
   with the left column being the root node and the right column being a leaf in $T$.
   The elements of $X$ correspond to the rows. The functions $f\from X\to X$ labeling the edges along the path $\pi$ are represented by the arrows between adjacent columns.
   The colorings $\xi_v$ are constructed from left to right, and are represented as colorings of each of the columns with the four colors.
The key property is that a path of any given color has a discontinuity
only when it hits a path with a smaller color.
   \textit{The querying phase.}
   Let $v$ and~$w$ be two columns with $v$ preceding $w$ (depicted in the figure),
   and suppose we want to compute the composition~$f_{vw}$ of the functions between $v$ and $w$, applied to the element that is blue at $v$, using the color coding from the figure.
   Let $v'$-$v''$ be the first pair of
   of consecutive nodes  $v'$-$v''$ following $v$ at which the blue path has a discontinuity (so $v''$ is \emph{blue-decreasing}).
   The blue element (at $v$) is mapped by~$f_{vw}$ to the blue element (at $w$)
   if $w$ is before $v''$. Otherwise, we jump from $v$ to the node $v'$, then
  map the blue element at $v'$ by the function labeling the edge $v'$-$v''$, obtaining a color that is smaller than blue (in this case, pink), and continue tracking the pink color from $v''$ to $w$. As each jump decreases the color, the number of jumps is bounded by $s-1=3$.
   }
   \label{fig:paths}
 \end{figure}

 At the root $r$,
 pick an arbitrary bijection $\xi_r\from X\to [s]$.  Next, if $\xi_v$
 is already defined for a node $v$, and $w$ is a child of $v$, then we
 define $\xi_w$ as follows:
\begin{itemize}
\item First, for each $y\in f_{vw}(X)$, set
 $$\xi_w(y)\coloneqq \min \left(\xi_v(f_{vw}^{-1}(y))\right),$$
 that is, the color of $y$ at $w$ is equal to the smallest color at $v$ of any element $x$ that is mapped to $y$ by $f_{vw}$, if it exists.
\item Next, extend $\xi_w$ to elements $x\notin f_{vw}(X)$
  arbitrarily, but so that $\xi_w$ becomes a bijection.
\end{itemize}

Property~\eqref{eq:beaver} follows immediately from the
construction. Clearly, the construction can be performed in time
$|X|^{\Oh(1)}\cdot |T|$.

For an index $i\in [s]$, we say that a node $w$ of $T$ is
{\em{$i$-decreasing}} if $w$ has a parent, say $v$, and there exists
$x\in X$ such that
 $$\xi_v(x)=i\qquad\textrm{and}\qquad \xi_w(f_{vw}(x))<i.$$
 That is, the function $f_{vw}$ maps the element that has color $i$ at $v$ to an element whose color becomes smaller at $w$.
 Let $L_i\subseteq V(T)$ be the set of $i$-decreasing nodes. Clearly,
 the sets $L_i$ can be computed in time $|X|^{\Oh(1)}\cdot |T|$.

 For each of the sets $L_i$, $i\in [s]$, we construct the data
 structure provided by \cref{lem:partition navigation}; call it
 $\Kb_i$. Also, we construct the data structure $\Lb$ provided by
 \cref{lem:navigation}.  This finishes the initialization of our data
 structure.  We now describe how it can be used to answer queries in
 time $|X|^{\Oh(1)}$.

 Suppose we are queried about the value of $f_{vw}$, where $v$ is
 an ancestor of $w$ in $T$. Clearly, it suffices to show how to
 compute $f_{vw}(x)$ for each $x\in X$ in time $|X|^{\Oh(1)}$.
 Fix $x\in X$, we compute $f_{vw}(x)$ as follows
 (see also the caption in \cref{fig:paths}).

 Let $i\coloneqq \xi_v(x)$ be the color of $x$ at $v$, and let $R(v,w)$ be the set of nodes
 lying on the $v$-$w$ path in $T$, excluding $v$ but including $w$.
 By a repeated use of~\eqref{eq:beaver} we have
 $c_w(f_{vw}(x))\leq i$, and the equality holds if and only if
 no node of $R(v,w)$ is $i$-decreasing, i.e., belongs to $L_i$.
 That is, the color of $f_{vw}(x)$ may only decrease or remain unchanged as we proceed along the $v$-$w$ path.
 We use the data structure $\Kb_i$ to find the topmost node
 of $L_i\cap R(v,w)$; this amounts to applying the query to $w$ and
 $\dir(v,w)$. If there is no such node, then we have
 $f_{vw}(x)=\xi_w^{-1}(i)$, that is, $f_{vw}(x)$ is the unique element
 that has color~$i$ at~$w$.
 This concludes the computation of
 $f_{vw}(x)$. Otherwise, let $v''$ be the node returned by
 $\Kb_i$ and let~$v'$ be its parent.
 Then
$$f_{vw}=f_{vv'};f_{v'v''};f_{v''w}.$$
We will first compute $x'\coloneqq f_{vv'}(x)$,
then  $x''\coloneqq f_{v'v''}(x')$, and finally
$y\coloneqq f_{v''w}(x'')$, which is equal to~$f_{vw}(x)$.
By definition of $v''$
we have that $c_{v'}(f_{vv'}(x))=i$, that is, $x'=f_{vv'}(x)$ has color $i$
at node $v'$.
Therefore, $x'$ can be computed in time $X^{\Oh(1)}$,
knowing $v''$ and hence $v'$, as well as the colorings.
Next, $x''\coloneqq f_{v'v''}(x')$ can be computed in time $X^{\Oh(1)}$,
since the function $f_{v'v''}$ is the label of the edge $v'v''$.
Finally, note that \[c_{v''}(x'')<i\] by definition of $v''$.
Hence, to compute $y= f_{v''w}(x'')=f_{vw}(x)$
it suffices to apply the procedure recursively to the nodes $v''$ and $w$, computing $y=f_{v''w}(x'')$. As
$\xi_{v''}(x'')<i,$
we conclude that with each recursive call, the value of $i$ strictly
decreases. Hence, the recursion depth is bounded by $s=|X|$, and the
whole algorithm runs in time $|X|^{\Oh(1)}$.
\end{proof}

We now use \cref{lem:functional-warp} to prove \cref{thm:colcombet-ds}
in full generality. This uses the fact that any semigroup~$S$ is a subsemigroup of some functional semigroup, namely $F_S$.

\begin{proof}[Proof of \cref{thm:colcombet-ds}]
  Let $S$ be the given fixed semigroup. We may assume that $S$ is a
  monoid, that is, there exists an element $\mathbf 1\in S$ such that
  $\mathbf 1\cdot m=m\cdot \mathbf 1=m$ for all $m\in S$. Indeed,
  every semigroup can be turned into a monoid by adding a fresh
  element $\mathbf 1$ satisfying the equations above.

  Consider the semigroup homomorphism $\eta\colon S\to F_S$ defined
  through right multiplication:
  \[\eta(a)\coloneqq \lambda b.(b\cdot a)\qquad \textrm{for }a\in S.\]
Then $\eta$ is injective, as $\eta(a)(\mathbf 1)=a$ for $a\in S$.

  Let $T'$ be the $F_S$-labeled tree obtained from $T$ by mapping
  each label through $h$. Clearly, $T'$ can be computed from $S$ in
  time $|S|^{\Oh(1)}\cdot |T|$ and using that many operations in
  $S$. Apply \cref{lem:functional-warp} to~$T'$, obtaining in time
  $|S|^{\Oh(1)}\cdot |T|$ a suitable data structure $\Db'$. Now, given
  nodes $x,y$ of $T$ where $x$ is a strict ancestor of $y$, to compute
  $\sigma(x,y)$ it suffices to use one query to $\Db'$ to get
  $\eta(\sigma(x,y))$, and then note that
  $\sigma(x,y)=\eta(\sigma(x,y))(\mathbf{1})$.
\end{proof}

%% file: tradeoff.tex

\section{Trading space usage for query complexity}\label{sec:tradeoff}

In this section we show how to modify the approach presented in
\cref{sec:general-queries} to construct a data structure that offers
queries in time $k^{\Oh(1)}$, at the cost of increasing the space
usage to quadratic in the size of the graph. The outcome is captured
in the following statement.

\begin{theorem}\label{thm:tradeoff}
  Given a graph $G$ and an integer $k$, one can in time
  $2^{\Oh(k\log k)}\cdot |G|^{\Oh(1)}$ construct a data structure
  that may answer the following queries in time $k^{\Oh(1)}$:
  given $u,v\in V(G)$ and a set $S$ consisting of at most $k$ vertices
  of~$G$, does $G\models \conn(u,v,S)$? The space usage of the data
  structure is $k^{\Oh(1)}\cdot |G|^2$.
\end{theorem}

Compared to \cref{thm:general-queries}, \cref{thm:tradeoff} offers a
much faster implementation of queries: the complexity is polynomial in
$k$, instead of doubly-exponential. The drawback is that while the
space usage is also polynomial in $k$, the dependency on the total
number of vertices becomes quadratic. Note that the initialization
time still depends exponentially on $k$.

The remainder of this section is devoted to the proof of
\cref{thm:tradeoff}. We will heavily rely on the approach presented in
\cref{sec:general-queries}, so we only discuss the modifications and
this section needs to be read in conjunction with
\cref{sec:general-queries}.

\paragraph*{The tree decomposition.} The main idea of the proof is to
rely only on weak unbreakability, instead of strong unbreakability as
in \cref{sec:general-queries}. This reduces the sizes of adhesions
from exponential in $k$ to just $k$. More precisely, given the input
graph $G$ and the parameter $k$, we apply
\cref{thm:weak-unbreakability} to compute, in time
$2^{\Oh(k\log k)}\cdot |G|^{\Oh(1)}$, a tree decomposition $\Tt$ with
the following properties:

\begin{itemize}
\item the adhesion of $\Tt$ is at most $k$; and
\item every bag of $\Tt$ is $(k,k)$-unbreakable in $G$.
\end{itemize}

Recall from the earlier discussion that we may assume that the
decomposition is regular. Obviously, the fact that $\Tt$ is only
weakly unbreakable renders the bulk of correctness argumentation
presented in \cref{sec:general-queries} invalid. Fixing this issue is
the main technical task of this section.

\paragraph*{Torso queries.} Note that even though now the adhesion of
$\Tt$ is bounded by $q=k$, if we just used \cref{lem:torso-ds} for
torso queries, we would still end up with exponential dependence on
$k$ in the query time. This can be mitigated by replacing the data
structure of \cref{lem:torso-ds} with a trivial one that just
memorizes all the answers. The drawback is that the space usage
becomes quadratic in $|G|$, as this is the total number of different
torso queries that can be asked.

\begin{lemma}\label{lem:torso-ds-trivial}
  One can compute in time $q^{\Oh(1)}\cdot |G|^{2}\|G\|$ a data
  structure that can answer the following query in time
  $q^{\Oh(1)}$: for given $x,y\in V(T)$, where $x$ is an ancestor of
  $y$, output the graph $\torso(x,y)$.
 The data structure takes space $q^{\Oh(1)}\cdot |T|^2$.
\end{lemma}
\begin{proof}
  Note that there are at most $\binom{|T|}{2}$ pairs $(x,y)\in V(T)^2$
  such that $x$ is a strict ancestor of $y$ in $T$, and for each of
  them $\torso(x,y)$ can be computed in time $q^{\Oh(1)}\cdot \|G\|$
  by running a breadth-first search from each vertex of
  $\adh(x)\cup \adh(y)$. Hence, all those graphs can be computed in
  time $q^{\Oh(1)}\cdot |G|^{2}\|G\|$ and stored in a two-dimensional
  table indexed by nodes of $T$, which takes space
  $q^{\Oh(1)}\cdot |T|^2$. Upon query, the data structure just returns
  the relevant graph $\torso(x,y)$.
\end{proof}

We remark that in what follows, we could still use \cref{lem:torso-ds}
instead of \cref{lem:torso-ds-trivial}. This would result in a variant
of \cref{thm:tradeoff} where the query time complexity would be
$2^{\Oh(k^2)}$ and the space usage would be $2^{\Oh(k^2)}\cdot \|G\|$.

\paragraph*{Dealing with weak unbreakability.} The data structure for
\cref{thm:tradeoff} is now the same as the one presented in
\cref{sec:general-queries}, except for the two modifications presented
above: we use \cref{thm:weak-unbreakability} instead of
\cref{thm:strong-unbreakability} for the computation of a tree
decomposition, and we use \cref{lem:torso-ds-trivial} instead of
\cref{lem:torso-ds} for torso queries.

It remains to show how, given this data structure, to implement a
$\conn(u,v,S)$ query, where $u,v\in V(G)$ and $S\subseteq V(G)$ is
such that $|S|\leq k$. The algorithm is actually {\em{exactly}} the
same as the one presented in \cref{sec:general-queries}. As argued
through \cref{lem:warp} and \cref{lem:aggregate}, this algorithm takes
time~$q^{\Oh(1)}$ and uses $\Oh(k)$ queries to a data structure
supporting torso queries. Since now the torso queries
take~$q^{\Oh(1)}$~time and $q=k$, it follows that the whole algorithm
runs in $k^{\Oh(1)}$ time. So we are left with arguing that the
outcome is still correct.

As in \cref{sec:general-queries}, let $X$ be the set of nodes of $T$
whose margins intersect $\wh{S}\coloneqq S\cup \{u,v\}$ (plus the root
of $T$), and let $Y$ be the least common ancestor closure of $X$.
Recall that the algorithm presented in \cref{sec:general-queries}
computes, for each $x\in Y$, the set
$\profile(x)\subseteq \binom{D(x)}{2}$, where
$D(x)=\adh(x)\cup (\{u,v\}\cap \cone(x))$, consisting of all pairs of
vertices of $D(x)$ that are connected in $G[\cone(x)]-S$.

The idea is that we relax the condition on what the algorithm computes
to the following over-approximation of $\profile(x)$. For a node
$x\in V(T)$, we say that $\Pi_x\subseteq \binom{D(x)}{2}$ is a
{\em{semi-correct profile}} at~$x$ if the following two conditions
hold for all $\{a,b\}\in \binom{D(x)}{2}$:

\begin{itemize}[nosep]
\item If $a$ and $b$ are connected in $G[\cone(x)]-S$, then
  $\{a,b\}\in \Pi_x$.
\item If $\{a,b\}\in \Pi_x$, then $a$ and $b$ are connected in $G-S$.
\end{itemize}

Thus, a semi-correct profile at $x$ necessarily contains all the pairs
that belong to $\profile(x)$, but it can contain some more pairs.
These additional pairs are not connected in $G[\cone(x)]-S$, but we
are certain that they are connected in the whole graph $G-S$.

The remainder of this section is devoted to arguing the following
claim.

\begin{claim}\label{cl:semi-correct-overall}
  Assuming $\Tt$ is weakly $(q,k)$-unbreakable and the algorithm of
  \cref{sec:general-queries} is run to answer a query $\conn(u,v,S)$,
  the following holds: whenever the algorithm computes the profile at
  any node $x$, this profile is semi-correct. \end{claim}

Note that for the root $r$ of $T$, we have $G[\cone(r)]=G$, so any
semi-correct profile at $r$ is equal to $\profile(r)$. Therefore,
\cref{cl:semi-correct-overall} implies that the outcome of the
algorithm is correct.
To prove \cref{cl:semi-correct-overall} it suffices to argue that the
assumption of semi-correctness can be pushed through the computation
presented in the proofs of \cref{lem:warp} and \cref{lem:aggregate}.
This we do in the next two lemmas.

\begin{lemma}\label{lem:warp-weak}
  Suppose $z$ is an ancestor of $y$ in $T$ such that $\wh{S}$ is
  disjoint with $\cmp(z)\setminus \cone(y)$, and $\Pi_y$ is a
  semi-correct profile for $y$. Let $J$ be a graph on vertex set
  $D(y)\cup \adh(z)$ whose edge set is the union of\,~$\Pi_y$ and the
  edge set of $\torso(z,y)$. Further, let
  $\Pi_z\subseteq \binom{D(z)}{2}$ consist of all pairs
  $\{a,b\}\in \binom{D(z)}{2}$ that are connected in $J-S$. Then
  $\Pi_z$ is a semi-correct profile for $z$.
\end{lemma}
\begin{proof}
  Suppose first that $\{a,b\}\in \binom{D(z)}{2}$ are connected in
  $G[\cone(z)]-S$. Let $P$ be a path witnessing this. Then $P$ can be
  uniquely written as concatenation of paths $P_1,P_2,\ldots,P_\ell$,
  where each path $P_i$ has endpoints in $V(J)$, but is otherwise
  vertex-disjoint with $V(J)$. Note that each path $P_i$ has all
  internal vertices either contained in $\cmp(z)\setminus \cone(y)$,
  or in $\cmp(y)$. Paths of the first kind correspond to edges with
  the same endpoints present in the edge set of $\torso(z,y)$, while
  paths of the second kind correspond to edges with the same endpoints
  present in $\Pi_y$, due to the semi-correctness of $\Pi_y$. It
  follows that $a$ and $b$ are connected by a path in $J$ that avoids
  $S$, so $\{a,b\}\in \Pi_z$ by the definition of~$\Pi_z$.

  Conversely, suppose $\{a,b\}\in \Pi_z$. Then in $J$ there is a path
  $P'$ that connects $a$ and $b$ and is disjoint with $S$. For every
  edge $e$ of $P'$, let $P_e$ be a path in $G-S$ defined as follows:
  \begin{itemize}[nosep]
  \item If $e$ originates from the edge set of $\torso(z,y)$, then let
    $P_e$ be the path in $G[\cone(z)\setminus \cmp(y)]$ with the same
    endpoints as $e$ that witnesses the existence of $e$ in
    $\torso(z,y)$. Note that $P_e$ is disjoint with $S$, because all
    its internal vertices belong to $\cmp(z)\setminus \cone(y)$.
  \item If $e\in \Pi_y$, then by semi-correctness of $\Pi_y$, there is
    a path $P_e$ in $G-S$ with the same endpoints as $e$.
  \end{itemize}
  By concatenating paths $\{P_e\colon e\in E(P')\}$ in order we obtain
  a walk connecting $a$ and $b$ in $G-S$.
\end{proof}

\cref{lem:warp-weak} proves that the procedure presented in the proof
of \cref{lem:warp} outputs a semi-correct profile for $z$, provided it
is supplied with a semi-correct profile for $y$.

\begin{lemma}\label{lem:aggregate-weak}
  Let $x\in Y$ and let $Z$ be the set of affected children of $x$.
  Suppose for each $z\in Z$ we have a semi-correct profile $\Pi_z$.
  Let $\bgraph'$ be a graph obtained from $G[\bag(x)]$ by adding the
  following edges:
  \begin{itemize}
  \item for each $z\in Z$, add all pairs of\, $\Pi_z$ as edges; and
  \item for each child $w$ of $x$ such that $w\notin Z$, add all pairs
    of $\binom{\adh(w)}{2}$ as edges.
  \end{itemize}
  Then for each pair of distinct vertices $a,b\in \bag(x)$, the
  following implications hold:
  \begin{itemize}
  \item If $a$ and $b$ are connected in $\bgraph'-S$, then $a$ and $b$
    are connected in $G-S$.
  \item If there exist $A,B\subseteq \bag(x)$, each of size $q+1$ and
    disjoint from $S$, such that $a\in A$, $b\in B$, $\bgraph'[A]$ is
    connected, and $\bgraph'[B]$ is connected, then $a$ and $b$ are
    connected in $G-S$.
 \end{itemize}
\end{lemma}
\begin{proof}
  For the first implication, observe that if $P'$ is a path connecting
  $a$ and $b$ in $\bgraph'-S$, then by the semi-correctness of
  profiles $\Pi_z$ and regularity of $\Tt$, every edge of $P'$ that is
  not present in $G[\bag(x)]-S$ can be replaced by a path in $G-S$
  connecting same endpoints. This yields a walk in $G-S$ connecting
  $a$ and $b$.

  For the second implication, the same argument as above shows that
  $A$ is contained in the vertex set of a single connected component
  of $G-S$, and the same goes for $B$. Since $A,B\subseteq \bag(x)$,
  $\bag(x)$ is $(q,k)$-unbreakable in $G$, and $|S|\leq k$, it follows
  that this must be the same connected component for both $A$ and $B$.
  So $a$ and $b$ belong to the same connected component of $G-S$.
\end{proof}

It is straightforward to verify that in the algorithm presented in
\cref{lem:aggregate}, the considered graph $\bgraph_S$ is exactly the
graph $\bgraph'-S$ from \cref{lem:aggregate-weak} (considered there
for $\Pi_z=\profile(z)$ for~$z\in Z$). Further, a pair
$\{a,b\}\in \binom{\bag(x)}{2}$ is deemed connected in $G[\cone(x)]-S$
if and only if one of the preconditions of the two implications
mentioned in \cref{lem:aggregate-weak} holds. Hence, if we run the
algorithm supplied with the semi-correct profiles
$\{\Pi_z\colon z\in Z\}$, then the outcome is
$\Pi_x\subseteq \binom{D(x)}{2}$ satisfying the following.
\begin{itemize}[nosep]
\item If distinct $a,b\in D(x)$ are connected in $G[\cone(x)]-S$, then
  $\{a,b\}\in \Pi_x$. This follows from the fact that the sets
  $\{\Pi_z\colon z\in Z\}$ are supersets of the profiles
  $\{\profile(z)\colon z\in Z\}$, so the edge set of the graph
  $\bgraph'-S$ is a superset of the edge set of the graph $\bgraph_S$
  considered in the proof of \cref{lem:aggregate}.
  Consequently, every pair $\{a,b\}\in \profile(x)$ is included in
  $\Pi_x$.
\item If $\{a,b\}\in \Pi_x$, then $a$ and $b$ are connected in $G-S$.
  This follows from the construction of $\Pi_x$ and
  \cref{lem:aggregate-weak}. (Semi-correctness of
  $\{\Pi_z\colon z\in Z\}$ needs also to be used for the corner cases
  when $\{a,b\}\nsubseteq \binom{\bag(x)}{2}$.)
\end{itemize}

In other words, the constructed set $\Pi_x$ is a semi-correct profile
for $x$. We conclude that if the algorithm presented in the proof of
\cref{lem:aggregate} is supplied with semi-correct profiles
$\{\Pi_z\colon z\in Z\}$, then it produces a semi-correct profile
$\Pi_x$ for $x$.

Now, \cref{cl:semi-correct-overall} follows by a straightforward
bottom-up induction using \cref{lem:warp-weak} and
\cref{lem:aggregate-weak}. This concludes the proof of
\cref{thm:tradeoff}.

%% file: tree-structures.tex

\section{Automata on augmented trees}\label{sec:automata}
We consider \emph{augmented trees}, that are, rooted trees in which the set of children of any given
node carries the structure of a graph, or more generally, a relational
structure (see \cref{fig:augmented}).
\begin{figure}[h!]
  \centering
   \includegraphics[page=4,scale=1.2]{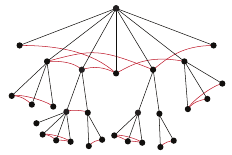}
   \caption{A tree augmented with graphs. The root is at the top, the edges of the tree
   are marked black, while the edges of the graphs are marked red,
   and they may only connect nodes that are siblings in the tree.
   As a first-order relational structure, the tree is equipped with
   the ancestor relation induced by the tree, as well as the red edge relation.
   }
   \label{fig:augmented}
 \end{figure}

Unranked, unordered trees are a special case of this, by
assuming the children form a set with no relations, while ordered
trees are a special case where the children form a total order. Here,
we allow the children to carry an arbitrary structure, but in our
application, this structure will be an edge- and vertex-colored
graph that excludes some fixed topological minor.  We also consider
automata that process such trees from the leaves towards the root in a deterministic fashion, by labeling each node by a state from a finite state space. To determine whether a node $v$ should be labeled with a state $q$,
 a first-order sentence $\delta_q$ is evaluated in the structure induced by the children of $v$.

Augmented trees can be viewed as relational structures equipped with
the ancestor relation $\preceq$ of the tree, as well as the relations $\Sigma$ relating the children of any given node, with which the tree is augmented.
Therefore, the logic $\FO(\preceq, \Sigma)$ is able to mix the two
structures, and express for example that every node whose children satisfy
a certain $\FO(\Sigma)$-definable property has an ancestor that satisfies another $\FO(\Sigma)$-definable property.

We use those automata as a tool for solving various computation problems related to the logic $\FO(\preceq, \Sigma)$.
Namely, the main results of this section are:
\begin{enumerate}
  \item every formula of $\FO(\preceq, \Sigma)$ can be converted into an equivalent automaton on augmented trees
  (\cref{thm:formulas to automata}),
  \item automata on augmented trees, and hence $\FO(\preceq, \Sigma)$ sentences, can be efficiently evaluated on trees augmented with structures from a class $\CC$, assuming that the model-checking of first-order logic can efficiently be solved on $\CC$
  (\cref{thm:formula evaluation,lem:automata evaluation}),
  \item automata on augmented trees, and hence $\FO(\preceq, \Sigma)$ formulas, can be efficiently queried on trees augmented with structures from a class $\CC$, assuming   first-order queries can be efficiently answered on $\CC$
  (\cref{thm:answering,lem:answering}),
  \item satisfying valuations to automata on augmented trees, and hence $\FO(\preceq, \Sigma)$ formulas, can be efficiently enumerated on trees augmented with structures from a class $\CC$, assuming first-order queries can be efficiently enumerated on $\CC$
   (\cref{thm:formula enumeration,lem:automata enumeration}).
\end{enumerate}
All the results proved in this section are generic and can be applied to trees augmented with various kinds of structures.
More generally, they are proved for the more
powerful logic $\FOMSO$ that
allows to access the tree structure using not only the relation $\preceq$, but indeed any $\MSO$-definable relation on trees labeled with labels from a finite alphabet $A$.
Our results generalize known results concerning the evaluation and enumeration of $\MSO$ queries on usual trees, e.g.~\cite{KazanaS13}.

In the next section, we use the results from this section
in the context of topological minor-free graph classes. Namely, we reduce the problem of evaluating/enumerating $\FOsep$ queries on such classes to the problem of evaluating/enumerating $\FOMSO$ queries on trees augmented with topological minor-free graph classes.
However,  the results obtained in this section may be of independent interest,
and potentially of broader applicability.

In the algorithmic statements proved in this section, we will assume that the class $\CC$ with which the trees are augmented
is a class for which $\FO$ model-checking,
query answering or query enumeration can be solved efficiently.
We remark here that all these assumptions are satisfied whenever $\CC$ is a topological minor-free graph class,
as will happen in our applications.
More generally, they are satisfied
for any class with bounded expansion, by results of~\cite{DBLP:journals/lmcs/KazanaS19} or~\cite{10.1145/3375395.3387660}.

\subsection{Augmented trees and their automata}

Fix a signature $\Sigma$ and an alphabet $A$. We assume that
$\Sigma\cap A=\emptyset$ and that $\Sigma$ does not contain the symbol
$\preceq$. Let $\CC$ be a class of $\Sigma$-structures. A
\emph{$\Cc$-augmented tree} over the alphabet $A$ is a rooted tree~$T$
where each node $v$ is labeled with a letter from $A$ and is associated with a
$\Sigma$-structure $\str A_v\in \CC$, called the \emph{whorl\footnote{
    In botany, a \emph{whorl} is an arrangement of similar parts,
e.g. branches, that radiate from a single point and surround the
    stem of a tree or plant.} at $v$}, whose elements are the children
of $v$ in $T$.
In some applications, it may be convenient to each node to be labeled with 
a \emph{set} of letters from $A$, rather than a single letter.
This scenario can be modeled using the above definition, 
by replacing the alphabet $A$ by $2^A$.

We view an augmented tree $T$ over the alphabet $A$ as relational
structure over the signature $\Sigma\cup\set\preceq\cup A$, as
follows:
\begin{itemize}
\item the domain consists of the nodes of $T$,
\item the unary relation $a$, for $a\in A$, contains a node $v$
  if and only if $v$ has label $a$,
\item $v\preceq w$ denotes that $v$ is an ancestor of $w$ in $T$
  (possibly $v=w$),
\item for each $R\in \Sigma$ of arity $t$, $R(v_1,\ldots,v_t)$ holds
  if and only if $v_1,\ldots,v_t$ have the same parent $v$, and
  $R(v_1,\ldots,v_t)$ holds in the whorl $\str A_v$ at $v$.
\end{itemize}

\paragraph{Automata.}
We now define automata that input augmented trees, and process them
from the leaves towards the root. As an automaton proceeds, it labels
each node by a state from a finite set $Q$ of states. To determine the
new state at a node $v$, assuming each child of $v$ already has an
assigned state, the automaton evaluates first-order sentences on the
whorl of $v$, which is now a $Q$-labeled $\Sigma$-structure (see \cref{fig:run}).

\begin{figure}[h!]
  \centering
   \includegraphics[page=1,scale=2]{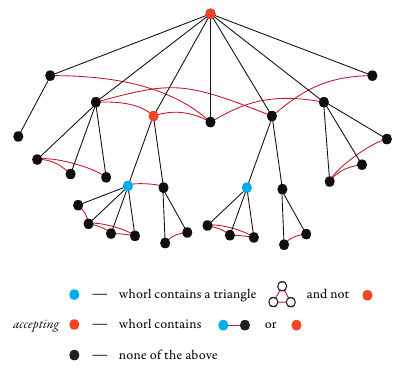}
   \caption{A run of an automaton on a tree augmented with graphs. The underlying tree is unlabeled, that is, over an alphabet with one element. The color of a node represents the state assigned in the run.
   The transition sentences are described above, and can be expressed using first-order sentences.
   The automaton accepts those augmented trees that contain two nodes $u$ and $v$ such that $u$ and $v$ are adjacent in the (red edge) graph relation, the whorl at $u$ contains a triangle, and the whorl at $v$ does not contain a triangle.
   }
   \label{fig:run}
 \end{figure}

\newcommand{\aut}[1]{{\mathcal #1}}

\pagebreak
Formally, an \emph{automaton $\aut A$ on augmented trees} consists of:
\begin{itemize}
\item an input \emph{alphabet} $A$: the automaton will process
  augmented trees over the alphabet $A$,
\item a finite set of \emph{states} $Q$,
\item a set of \emph{accepting states} $F\subset Q$,
\item for each state $q\in Q$ and letter $a\in A$, a first-order
  sentence $\delta_{q,a}$ in the signature $\Sigma\cup Q$, where each
  $q\in Q$ is viewed as a unary relation symbol, such that for every
  fixed $a\in A$, the sentences $(\delta_{q,a})_{q\in Q}$ are mutually
  inconsistent and their disjunction $\bigvee_{q\in Q}\delta_{q,a}$ is
  equivalent to true. The sentences $\delta_{q,a}$ are called the
  \emph{transition sentences} of $\aut A$.
\end{itemize}

A \emph{run} of $\aut A$ on an augmented tree $T$ is a labeling
$\rho\from V(T)\to Q$ such that for every node $v$ with letter~$a$ and
whorl $\str A_v$, the state $q=\rho(v)$ is the unique state such that
$\delta_{q,a}$ holds in the \mbox{$Q$-labeled} structure~$\str A_v$
with labeling $\rho$ restricted to $\str A_v$.  Formally, $\str A_v$
is viewed as a structure over the signature~\mbox{$\Sigma\cup Q$},
where a predicate $q\in Q$ holds on a vertex $v$ if and only if
$\rho(v)=q$.  The run is \emph{accepting} if it labels the root with
an accepting state. The automaton $\aut A$ \emph{accepts} an augmented
tree $T$ if it has an accepting run on it.

Note that there is exactly one run of $\aut A$ on any given augmented
tree $T$, by the assumption that the sentences
$(\delta_{q,a})_{q\in Q}$ are mutually inconsistent and their
disjunction is equivalent to true. It is undecidable if a given
collection of transition sentences satisfies the
requirements. However, this will not be a concern for us, as all the
automata constructed in the proofs below will satisfy the required
conditions by construction. Moreover, it is not difficult to come up
with an equivalent model with decidable syntax, by syntactically
enforcing that the transition sentences satisfy the required
conditions.


\subsection{From $\FOMSO$ to automata}
\medskip As we will see, automata will have the power to simulate first-order sentences over augmented trees,
that is, for every first-order sentence $\phi$ in the signature of augmented trees there is a automaton $\aut A_\phi$ on augmented trees that accepts an augmented tree $T$ if and only if $T$ satisfies $\phi$.
We also consider automata that process trees together
with valuations, in order to capture formulas $\phi(\tup x)$ with free variables. A tree $T$ over an alphabet $A$ together with a valuation $\tup a\from\tup x\to V(T)$ is modeled as a tree $T\otimes \tup a$ over an extended alphabet, defined as follows.

Let $\tup x$ be a set of variables.  Given an
alphabet $A$, construct a new alphabet
$A\otimes \tup x:=A\times 2^{\tup x}$.
That is, an element of this alphabet is a pair $(a,\tup y)$ consisting of a letter  $a\in A$ and a set of variables $\tup y\subset \tup x$.
Given a tree $T$ with labels
from~$A$ and a partial valuation (partial function)
$\tup a\from\tup x\partto V(T)$, define a new tree $T\otimes \tup a$
over the alphabet~$A\otimes \tup x$, in which a node $v$ has label
$(a,\tup y)$ if and only if $v$ has label $a$ in $T$ and
$\tup y=\setof{x\in\tup x}{\tup a(x)=v}$.  In particular, taking the
empty partial function $\emptyset\from \tup x\partto V(T)$, we obtain
the tree $T\otimes\emptyset$, in which every node $v$ that has label
$a$ in $T$ has label $(a,\emptyset)$ in $T\otimes\emptyset$.

Fix a finite alphabet $A$ and a set of variables $\tup x$.
An automaton \emph{with variables $\tup x$} on augmented trees over the alphabet $A$, denoted $\aut A(\tup x)$,
is an automaton $\aut A$ on augmented trees over the alphabet
$A\otimes \tup x$.  Given an augmented tree $T$ we may view $\aut A(\tup x)$ as a query that given a tuple $\tup a\in V(T)^{\tup x}$
determines whether or not $\aut A$ accepts the tree $T\otimes \tup a$. Let $\aut A(T)\subset V(T)^{\tup x}$ denote the set of valuations
$\tup a\from \tup x\to V(T)$ such that $\aut A$ accepts
$T\otimes \tup a$.

The following theorem states that this automata model is at least as
powerful as $\FOMSO$.

\begin{theorem}\label{thm:formulas to automata}
  For every formula $\phi(\tup x)$ of $\FOMSO$ there is an automaton
  $\aut A$ on augmented trees over the alphabet $A\otimes \tup x$ such
  that for every augmented tree $T$ over the alphabet $A$ we have
  \mbox{$\aut A(T)=\phi(T)$}.
  Moreover,
  $\aut A$ is computable from $\phi$.
\end{theorem}

  In the case when $\Sigma$ is empty, $\FOMSO$ boils down to
  $\MSO(\preceq,A)$. It is well known (see
  e.g.~\cite{loding2012basics}) that $\MSO$ on trees is captured by
  tree automata, where the automata may count the number of
  occurrences of any fixed state in the children, up to a certain
  threshold. This can be easily expressed using first-order sentences
  (in fact, boolean combinations of existential sentences) over the
  signature consisting only of unary predicates denoting the
  states. Hence, this yields a tree automaton~$\aut A$ on augmented
  trees for $\Sigma=\emptyset$. Therefore, the case when
  $\Sigma=\emptyset$ follows from known results.

\begin{proof}[Proof of Theorem~\ref{thm:formulas to automata}]
  We proceed by induction on the structure of the formula
  $\phi(\tup x)$.  Formally, in the induction we are considering not
  only the formula $\phi$ itself, but also a specified set $\tup x$ of
  variables containing the free variables of $\phi$. In particular, if
  $\tup x\subsetneq \tup z$, then $\phi(\tup x)$ and $\phi(\tup z)$
  are separate cases that we need to consider.  However, note that if
  we have an automaton $\aut A$ over the alphabet $A\otimes \tup x$
  for $\phi(\tup x)$ as in the statement and $\tup x\subset \tup z$,
  then the automaton $\aut A$ can be also viewed as an automaton over
  the alphabet $A\otimes \tup z$, and this automaton then satisfies
  the statement for $\phi(\tup z)$.

  In the inductive base, $\phi(\tup x)$ is an atomic formula: either
  $\mu(\tup x)$ for some formula $\mu(\tup x)\in\MSO(\preceq,A)$ or
  $R(x_1,\ldots,x_k)$ for some $R\in\Sigma$.  The first case is
  well-known, see the remark preceding the proof, and this yields an
  automaton on augmented trees in which the transition sentences do
  not use the relations from $\Sigma$.  We consider the second case.


  The automaton $\aut A$ for $R(x_1,\ldots,x_t)$, where $R\in \Sigma$,
  has states $q_{acc}$ and $q_{rej}$, where $q_{acc}$ is an accepting
  state and~$q_{rej}$ is a rejecting state. Observe that
  $R(a_1,\ldots, a_t)$ can only be true if all $a_i$ are elements of a
  fixed $\str A_v$, as the relation $R$ does not contain tuples with
  elements of different whorls.  Hence, the automaton can proceed from
  the leaves, which we initiate with $q_{rej}$, to the root as
  follows. For every node $v$ it checks if all the variables $x_i$ are
  interpreted in $\str A_v$ (recall that the alphabet is
  $A\otimes \tup x=A\times 2^{\bar x}$) and satisfy $R$. If this is
  the case, then the automaton changes state to $q_{acc}$ and
  propagates this state. Otherwise, it propagates the
  state~$q_{rej}$. This intuition is formalized by the following
  formulas
  \begin{align*}
  \delta_{q_{acc}, (a,U)} & =
  \exists y.\,q_{acc}(y)
  \vee
  \exists y_1,\ldots,y_t.\, \big(
  	R(y_1,\ldots,y_t)\wedge  \bigwedge\limits_{i\le t} \hspace{2mm}
  	\bigvee\limits_{\stackrel{W\subseteq 2^{\bar x}}{x_i\in W}} \hspace{2mm}
  	\bigvee_{a\in A}(a,W)\hspace{1mm}(y_i)
  \big).\\
  \delta_{q_{rej}, (a,U)} &= \neg \delta_{q_{acc}, (a,U)}.
  \end{align*}


  \medskip The case when $\phi(\tup x)$ is a disjunction of formulas
  $\alpha(\tup x)\lor\beta(\tup x)$ is handled by the usual product
  automaton construction. Namely, let $\aut A$ and $\aut B$ be
  automata obtained by inductive assumption applied to
  $\alpha(\tup x)$ and $\beta(\tup x)$. Construct the automaton
  $\aut A\times\aut B$ whose states are pairs $(p,q)$, where $p$ is a
  state of~$\aut A$ and $q$ is a state of~$\aut B$, and transition
  sentences are
  $\delta_{(p,q),a}\equiv\delta^{\aut A}_{p,a}\land \delta^{\aut
    B}_{q,a}$, where $\delta^{\aut A}_-$ and $\delta^{\aut B}_-$ are
  the transition sentences of $\aut A$ and $\aut B$, respectively. The
  accepting states of $\aut A\times \aut B$ are pairs $(p,q)$ such
  that~$p$ is accepting in $\aut A$ or $q$ is accepting in $\aut
  B$. Then the automaton $\aut A\times \aut B$ satisfies the condition
  in the statement, for the formula $\phi(\tup x)$.

  The case when $\phi$ is a negation of a formula is handled by the
  usual complementation of automata, since the automata are
  deterministic and have exactly one run on every given tree.  Hence,
  replacing the set $F$ of accepting states with its complement
  $Q\setminus F$, we obtain an automaton $\aut A'$ that accepts a tree
  $T\otimes \tup a$ if and only if $\aut A$ does not accept
  $T\otimes \tup a$.

  It remains to consider the case when $\phi$ is of the form
  $\exists x.\psi(x,\tup y)$, for some formula $\psi(x,\tup y)$. To
  simplify notation we assume that $\tup y$ is empty, so that $\phi$
  is of the form $\exists x.\psi(x)$. This is also without loss of
  generality, since a formula with free variables $\set{x}\cup \tup y$
  may be viewed as a formula with one free variable $x$, in the
  signature of augmented trees over the alphabet $A\otimes \tup y$.

  For an augmented tree $T$ and a node $v\in V(T)$, let $T\otimes v$
  denote $T\otimes \tup a$, where $\tup a$ is the valuation mapping
  $x$ to $v$.  By inductive assumption, there is an automaton $\aut A$
  that accepts an augmented tree~$T\otimes v$ if and only if $\psi(v)$
  holds in $T$.  Let $Q, F$ and $(\delta_q)_{q\in Q}$ be the
  components of $\aut A$.

  Given a tree $T$ and a node $v\in V(T)$ let $T_v$ denote the
  \emph{subtree} of $T$ rooted at $v$, consisting of $v$ and all the
  descendants of $v$, and for a valuation $\bar
  a$ 
  let $\state(\aut A, T\otimes \bar a)$ denote the state at the root of $T$
  in the unique run of $\aut A$ on~$T\otimes \bar a$.

  Given an augmented tree $T$ over the alphabet $A$, define a labeling
  $\rho'\from V(T)\to Q\times Q\times 2^Q\times 2^Q$, such that for
  every node $v\in V(T)$, the label $\rho'(v)$ is the tuple
  $(q,q',N,R)$ with:
  \begin{align*}
    q&=\state(\aut A,T_v\mathop{\otimes} \emptyset),&
                                                        q'&=\state(\aut A, T_v \mathop{\otimes} v),\\
    N&=\setof{\state(\aut A,T_v \mathop{\otimes} w)}{w\text{ is a child of $v$ in $T$}}&
                                                                                          R&=\setof{\state(\aut A, T_v \mathop{\otimes} w)}{w\in V(T), v\prec w}.
  \end{align*}

  \begin{lemma}\label{lem:powerset construction}
    There is an automaton $\aut A'$ with states
    $Q\times Q\times 2^Q\times 2^Q$ such that for every augmented
    tree~$T$ over the alphabet~$A$, the run of $\aut A'$ on $T$ is
    equal to $\rho'$.
  \end{lemma}
\begin{proof}
  Denote by $\rho'_i(v)$ the $i$th component of $\rho'(v)$, for
  $i=1,2,3,4$. Let $\delta_{q,(a,\emptyset)}$ and
  $\delta_{q,(a,\set x)}$, for $a\in A$, be the transition sentences
  of $\aut A$. Thus, for a node $v$ with letter $a$, the automaton
  $\aut A$ labels $v$ by $q$ if either~$v$ is not occupied by the
  variable $x$ and $\delta_{q,(a,\emptyset)}$ holds, or $v$ is
  occupied by $x$ and $\delta_{q,(a,\set x)}$ holds.

  Let $v$ be a node of $T$ with label $a$. Then $\rho'(v)$ is equal to
  the unique tuple $(q,q',N,R)$ such that the following conditions
  hold:

  \begin{itemize}
  \item $\delta_{q,(a,\emptyset)}$ holds in the whorl at $v$ labeled
    by $\rho'_1$,
  \item $\delta_{q',(a, \set x)}$ holds in the whorl at $v$ labeled
    by $\rho'_1$,
  \item $N$ is the set of those states $p\in Q$ such that there is
    some child $w$ of $v$ such that $\delta_{p,(a, \emptyset)}$ holds
    in the whorl at $v$ labeled by $\rho_1'$, except for $w$, which is
    labeled~$\rho_2'(w)$.
  \item $R$ is the set of those states $p\in Q$ such that $p\in N$ or
    there is some child $w$ of $v$ and state $r\in \rho_4'(w)$ such
    that $\delta_{p,(a, \emptyset)}$ holds in the whorl at $v$
    labeled by $\rho_1'$, except for $w$ which is labeled~$r$.
    \end{itemize}

    The above conditions can be expressed by a first-order sentence
    $\psi_{(q,q',N,R),a}$ that is evaluated in the whorl at $v$
    labeled by $\rho'$.

    Then the automaton $\aut A'$ over the alphabet $A\otimes \set x$
    with states $Q\times Q\times 2^Q\times 2^Q$ and
    sentences~$\psi_{(q,q',N,R),a}$, for ${q,q'\in Q}$ and $N,R\subset Q$
    and $a\in A$, satisfies the statement in the lemma.
  \end{proof}

  In particular, by letting $\aut A'$ be as in the lemma, and
  equipping it with accepting states of the form~$(q,q',N,R)$, where
  $q,q'\in Q$ and $N,R\subset Q$ are such that $R\cup \set{q'}$
  intersects $F$, we obtain an automaton $\aut A'$ such that $\aut A'$
  accepts $T$ if and only if there is some node $w$ in $T$ such that
  $T\otimes w$ is accepted by $\aut A$.  By definition of $\aut A$,
  this is equivalent to saying that $T$ satisfies $\exists x.\psi(x)$,
  or that $T$ satisfies $\phi$.

  \medskip This concludes the inductive proof of the theorem.
\end{proof}


In the following sections, we  show how to efficiently evaluate and enumerate $\FOMSO$-queries
on $\CC$-augmented trees.

\subsection{Evaluation}
First, we show how to evaluate automata on
augmented trees.  By \emph{evaluating} an automaton $\aut A$ on an
augmented tree $T$ we mean the problem of computing the (unique) run
$\rho\from V(T)\to Q$ of $\aut A$ on $T$, where $Q$ is the set of
states of $\aut A$.

Let $\mathcal L$ be a logic. We say that labeled model-checking for $\mathcal L$ on a class $\Cc$ of $\Sigma$-structures is
\emph{efficient} if there is a computable function $f\from\N\to\N$ and
an algorithm that, given a $\Sigma\cup A$-sentence $\delta$
(where the elements of $A$ are viewed as unary predicates)
and a structure
$\str A\in \Cc$ labeled with elements of $A$, determines whether it satisfies $\delta$ in time
$f(|\delta|)\cdot |\str A|$. We say that evaluation of automata on
$\Cc$-augmented trees is \emph{efficient} if there is a computable
function $g\from\N\to\N$ and an algorithm that, given an automaton on
augmented trees and a $\Cc$-augmented tree $T$ over an alphabet $A$,
evaluates~$\aut A$ on $T$ in time $f(|\aut A|)\cdot |T|$. Note that
for a $\Cc$-augmented tree $T$ we have
$|T|=\sum_{v\in V(T)} |\str A_v|+1$.

\begin{lemma}\label{lem:automata evaluation}
  Let $\CC$ be a class of\, $\Sigma$-structures such that labeled model-checking for first-order logic is
  efficient on $\CC$.
   Then automata evaluation on $\Cc$-augmented
  trees is efficient.
\end{lemma}
\begin{proof}Let $Q$ be the set of states of $\aut A$, and  $f\from\N \to\N$ be as in the definition of efficiency of labeled model-checking on $\CC$.

  Given a $\CC$-augmented tree $T$,
    compute the run $\rho\from V(T)\to Q$ of
  $\aut A$ on $T$ by processing the nodes~$v$ of~$T$ from the leaves
  towards the root. Namely, if $v$ has label $a$, evaluate each of the
  transition sentences~$\delta_{q,a}$ (for $q\in Q$) in the whorl
  $\str A_v\in \Cc$ at $v$, labeled by $\rho$ restricted to
  $\str A_v$. Then exactly one of the sentences $\delta_{q,a}$ holds,
  and set~$\rho(v)$ to $q$.

  The total time spent on evaluating the sentences $\delta_{q,a}$ is
  bounded by
  $\sum_{q\in Q, a\in A, v \in V(T)} f(|\delta_{q,a}|)\cdot |\str
  A_v|$, which is bounded by $g(|\aut A|)\cdot |T|$ for a computable
  function $g$ derived from $f$.
\end{proof}

As a corollary, by \cref{thm:formulas to automata}, we get:
\begin{theorem}\label{thm:formula evaluation}
  Let $\CC$ be a class of\, $\Sigma$-structures such that labeled model-checking for first-order logic is
  efficient on $\CC$.
    Then model-checking for \mbox{$\FOMSO$} is
  efficient on $\CC$-augmented trees with labels from a finite alphabet $A$.
\end{theorem}

\subsection{Query answering}
In the query-answering problem, a query $\phi(\tup x)$ and an augmented tree $T$ are given, and the task is to compute
in time linear in $|T|$ a data structure that allows to answer in constant time queries of the form: Does a given tuple $\tup a\in V(T)^{\tup x}$ satisfy $\phi(\tup x)$?
This boils down to the query evaluation problem in case when $\phi$ has no free variables. In this section, we show how to efficiently solve the query-answering problem on $\CC$-augmented trees.

Say that labeled query-answering is \emph{efficient}
for a logic $\mathcal L$
on a class of $\Sigma$-structures $\CC$
if there is a computable function $f\from \N\to\N$ and an algorithm
that, given a structure $\str A\in\CC$ vertex-labeled
by elements of a finite alphabet $A$ and a $\cal L$-formula $\phi(\tup x)$ computes in time $f(|\phi|)\cdot |\str A|$ a data structure that allows to answer in constant time queries of the form: Does a given tuple $\tup a\in V(\str A)^{\tup x}$ satisfy $\phi(\tup x)$?

Our goal now is to prove the following.
\begin{theorem}\label{thm:answering}
  Let $\CC$ be a class of $\Sigma$-structures for which labeled query-answering of first-order queries is efficient.
  Then labeled query-answering of $\FOMSO$-queries is efficient
  over $\CC$-augmented trees vertex-labeled with $A$.
\end{theorem}

As previously, we start with reformulating our task
in the language of automata. Fix a signature~$\Sigma$, a class $\CC$ of $\Sigma$-structures, a finite alphabet $A$ and a tuple of variables $\tup x$. Recall that given an automaton~$\aut A(\tup x)$
with variables $\tup x$  and a $\CC$-augmented tree $T$ labeled with $A$, we may consider the query:
Is a given tuple $\tup a\in V(T)^{\tup x}$ accepted by $\aut A(\tup x)$? Equivalently, does $\aut A$ accept $T\otimes \tup a$?
Thus we view $\aut A(\tup x)$ as a query.

\begin{lemma}\label{lem:answering}
  Let $\CC$ be a class of $\Sigma$-structures for which labeled query-answering is efficient.
  Then labeled query-answering of automata on $\CC$-augmented trees is efficient.
\end{lemma}
\cref{lem:answering} together with \cref{thm:formulas to automata} yields \cref{thm:answering}, so it remains to prove the lemma.

\paragraph{Proof sketch}
We first give a brief sketch of the proof.
The general idea is as follows.
Given a tree $T$, first compute the run $\rho_0$ of $\aut A$ over $T\otimes \emptyset$, in time linear in $|T|$.
Now, given a tuple $\tup a$, the run $\rho$ of $\aut A$ over $T\otimes \tup a$ will be obtained by updating $\rho_0$
in a few of nodes -- roughly the nodes
occurring in $\tup a$ and least common ancestors of any two such nodes. The number of nodes at which we recompute the new states is linear in $|\tup x|$, which is a constant with respect to $|T|$.

More precisely,
we want to compute $\state(\aut A,T\otimes \tup a)$ in constant time. Let $X$ be the set of nodes in that appear in the tuple $\tup a$;
then $|X|\le |\tup x|$. Let $Y$ be the closure of
$X$ under least common ancestors, and also add the root of $T$ to $Y$; then $|Y|\le 2|X|$.
Furthermore, let $Y'\supseteq Y$ be obtained from $Y$ by adding
to $Y$ all children of elements from $Y$ in the directions
of elements of $X$, that is, for every $y\in Y$ and its descendant $x\in X$, the child $\dir(y,x)$ of $y$ in the direction of $x$ is added to $Y'$ (unless it already belongs to $Y$).
Then $|Y'|\le 4|X|$.

\begin{figure}[h!]
  \centering
   \includegraphics[page=6,scale=0.7]{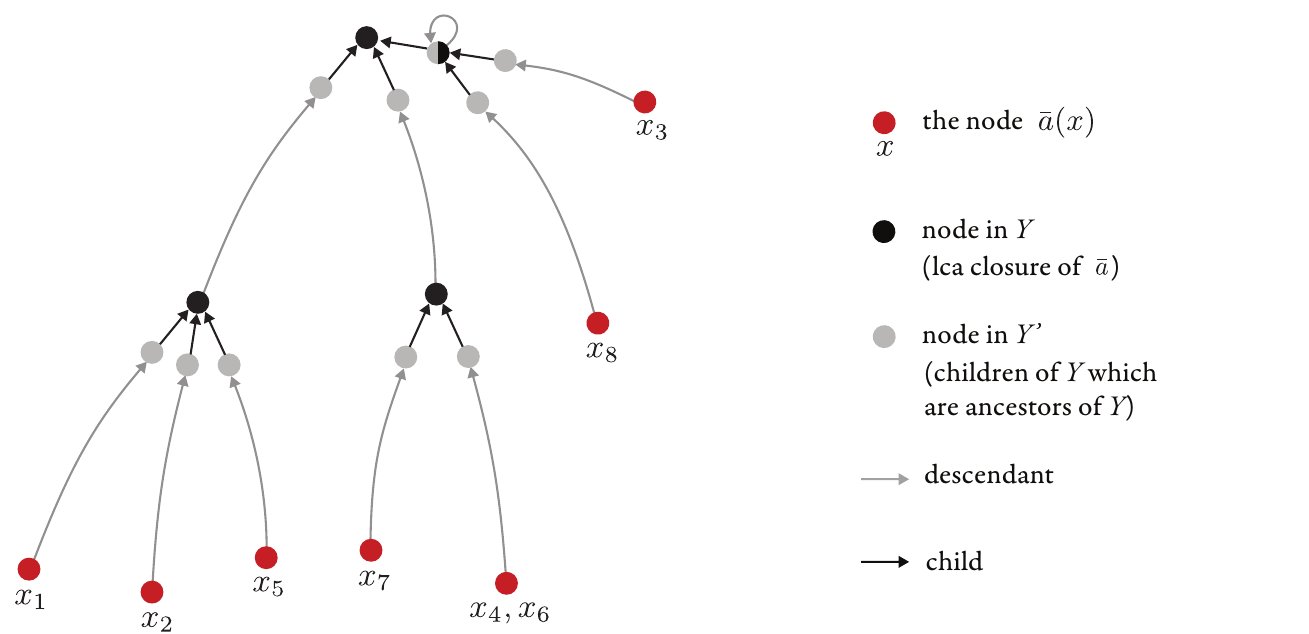}
   \caption{The states of the new run $\rho$ on $T\otimes \tup a$ are recomputed only at the nodes marked above, basing on the run $\rho_0\otimes \emptyset$. Here $\tup x=\set{x_1,\ldots, x_8}$ and $\tup a\in V(T)^{\tup x}$.
   Note that each node in the figure above may have many children that are not affected by the recomputation, and therefore are omitted in the figure.
   At each node $v\in Y$ (black and red nodes) the new state $\rho(v)$ is inferred from the updated states in the children of $v$ that have some descendant in $Y$, that is in the gray children of $v$.
   At each node $v'\in Y'$ (gray nodes) the new state $\rho(v')$ is inferred from the updated state at its closest descendant in~$Y$ (of $v'$ itself if $v'\in Y$).
   }
   \label{fig:jumps}
 \end{figure}

Our constant-time query-answering algorithm will
recompute the states of the new run $\rho$ of $\aut A$ on $T\otimes \tup a$ in each of the nodes
in $Y'$, in a bottom-up fashion (see \cref{fig:jumps}). For each lowest node $v$
in~$Y'$
(those nodes occur in $\tup a$),
the new state $\rho(v)$ in $v$ is easily computable in constant time from the
state~$\rho_0(v)$ and the set of variables of $\tup x$ that are
mapped to that node by the valuation $\tup a$.
More generally, for each node $v\in Y$, the new state $\rho(v)$ in $v$ is inferred from the following data:
\begin{itemize}
  \item the old state
  $\rho_0(v)$,
  \item the set of variables of $\tup x$ that are
  mapped to that node by the valuation $\tup a$ (possibly empty),
  \item the new state $\rho(v')$ at each child $v'$ of $v$ that belongs to $Y'$.
\end{itemize}
This inference is achieved by querying a first-order formula
in the whorl $\str A_v$ at $v$ that asks:
Which of the transition sentences holds in $\str A_v$
labeled by $\rho_0$, if we only modify the labels
in the elements of~$\str A_v$ that belong to $Y'$?
This can be done in constant time, assuming we have initialized an appropriate query-answering data structure for the structure $\str A_v\in\CC$.

Finally, for each node $v'\in Y'$, the new state in $v'$ is inferred from the state in the closest descendant of $v'$ that belongs to $Y$ (possibly $v'$ itself). The observation is that in the subtree of $T$ rooted at $v'$, the new run $\rho$  and $\rho_0$ only differ on those nodes $w$ that are ancestors of $v$. In other words,  for each node
that is a sibling of a node on the path, the runs $\rho$ and $\rho_0$ are identical. Therefore, recomputing the new state at $v'$ amounts to recomputing a number of transitions along the path from $v$ to $v'$, or equivalently, computing the composition of a sequence of functions mapping states to states along that path. This can be achieved in constant time using \cref{thm:colcombet-ds}.

\pagebreak
We start with introducing some auxiliary notions.

\paragraph{Entailment}
We now describe a subroutine of our algorithm used to propagate states along a single path. In the terminology of the description above, this subroutine will be used to compute the new state at a node $v'\in Y'$, having already computed the state at its closest descendant $v\in Y$. Referring to \cref{fig:jumps}, this will allows us to propagate the updated states along the gray edges.

Fix an augmented tree $T$ over the alphabet $A$ and let $\rho_0$ be the run of $\aut A$ on $T\otimes \emptyset$.
For a node $w\in V(T)$ and state $p\in Q$, let $\rho_0[w\mapsto p]$
be the labeling that assigns states from $Q$ to nodes of $T$ such that:
\begin{itemize}
\item $\rho_0[w\mapsto p]$ agrees with $\rho_0$ on nodes that are not ancestors of $w$,
\item $\rho_0[w\mapsto p]$ maps $w$ to $p$, and
\item $\rho_0[w\mapsto p]$ inductively maps a strict ancestor
$w'$ of $w$
to a state $q$ such that
$\delta_{q,(a,\emptyset)}$ holds in the whorl of $w'$ labeled by $\rho_0[w\mapsto p]$.
\end{itemize}

In other words, $\rho_0[w\mapsto p]$ is obtained from $\rho_0$ by enforcing the state at $w$ to be $p$,
and otherwise continuing with the bottom-up evaluation of $\aut A$ on $T\otimes \emptyset$. If $w'$ is a node of $T$, then write \mbox{$[w\mapsto p]\vdash [w'\mapsto p']$} if $\rho_0[w\mapsto p]$ labels $w'$ with $p'$.

By definition we get the following transitivity property of the entailment $\vdash$.
\begin{lemma}\label{lem:vdash transitivity}
Let $w,w',w''$ be nodes in $T$ such that $w$ is a descendant of $w'$ and $w'$ is a descendant of $w''$. Let $p,p',p''$ be states.
If $[w\mapsto p]\vdash [w'\mapsto p']$
and $[w'\mapsto p']\vdash [w''\mapsto p'']$, then
$[w\mapsto p]\vdash [w''\mapsto p'']$.
\end{lemma}

This, together with \cref{thm:colcombet-ds}, yields the following algorithmic consequence.
\begin{lemma}\label{lem:colcombet-rephrased}
There is a computable function $f\from\N\to\N$ and an algorithm that,
given a $\CC$-augmented tree~$T$, compute in time $f(|Q|)\cdot |T|$
a data structure that can answer the following queries:
given two nodes~$w,w'$ of\, $T$,
where $w'$ is an ancestor of $w$,
and a state $p\in Q$,
return $p'\in Q$ such that
$[w\mapsto p]\vdash [w'\mapsto p']$.
\end{lemma}
\begin{proof}
Fix an augmented tree $T$ over the alphabet $A$.

For each pair of nodes $w,w'$ of $T$, where $w'$ is an ancestor of $w$, let $f_{ww'}\from Q\to Q$ be the function
such that $f_{ww'}(q)=q'$ if $[w\mapsto q]\vdash [w'\mapsto q']$.
From \cref{lem:vdash transitivity} it follows that if
if $w_1,\ldots, w_n$ is the path from $w_1=w$ to its ancestor $w_n=w'$ then \[f_{ww'}=f_{w_1w_2};f_{w_2w_3};\ldots;f_{w_{n-1}w_n}.\]

The algorithm is as follows.
First, compute  $\rho_0$  from $T$ in time linear in $|T|$, using \cref{lem:automata evaluation}.
Then apply \cref{thm:colcombet-ds} to
the tree $T$ in which every edge $w'w$ is labeled by
the function $f_{ww'}$, treated as an element fo the semigroup $Q^Q$ of functions from $Q$ to $Q$.
This yields in time linear in $|T|$ a data structure that
allows to query the function $f_{ww'}\from Q\to Q$, for any nodes $w'$ and $w$ such that $w'$ is an ancestor of~$w$.
In particular, $f_{w'w}(p)$ is the sought state $p'$
such that $[w\mapsto p]\vdash[w'\mapsto p']$.
\end{proof}

\paragraph{Shapes}
  Fix an automaton $\aut A(\tup x)$ with variables $\tup x$,
  that is  an automaton $\aut A$ over the alphabet $A\otimes \tup x$. Let $Q$ be its set of states.
A \emph{shape} $\sigma$ consists of the following (see \cref{fig:shape}, left):

\begin{figure}[h!]
  \centering
   \includegraphics[page=7,scale=0.5]{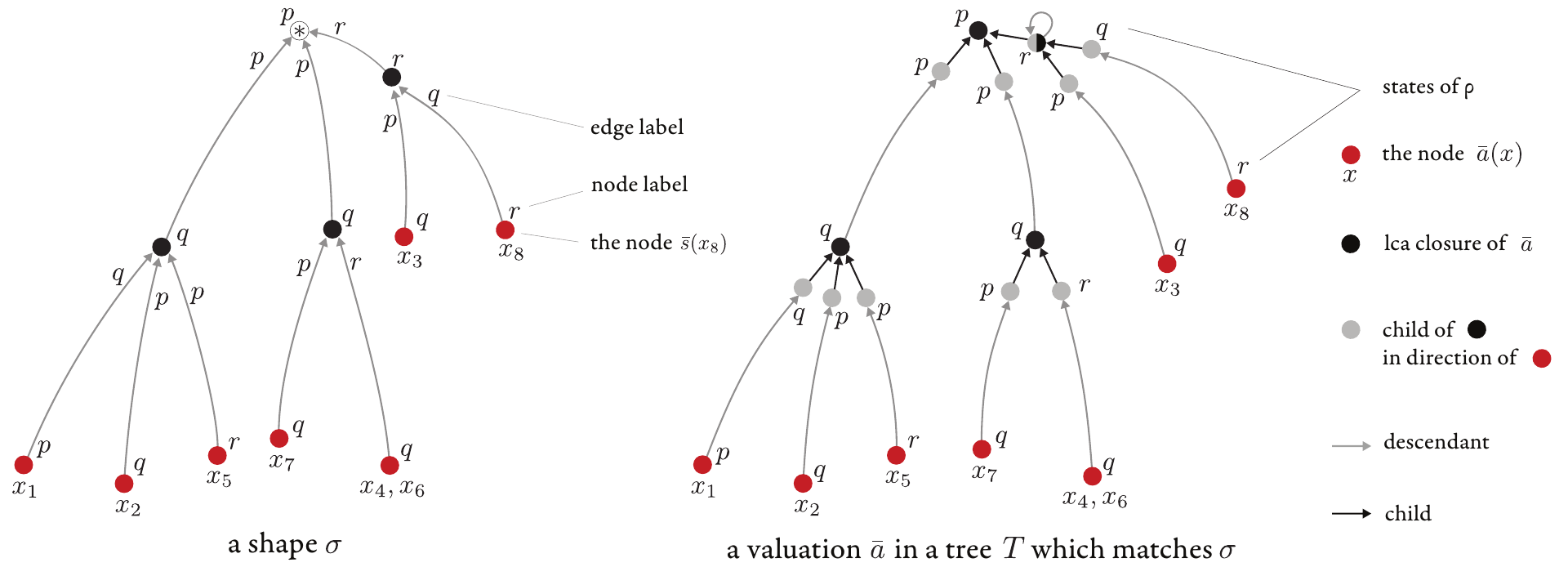}
   \caption{\emph{Left:} A shape with variables $\tup y=\set{x_1,\ldots,x_8}$ and states $\set{p,q,r}$. The nodes belonging to $\tup s$ are marked red, the remaining ones are black. \emph{Right:} A valuation $\tup a$ in a tree $T$ that matches the shape $\sigma$ to the left. The nodes belonging to $\tup a$ are marked red, their least common ancestors are marked black, and the children of the least common ancestors in the direction of nodes in $\tup a$ are marked gray. Each node is labeled by the run  $\rho$ of $\aut A(\tup x)$ on $T\otimes \tup a$.
   }
   \label{fig:shape}
 \end{figure}
\begin{itemize}
  \item a set of variables $\tup y$ contained in $\tup x$,
  \item a rooted tree $S$, with the root denoted $\ast$;
  \item a labeling $L\colon V(S)\cup E(S) \to Q$ assigning states from $Q$ to vertices and edges of $S$,
  \item a valuation $\tup s\from \tup y\to V(S)$
  such that every node in $S$ is either in the image of $\tup s$, or is the least common ancestor of some two nodes in the image of $\tup s$.
\end{itemize}
We may write $V(\sigma)$ for $V(S)$. We also say that a node $v$ of $S$ is \emph{occupied} by a variable $x\in \tup y$ if $\tup s(x)=v$.

\medskip

Thanks to the last condition above $S$ has at most $2|\tup y|-1$ nodes, and we obtain the following.

\begin{lemma}\label{lem:shapes}Fix $\tup y\subset \tup x$.
    The set of shapes with variables $\tup y$ is finite, and can be effectively computed, given~$\aut A$.
\end{lemma}

\pagebreak
Fix a tree $T$ over the alphabet $A$, a shape $\sigma$ as above, a valuation $\tup a\from \tup y\to V(T)$,
and let $\rho\from V(T)\to Q$ be the (unique) run of $\aut A$ on $T\otimes \tup a$.
The valuation $\tup a$ \emph{matches} $\sigma$ with \emph{starting node} $v\in V(T)$ if there is a function $F\from V(S)\to V(T)$ such that (see \cref{fig:shape}, right):
\begin{enumerate}
  \item $F$ preserves least common ancestors, that is, for all $a,b,c\in V(S)$, if $c$
  is the least common ancestor of $a$ and $b$ in $S$ then $F(c)$ is the least common ancestor of $F(a)$ and $F(b)$ in $T$. In particular, $F$ preserves $\preceq$.
  \item $F$ is injective,
  \item  $F(\tup s)=\tup a$, that is, $F(\tup s(x))=\tup a(x)$ for every $x\in \tup y$,
  \item $F(\ast)=v$, where $\ast$ is the root of $S$,

  \item for every vertex $w$ of $S$, $\rho(F(w))= L(w)$,
  \item for every edge $e=ww'$ of $S$ between a node $w$ and its child $w'$, the
  $F(w')$-directed child of $F(w)$ in $T$ satisfies $\rho(v')=L(e)$.

\end{enumerate}

Note that if $\bar a$ matches $\sigma$, then the function
$F\from V(\sigma)\to V(T)$ witnessing this is unique. Moreover, no valuation~$\tup a$ matches two distinct shapes,
since the rooted tree $S$ and the labeling $L$ of its nodes and edges is uniquely determined by the valuation $\tup a$ and the unique run of $\aut A$ on $T\otimes \tup a$.

\paragraph{Efficient shape-matching}
We will prove the following lemma stating that shape matching can be efficiently tested.
Below, $\CC$ is a class of $\Sigma$-structures for which labeled query-answering is efficient, and $\aut A(\tup x)$ is an automaton with variables $\tup x$ over $A$-labeled $\CC$-augmented trees.

\begin{lemma}\label{lem:shape answering}
  There is a computable function $f\from\N\to\N$ such that the following holds.
  For every shape~$\sigma(\tup y)$ there is an algorithm
  that given a $\CC$-augmented tree $T$ computes in time $f(|\aut A|)\cdot |T|$ a data structure that allows to answer the following queries in time $f(|\aut A|)$:
  Given a tuple $\tup a\in V(T)^{\tup y}$,
  does it match $\sigma$?
\end{lemma}

First we show how \cref{lem:shape answering} implies \cref{lem:answering}.
\begin{proof}[Proof of \cref{lem:answering}]
  Let $T$ be a $\CC$-augmented tree over the alphabet $A$.
  We extend $\tup x$ to $\tup x\cup\set{z}$ by adding a dummy variable~$z$. The intention is that this variable will occupy the root of a shape, and on the other hand, we will want to match it with the root of a given tree $T$.

  We may view
  $\aut A$ as an automaton over the alphabet $A\otimes (\tup x\cup \{z\})$. Using~\cref{lem:shapes}, compute the set~$\Delta$ of all shapes $\sigma(\tup x\cup\set z)$ in which the root node of is occupied by the variable
  $z$ (that is, $\tup s(z)=\ast$) and
  is labeled with an accepting state of $\aut A$ (that is, $L(\ast)$ is an accepting state).

  In the initialization phase, our algorithm
  computes for each shape $\sigma\in \Delta$
  the data structure given by \cref{lem:shape answering},
    in time linear in $|T|$.

  In the query-answering phase, we proceed as follows.
  Let  $\tup a\in V(T)^{\tup x}$ be a tuple given on input.
Extend $\tup a$ to the tuple $\tup ar\in V(T)^{\tup x\cup \set z}$
that maps the variable $z$ to the root $r$ of $T$.
Then for each shape $\sigma(\tup x\cup \set z)\in \Delta$, query the appropriate data structure to determine whether
or not $\tup ar$ matches $\sigma(\tup x\cup \set z)$.
If this is the case for some shape $\sigma\in\Delta$,
then $\aut A$ accepts $T\otimes \tup a$,
otherwise it doesn't.
\end{proof}


\medskip

We now prove \cref{lem:shape answering}.
\begin{proof}[Proof of \cref{lem:shape answering}]
The proof proceeds by induction on the number of nodes $|S|$ of the tree $S$ underlying the shape $\sigma$.

In the base case, $S$ has only one node $\ast$, and all the variables $\tup y$ of $\sigma$ are associated with that node.
Suppose that $\aut A$ has states $Q$.
Construct an automaton $\aut A'$ over the alphabet $A$ with states $Q^2$ and with the following property:
  $\aut A'$ is in state $(p,q)$ at a node $v$ in its run on $T$ where $p=\state(\aut A,T\otimes \emptyset)$ and $q=\state(\aut A,T\otimes \tup a)$ where $\tup a\from\tup y\to V(T)$ is the valuation mapping all variables in $\tup y$ to $v$.

Such an automaton $\aut A'$ can be obtained by defining its  transition sentences $\delta_{(p,q),a}$ for $p,q\in Q$ and $a\in A$,
where $\delta_{(p,q),a}$ is the transition sentence $\delta_{p,(a,\emptyset)}\land
\delta_{q,(a,\tup y)}$  of $\aut A$, where $\tup y$ are the variables of the shape $\sigma$.

In the preprocessing phase, given $T$ compute the run of $\aut A'$ on $T\otimes \emptyset$
using \cref{lem:automata evaluation}.
In the answering phase, given a valuation $\tup a\in V(T)^{\tup y}$, proceed as follows.
If $\tup a$ does not map all variables of~$\tup y$ to the same
node of $T$, answer the query negatively.
Otherwise, let $v\in V(T)$ be such that $\tup a$ maps all variables of $\tup y$ to $v$.
If $\rho_0(v)=(p,q_\ast)$ where $q_\ast$ is the state labeling the root $\ast$ of the shape $\sigma$ then answer the query positively, otherwise answer negatively. This finishes the base case.

\medskip
We now consider the inductive step.
Let $\sigma$ be a shape with underlying tree $S$ and assume $|V(S)|>a$. Following the notation from
    above, we write $\ast$ for the root of $S$.
Let $W$ be the set of children of $\ast$ in~$S$.
    For each  $w\in W$
     define a shape
    $\sigma_w$ with the following components:
    \begin{itemize}
        \item the variables $\bar y_w:=\setof{y\in\tup y}{w\preceq \tup s(y)}$,
        \item the subtree $S_w$ of $S$
        rooted at $w$,
        \item the restriction $L_w$ of the labeling $L$ to $S_w$,
        \item the restriction $\tup s_w$ of the valuation $\tup s$ to $\tup y_w$.
    \end{itemize}
    Note that each shape $\sigma_w$ has fewer nodes than $\sigma$,
    so we may apply the inductive assumption to it.

The key observation is encompassed by the following claim, whose proof is immediate by the definitions. Let $q_\ast$ denote the state labeling the root $\ast$ of $\sigma$.

\begin{claim}\label{claim:conditions}
  A tuple $\tup a\in V(T)^{\tup y}$ matches $\sigma(\tup y)$
if and only if each of the following conditions hold:
\begin{enumerate}
  \item For each $w\in W$, the restriction of $\tup a$ to $\tup y_w$ matches the shape $\sigma_w(\tup y_w)$.
  Let $v_w$ denote the node of $T$ to which the root of $\sigma_w$ is mapped.

  \item There is a vertex $v$ of $T$ (this will be the vertex to which the root of $\sigma$ is mapped to) such that
  for any two distinct $w,w'\in W$, the least common ancestor of $v_w$ and $v_{w'}$ is equal to $v$, and furthermore $v$ is different than $v_w$ and $v_{w'}$.

  \item For each $w\in W$ let $u_{w}$ denote the child of $v$ in the direction of $v_w$, that is $u_w=\dir(v,v_w)$. Then the function $f_{v_wu_w}$ maps the root state of $\sigma_w$ to the state $q_w$ labeling the edge $\ast w$ in $S$ (that is, the edge connecting the root $\ast$ with $w$).

  \item Let $\tup z\subset \tup y$ be the set of variables
  that are occupying the root $\ast$ of $\sigma$, that is,
  $\tup z=\tup s^{-1}(\ast)$, where $\tup s$ is the valuation comprising $\sigma$. Then $\tup a$ maps each variable in $\tup z$ to $v$.

  \item Let $c\in A$ be the label of $v$ in $T$, and let $\tup z$ be as above. Then $(c,\tup z)$ is a letter of the alphabet $A\otimes \tup x$.
  The final condition is that the sentence $\delta_{q_\ast,(c,\tup z)}$ holds
   in the whorl $\str A_v$ at $v$ labeled by the function $\rho_0$,
   apart from each of the vertices $u_w$ for $w\in W$, which is instead labeled with the state $q_w$ labeling the edge $\ast w$ in $S$.
\end{enumerate}
\end{claim}

We can now describe our query-answering algorithm.
In the initialization phase, given a \mbox{$\CC$-augmented} tree $T$, we initialize the following data structures, corresponding to the items 1, 2, and 3 above:
\begin{itemize}
  \item the data structures
  for the shapes $\sigma_w$, where $w\in W$, obtained by inductive assumption,

  \item the data structure given by
  \cref{lem:navigation} answering $\lca$ and $\dir$ queries in $T$,

  \item the data structure given by \cref{lem:colcombet-rephrased}, allowing to compute the function $f_{ww'}$, given a node $w$ and its ancestor $w'$.
\end{itemize}
We need one more data structure, corresponding to  item 5 above.
Let $\tup w$ be the set $W$, now treated as variable names.
For each $c\in A$ consider the first-order formula $\gamma_c(\tup w)$
such that for every $Q$-labeled $\Sigma$-structure $\str A$,  a given valuation $\tup a\in V(\str A)^{\tup w}$ satisfies
$\gamma_c(\tup w)$ if and only if
 the transition sentence~$\delta_{q_\ast,(c,\tup z)}$ holds in the $Q$-labeled $\Sigma$-structure $\str A$ in which the label of $\tup a(w)$ is replaced by the state $q_w$ (the state labeling the edge connecting $\ast$ with $w$ in $\sigma$). It is straightforward to construct such a first-order formula~$\gamma_c(\tup w)$.

 Now, given a tree $T$, for each node $v$ of $T$
 initialize the data structure for answering queries to~$\gamma_c(\tup w)$, where $c$ is the label of $v$. By the assumption about the class $\CC$,
 such a data structure can be computed in time linear in $|\str A_v|$,
 so jointly all of them can be computed in time linear in $|T|$.

 This finishes the description of the initialization phase.
 In the query-answering phase, given a tuple~$\tup a$,
 verify the conditions 1-5 listed in \cref{claim:conditions}.
 This can be done using the precomputed data structures
 in time which is bounded in terms of $|\aut A|$,
 and which is independent of $|T|$.
 This finishes the inductive proof of \cref{lem:shape answering}.
\end{proof}

Recall that \cref{lem:shape answering} yields \cref{lem:answering}, which together with \cref{thm:formulas to automata} proves \cref{thm:answering}.

\begin{remark}\label{rem:to-FO}
  By revisiting the proof from this section,
  one can prove that for any automaton $\aut A$ over augmented
trees there is a finite alphabet $B$ and sentence $\phi$ of $\FO(\preceq,B\cup \Sigma)$ such that given an augmented tree $T$ one can compute another augmented tree $T'$ over the alphabet $B$ such that $\aut A$ accepts $T$ if and only if\, $T'$ satisfies $\phi$.
Moreover, if model-checking first-order logic is efficient on $\CC$, the computation of\, $T'$ takes time linear in $|T|$. This argument ultimately relies on the observation that the proof of \cref{thm:colcombet-ds} yields (in linear time) a labeling of a given input tree $T$ using a finite alphabet $B$, such that the composition of functions labeling the edges from a path between two given nodes $v$ and $w$ can be expressed in $\FO(\preceq,B)$. This observation originates from~\cite{Colcombet07}.
\end{remark}

\subsection{Enumeration}
We now extend our results to obtain efficient enumeration of \FO-queries on $\CC$-augmented trees.

\medskip An \emph{enumerator} for a sequence of elements
$a_1,\ldots,a_k$ is a data structure that implements a function that, when invoked for the $i$th time,
outputs $a_i$ if $1\le i\le k$, or reports an error
if $i>k$.  The enumerator has \emph{delay} $T$ if after each
invocation of the function, the output is returned in time at most
$T$.  An enumerator for a set $A$ is an enumerator for any sequence
enumerating $A$ without repetitions.  More information on enumeration
algorithms, together with constant delay specifically can be found
in~\cite[Section~2 and Section~5.3]{DBLP:journals/eatcs/Strozecki19}.

 Say that first-order enumeration is \emph{efficient} on a class $\CC$ of
$\Sigma$-structures if there is a computable function $f\from\N\to\N$
and an algorithm that, given a structure $\str A\in\CC$ vertex-labeled
by a finite alphabet~$A$, and a first-order formula $\phi(\tup x)$ in the signature $\Sigma\cup A$ (where elements of $A$ are viewed as unary predicates),  computes in time $f(|\phi|)\cdot |\str A|$
an enumerator for~$\phi(\str A)$ with delay $f(|\phi|)$.

\begin{theorem}\label{thm:formula enumeration}
  Suppose that first-order enumeration is efficient on~$\CC$.  Then there is a computable function $f\from\N \to \N$ and an algorithm that, given a $\CC$-augmented tree $T$ over an alphabet $A$  and a $\FOMSO$-formula $\phi(\tup x)$ computes in time $f(|\phi|)\cdot |T|$
  an enumerator for~$\phi(\str A)$ with delay~$f(|\phi|)$.
\end{theorem}

To prove the theorem, we reduce its statement to a statement about
automata on augmented trees. Recall that for an automaton
$\aut A$ on augmented trees
  over an alphabet of the form $A\otimes \tup x$ and an augmented tree $T$ over the alphabet $A$, by $\aut A(T)$ we denote the set of all valuations $\tup a\from \tup x\to V(T)$ such that
  $\aut A$ accepts $T\otimes \tup a$.

\begin{restatable}{lemma}{lemenum}\label{lem:automata enumeration}
  Suppose that first-order enumeration is efficient on~$\CC$.    Then there is a computable function $f\from\N\to\N$ and
  an algorithm that, given an automaton $\aut A$ on augmented trees
  over an alphabet of the form~$A\otimes \tup x$ and a $\CC$-augmented
  tree $T$ over the alphabet $A$, computes in time
  $f(|\phi|)\cdot |T|$ an enumerator for~$\aut A(T)$ with delay
  $f(|\phi|)$.
\end{restatable}

\noindent Clearly,~\cref{lem:automata enumeration}
and~\cref{thm:formulas to automata} imply~\cref{thm:formula
  enumeration}. It therefore remains to prove~\cref{lem:automata
  enumeration}.  This is done in
  the remainder of this section.

\medskip
We use the notion of shapes as defined in the previous section.
Fix a shape $\sigma$. Our next goal is to
compute for each possible starting node $v$ an enumerator
for the set of all valuations $\tup a$ that match the shape $\sigma$ with starting node~$v$.


\begin{lemma}\label{lem:shape enumeration}
    For every shape $\sigma$ there is a constant $c$ computable from $\aut A$ and $\sigma$, and an algorithm that, given a tree~$T$ over the alphabet $A$, computes
    in time  $c\cdot |T|$ a collection $(\mathcal E_v)_{v\in V(T)}$ of enumerators,
    where for each $v\in V(T)$,  $\mathcal E_v$ is
    an enumerator with delay $c$ for the set of all valuations $\tup a$ that match $\sigma$ with starting node $v$.
\end{lemma}

We first show how~\cref{lem:shape enumeration} implies~\cref{lem:automata enumeration}. This is analogous to how \cref{lem:shape answering} implies \cref{lem:answering}.
\begin{proof}
Let $\CC$ be a class of structures for which enumeration of
$\FO(\Sigma\cup A)$ queries is efficient for every finite alphabet $A$.
Let $\aut A$ be an automaton on augmented trees over an  alphabet
of the form $A\otimes \tup x$ and let $T$ be a $\CC$-augmented tree over the alphabet $A$.
We extend $\tup x$ to $\tup x\cup\set{z}$ by adding a dummy variable~$z$.
We view
$\aut A$ as an automaton over the alphabet $A\otimes (\tup x\cup \{z\})$. Using~\cref{lem:shapes}, compute the set~$\Delta$ of all shapes $\sigma$ in which the root node of the shape carries the variable
$z$ (that is, $\tup s(z)=\ast$) and
is labeled with an accepting state of $\aut A$ (that is, $L(\ast)$ is an accepting state). For each such shape $\sigma$ compute
the enumerator $\aut E_r^\sigma$ for the set of all valuations $\tup a\from \tup x\cup \set{z}\to V(T)$ that
match $\sigma$ with starting node $r$, where $r$ is the root of $T$. This is possible in time linear
in $|T|$ by~\cref{lem:shape enumeration}.

An enumerator for $\aut A(T)$ as in the statement of~\cref{lem:automata enumeration}
is obtained
by concatenating the enumerators $\aut E_r^\sigma$ for all the shapes $\sigma\in \Delta$,
where each output $\tup a\from \tup x\cup \set{z}\to V(T)$ of each of those enumerators is post-processed by projecting it
to a valuation $\tup a'\from \tup x\to V(T)$, that is, by removing
the entry corresponding to the dummy variable $z$.

 As remarked above (see comment following the definition of matching a shape), for distinct shapes the sets of valuations that
match the shapes are disjoint.
Hence, all valuations are enumerated without
repetition.  By choosing $f$ appropriately this
yields an enumerator that is computable in time $f(|\phi|)\cdot |T|$
and has delay~$f(|\phi|)$.
\end{proof}

It remains to prove \cref{lem:shape enumeration}.
The idea is to reverse the analysis performed when considering the query-answering problem. In a sense, for many of the
basic queries used in the previous sections,
we will need to devise its enumeration analogue.
For instance, we will need to devise an
enumeration analogue of \cref{thm:colcombet-ds}, that is an algorithm that enumerates all descendants $w$ of a given node $v$ such that the composition of the functions labeling the edges along the path from $w$ to $v$ has certain properties.

\subsubsection*{Auxiliary enumerators for trees}
We will need the following lemmas regarding enumeration of specific sets on unranked trees.

\medskip
For a set $U$ of nodes of a tree $T$ and two nodes $v,w$, say that $w$ is a \emph{minimal} descendant of $v$ in $U$ if $w\in U$, $w$ is a descendant of $v$,
and no inner node of the path from $v$ to $w$ belongs to $U$.

\begin{lemma}\label{lem:minimal-descendants}
  There is a constant $c$ and an algorithm that, given a tree $T$ and a set $U$ of its nodes, computes in time $c\cdot |T|$ a collection $(\mathcal E_v)_{v\in V(T)}$
  of enumerators, where for each $v\in V(T)$,  $\mathcal E_v$ is
  an enumerator with delay $c$ for the set of minimal strict descendants $w$ of $v$ that belong to $U$.
\end{lemma}
\begin{proof}
Assume without loss of generality that the root of $T$ belongs to $U$.
For each node $v$ of $T$ let~$M_v$ denote the set of all the minimal strict descendants of $v$ that belong to $U$,
and let $g(v)$ denote the $\preceq$-largest ancestor of $v$ that belongs to $U$ (with $g(v)=v$ if $v\in U$).
Also, let $g'(v)$ denote the $\prec$-largest strict ancestor of $v$ that belongs to $U$ (which is not defined for the root).
Note that $M_v\subseteq M_{g(v)}$ for every node $v$ of $T$.

By performing a depth-first search traversal, we can compute for each $u\in U$ and jointly in time~$\Oh(|T|)$,
a linked list enumerating all elements of $M_u$, ordered by the search.
Observe now that for every $v$ of $T$, the list $M_v$ is composed of
successive elements of $M_{g(v)}$.
Again by performing a depth-first search traversal, we can compute, for each $v\in T$ and jointly in time $\Oh(|T|)$,
the two elements $\min(M_v)$ and $\max(M_v)$, respectively the first and last element of $M_v$ according to the order given by the search. Finally, every element $u$ in $U$ (except for the root) points towards its position in the linked list $M_{g'(u)}$.

Now, to obtain an enumerator for $M_v$ given a node $v$ of $T$, first lookup $\min(M_v)$ and $\max(M_v)$.
Then, go to the position of $\min(M_v)$ in $M_{g(v)}=M_{g'(\min(M_v))}$. Finally, follow the list enumerating the elements of $M_{g(v)}$ until $\max(M_v)$ is reached. This enumerates $M_v$, with constant delay.
\end{proof}

The following is an enumeration variant of~\cref{lem:functional-warp},
and its proof follows the same ideas.
\begin{lemma}\label{lem:enumerate}
  Fix a finite set $Q$ of size $q$
  and two elements $x,y\in Q$.
  There is a constant $c$ computable from $q$, and an algorithm that,
  given a tree $T$ in which each edge $vw$ is labeled by a function $f_{vw}\colon Q\to Q$, computes in time $c\cdot |T|$ a collection
  $(\mathcal E_v)_{v\in V(T)}$
  of enumerators, where for each $v\in V(T)$,  $\mathcal E_v$ is
    an enumerator with delay $c$ that enumerates all descendants $w$ of $v$ such that the composition of the functions labeling the edges of the path from $v$ to $w$ maps $x$ to $y$.
\end{lemma}
\begin{proof}
  For two nodes $v,w$ with $v\preceq w$ let $f_{vw}$ denote the composition of the functions labeling the edges of the path from $v$ to $w$.
  Let $D_v^{xy}$ be the set of all strict descendants $w$ of $v$ such that $f_{vw}(x)=y$, and
  let $M_v^{xy}\subset D_v^{xy}$ be the set of \emph{minimal}
  such descendants, where $w$ is minimal if there is no $u\in D_v^{xy}$ with $v\prec u\prec w$.
  It suffices to compute, for each $v\in V(T)$, an enumerator
  for the set $M_v^{xy}$. Indeed, the enumerator for $D_v^{xy}$ can be then obtained by enumerating all nodes $w\in M_v^{xy}$
  and then, for each such~$w$,
  recursively enumerating the set
  $D_w^{yy}$.

  Proceeding from the root of $T$ towards the leaves
  (see \cref{fig:paths}), for
  each node $v$ of $T$ construct an injective coloring function
  $c_v\from Q\to [q]$, so that the following property holds:
  \[c_w(f_{vw}(x))\le c_v(x)\qquad\text{for every edge $vw$ ($w$ child of $v$) and all
  $x\in Q$}.\]

  This can be done as follows. At the root $v$,
  pick an arbitrary injection $c_v\from Q\to [q]$. If $c_v$ is
  defined for some node $v$ and $w$ is its child, then first define
  $c_{w}\from f_{vw}(Q)\to [q]$ by setting $c_w(y)$, for
  $y\in f_{vw}(Q)$, to the least color $c\in [q]$ such that
  $c_v(x)=c$ for some $x\in Q$ with $f_{vw}(x)=y$. Next, extend
  $c_w$ to an injection $c_w\from Q\to [q]$ arbitrarily.  The
  required property then holds by construction.  The computation of
  the colorings $c_v$ therefore takes time $\Oh(|T|\cdot r)$.


  For any pair of nodes $v$ and $w$ such that $v$ is an ancestor of $w$, let $f_{vw}\from Q\to Q$ denote the composition of the functions along the path from $v$ to $w$.
  Then
  \begin{align}\label{eq:monotone colors}
    c_w(f_{vw}(x))\le c_v(x)\qquad\text{for all $v,w$ with $v\preceq w$ and all $x\in Q$}.
  \end{align}

  This is proved easily by induction on the distance between $w$ and $v$. In other words, starting with an element $x\in Q$ and progressing downwards along a path in the tree, applying the functions labeling the edges in the path to the current element of $Q$,
  we obtain a sequence of elements whose colors are monotonously decreasing. As the number $q$ of colors is finite, those colors change only a bounded number of times. We exploit this observation in our enumeration algorithm.

  For every $y\in Q$ and color $d\in [q]$ let $C[y/d]$ denote the
  set of nodes $v$ such that $c_v(y)=d$. Then
  for all $x,y\in Q$, $d\in [q]$, the set
  \[E_v^{xyd}:=M_v^{xy}\cap C[y/d]\]
  consists of those minimal descendants $w$ of $v$ satisfying
  $f_{vw}(x)=y$, such that $c_w(y)=d$.

  We describe how to compute in time $\Oh_r(|V(T)|)$ a collection of constant-delay enumerators $\mathcal E_v^{xyd}$ for the set $E_v^{xyd}$, for all $x,y\in Q$, $d\in [q]$ and all nodes $v\in V(T)$. As $M_v^{xy}$ is the disjoint union of the sets
  $E_v^{xyd}$, for $d\in [q]$, an enumerator for~$M_v^{xy}$ can be obtained by concatenating the enumerators $\mathcal E_v^{xyd}$ for all $d\in [q]$.

  Fix elements $x,y\in Q$ and colors $c,d\in [q]$. We will compute (in time $\Oh_r(|V(T)|)$) constant-delay enumerators
  for each of the sets $E_v^{xyd}$, for all nodes $v\in C[x/c]$. Since the sets $(C[x/c])_{c\in [q]}$ partition~$V(T)$, this will jointly give all the required enumerators.
  Assume in what follows that $v\in C[x/c]$, that is, $c_v(x)=c$.
  Our goal now is to compute an enumerator for the set  $E_v^{xyd}$ of those minimal descendants $w$ of $v$ satisfying
  $f_{vw}(x)=y$, such that $c_w(y)=d$.

  Let $F_c\subset E(T)$ be the set of those edges $vw$ in $T$ such that
  $f_{vw}$ maps the element of color $c$ to an element of strictly smaller color.
  The set $F_c$ can be computed in time  $\Oh_r(|V(T)|)$.
  Write $v\sim_c w$ if the path between two nodes $v$ and $w$ does not contain any edge from $F_c$.
  Furthermore, compute in time $\Oh_r(|V(T)|)$ the set $C[y/d]$.

  Fix a node $v\in C[x/c]$, so that $c_v(x)=c$.
  Assume first that $c<d$. Then $E_v^{xyd}=\emptyset$,
  since by~\eqref{eq:monotone colors}
  $v$ has no descendant $w$ with $f_{vw}(c)>c$.
  If $c=d$, then for every descendant $w$ of $v$
  we have $w\in E_v^{xyd}$ if and only if
  $v\sim_cw$, $w\in C[y/d]$ and
  no inner node of the path from $v$ to $w$ belongs to $C[y/d]$.
  Consider the rooted forest $T-F_c$ obtained from $T$ by removing the edges in $F_c$, where the root of each tree in the forest is the its smallest node with respect to the ancestor relation in $T$.
  By considering each tree in this forest separately,
  the problem is solved by the algorithm from~\cref{lem:minimal-descendants} applied to that tree and the set $C[y/d]$ restricted to that tree.

  We are left with the case when $c>d$.
  We proceed by induction on $c-d$.
  Intuitively, to compute~$E_v^{xyd}$ we need to
  enumerate all minimal descendants $w$ of $v$ such that the color of $f_{vw}(x)$ at $w$ is some color~$c'$ that is smaller than $c$, and at each such node $w$ invoke the enumerator for the set $E_w^{x'yd}$ where $x'$ is the element with color $c'$ at $w$,
  unless this set is empty. We make this more precise below.

  \medskip
  Compute the set $U$ of all nodes $u$ of $T$ with the following properties:
  \begin{itemize}
    \item $u$ is the lower endpoint of an edge $u'u$ in $F_c$,
    \item $M_u^{x'y}\cap C[y/d]$ is non-empty,
    where
    $x'=f_{u'u}(z)$ and $z\in Z$ is such that $c_{u'}(z)=c$.
  \end{itemize}
  Note that as the edge $u'u$ is in $F_c$ it follows that $c_u(x')<c$,
  and we can apply the inductive assumption, thus assuming that
  each node $u$ is equipped with a constant-delay enumerator for $M_u^{x'y}\cap C[y/d]$, for all $x'\in Q$. In particular we can test whether $M_u^{x'y}\cap C[y/d]$ is non-empty in constant time, and so the set $U$ can be computed in linear time.

  The following conditions are equivalent for a node $v\in C[x/c]$
  and its descendant $w$:
  \begin{itemize}
    \item $w\in E_v^{xyd}=M_v^{xy}\cap C[y/d]$,
    \item $v$ has some minimal descendant $u\in U$ such that $w\in E_u^{x'yd}=M_u^{x'y}\cap C[y/d]$, where $x'$ is as above.
  \end{itemize}
  It follows that $E_v^{xyd}$
  is the disjoint union, over all minimal descendants $u\in U$ of $v$,
  of the (non-empty) set
  $E_u^{x'yd}$, where $x'$ is as above.
  Hence, an enumerator for $E_v^{xyd}$
  can be obtained by first enumerating all minimal descendants $u\in U$ of $v$ using the data structure from~\cref{lem:minimal-descendants}, and then, for each such~$u$,
  using the enumerator for $E_u^{x'yd}$ obtained by inductive assumption.
\end{proof}

In the lemma above, the functions are composed along the paths leading from an ancestor to a descendant.
In our application, we will need to compose the functions in the reverse order.
Hence the following.

\begin{lemma}\label{lem:reverse enumerate}
  Fix a finite set $Q$ of size $q$
  and a set of functions ${\cal F}\subset Q^Q$.
  There is a constant $c$ computable from $q$, and an algorithm that,
  given a tree $T$ in which each edge $vw$ ($w$ child of $v$)
  is labeled by a function $f_{wv}\colon Q\to Q$, computes in time $c\cdot |T|$ a collection
  $(\mathcal E_v)_{v\in V(T)}$
  of enumerators, where for each $v\in V(T)$,  $\mathcal E_v$ is
    an enumerator with delay $c$ that enumerates all descendants $w$ of $v$ such that the composition of the functions labeling the edges of the path from $w$ to $v$ belongs to ${\cal F}$.
\end{lemma}
\begin{proof}
     A function $f\from Q\to Q$ induces
    a function $\rho_f\from Q^Q\to Q^Q$,
    defined by $\rho_f(g)=f;g$ for $g\from Q\to Q$.
This reverses the order of compositions:
     $\rho_{f_1;f_2}=\rho_{f_2};\rho_{f_1}$.
     Given a tree $T$ as above,
     define a new label for each edge $vw$, namely
     $F_{vw}:=\rho_{f}$ where $f=f_{wv}$ is the original label.
     Then for any two nodes $v, w$ where $v\preceq w$,
    the composition of the new labels along the path from $v$ to $w$ is a function $F\from Q^Q\to Q^Q$ which is equal to $\rho_{f}$, where $f$ is the composition of the old labels along the reverse path, from $w$ to $v$.
    In particular, $F(\text{id})=\rho_f(\text{id})=f$,
    where $\text{id}$ is the identity on $Q$.

    Reassuming, for every function $f\from Q\to Q$, the composition of the new labels along the path from~$v$ to~$w$ is a function $F$ mapping $\text{id}$ to $f$ if and only if the composition of the old labels along the path from~$w$ to~$v$ is equal to $f$.

     We now apply~\cref{lem:enumerate} to the tree $T$ with the new labels, and to every pair of elements $x,y\in Q^Q$ with $x=\text{id}\in Q^Q$ and $y\in \cal F$.
\end{proof}

\subsubsection*{Enumerator for the entailment relation}

Fix an augmented tree $T$ and let $\rho_0$ be the run of $\aut A$ on $T\otimes \emptyset$.
Recall the entailment relation \mbox{$[w\mapsto p]\vdash [w'\mapsto p']$}
which intuitively means that enforcing the state in $\rho_0$
at node $w$ to be equal to $p$ results in the automaton $\aut A$ reaching state $p'$ at the node $w'$.
    We first translate \cref{lem:reverse enumerate} to the language of entailment, obtaining an enumeration analogue of \cref{lem:colcombet-rephrased}.

\begin{lemma}\label{lem:state entailment}
    Fix states $q_1,q_2\in Q$.
  There is a constant $c$ computable from $\aut A$ and an algorithm that,
  given a tree $T$ and a set $X\subseteq V(T)$ computes in time $c\cdot |T|$ two collections
  $(\mathcal E_v)_{v\in V(T)}$ and $(\mathcal E'_v)_{v\in V(T)}$
  of enumerators with delay~$c$, doing the following: For each $v\in V(T)$,  $\mathcal E_v$ is
    an enumerator that enumerates all children $u_1$ of~$v$ that have some descendant $u_2\in X$ such that $[u_2\mapsto q_2]\vdash [u_1\mapsto q_1]$.
    For each $v\in V(T)$,  $\mathcal E'_v$ is
    an enumerator that enumerates all descendants $u_2\in X$ of $v$ such that $[u_2\mapsto q_2]\vdash [v\mapsto q_1]$.
\end{lemma}
\begin{proof}
We let $Q'\coloneqq Q\mathop{\dot\cup} \{m,0,1\}$. We label each edge $vw$ of $T$ by the
function $f_{wv}\from Q'\rightarrow Q'$ such that
\begin{itemize}
\item
$f_{wv}(p)=q$ if $[w\mapsto p]\vdash [v\mapsto q]$, for all $p,q\in Q$,
\item $f_{wv}(m)=1$ if
$w\in X$ and $f_{wv}(m)=0$ if $w\not\in X$, and
\item   $f_{wv}(0)=0$, and $f_{wv}(1)=1$.
\end{itemize}

Observe that if $w$ is any descendant of $v$, then $[w\mapsto p]\vdash [v\mapsto q]$  if and only if the composition of the functions labeling the edges along the path from $w$ to $v$ is a function $f$ such that $f(p)=q$. Furthermore, this function
maps $m$ to $1$ if and only if $w\in X$.

We now apply \cref{lem:reverse enumerate} to apply in the claimed time
the second collection $(\mathcal E'_v)_{v\in V(T)}$ of enumerators with the
desired properties.

We now compute the set $Y$ of
all vertices $v$ that have a descendant $u\in X$ such that
$[u\mapsto q_2]\vdash [v\mapsto q_1]$. For this, we
iterate over all $v\in V(T)$ and let $\mathcal{E}'_v$ enumerate
its first output. If $\mathcal{E}'_v$ does not report an error,
there exists a descendant $u\in X$ of $v$ such that
$[u\mapsto q_2]\vdash [v\mapsto q_1]$. Assuming that each
enumeration takes a constant time $c'$, we can compute
this set in time $c'\cdot |T|$. Now it is easy to construct
the first collection $(\mathcal E_v)_{v\in V(T)}$
of enumerators with the desired properties (compare with
the more complicated construction in the proof of
\cref{lem:minimal-descendants}).
%
%
%
%
%
\end{proof}

\subsubsection*{Proof of \cref{lem:shape enumeration}}

We now prove~\cref{lem:shape enumeration}.
\begin{proof}
    Fix a shape $\sigma$. The proof proceeds by induction on the size of the tree $S$ underlying $\sigma$.

    In the base case $|S|=1$
there may be at most one valuation $\tup a$ that matches $\sigma$ with a starting node~$v$, namely the valuation $\tup a\from\tup y\to V(T)$  that maps all the variables $\tup y$ of $\sigma$ to $v$.
This valuation matches~$\sigma$ precisely whenever $\rho_v(v)=p$, where $\rho_v\from V(T)\to Q$ is the run of $\aut A$ on $T\otimes \tup a$
and $p=L(\ast)$.
We may therefore use the same reasoning as in the base case of the inductive proof of \cref{lem:shape answering}.
Namely,
we construct an automaton $\aut A'$ over the alphabet $A$ with states $Q^2$ (where $Q$ are the states of $\aut A$) and with the following property:
  $\aut A'$ is in state $(q,q')$ at a node $v$ in its run on $T$ where $q=\state(\aut A,T\otimes \emptyset)$ and $q'=\state(\aut A,T\otimes \tup a)$ where $\tup a\from\tup y\to V(T)$ is the valuation mapping all variables in $\tup y$ to $v$.

  Given a tree $T$, we compute the run $\rho'$ of $\aut A'$ on $T\otimes \emptyset$,
  in time linear in $|T|$ using \cref{lem:automata evaluation}.
We now output the following enumerators
$(\mathcal E_v)_{v\in V(T)}$.
For all nodes $v\in V(T)$, if $\rho'$ labels
$v$ by a pair in $Q\times \set p$, then let
$\mathcal E_v$ be the enumerator that outputs the
unique valuation mapping all variables in $\tup y$ to~$v$,
otherwise, $\mathcal E_v$ is an enumerator for the empty set.
This finishes the base case.

\smallskip
Now assume $|S|>1$. Following the notation from
    above, we write $\ast$ for the root of $S$.
Let $W$ be the set of children of $\ast$ in~$S$.
    As in the case of query answering, for each  $w\in W$
     define a shape
    $\sigma_w$ with the following components:
    \begin{itemize}
        \item the variables $\bar y_w:=\setof{y\in\tup y}{w\preceq \tup s(y)}$,
        \item the subtree $S_w$ of $S$
        rooted at $w$,
        \item the restriction $L_w$ of the labeling $L$ to $S_w$,
        \item the restriction $\tup s_w$ of the valuation $\tup s$ to $\tup y_w$.
    \end{itemize}

Note that $S_w$ has fewer nodes than $S$, so we can apply the inductive assumption to $\sigma_w$,
for each~${w\in W}$.

Let $q$ be the label of the root in $S$ and let $\tup z\subset \tup y$ be the set of those variables $y\in \tup y$ such that $\tup s(y)$ is the root $\ast$ in $\sigma$.
Let $\rho_0$ be the run of $\aut A$ on $T\otimes \emptyset$.

For a node $v\in V(T)$ with label $a$ and injective function $G\from W\to \str A_v$,
say that  $G$ is \emph{satisfiable} if there is a function $G'\from W\to\str A_v$ such that the  following conditions hold (recall that
$L$ also labels edges of $S$ and that $\ast w$ denotes the edge between the
root $\ast$ of $S$ and $w$):
\begin{itemize}
    \item
    $\delta_{q,(a,\tup z)}$ holds in $\str A_v$ labeled by $\rho_0$, except for the nodes $G(w)$ for $w\in W$ that are instead labeled by~$L(\ast w)$,
    \item for each $w\in W$ we have
     $G'(w)\succeq G(w)$
    and $[G'(w) \mapsto L(w)]\vdash [G(w)\mapsto L(\ast w)]$, and
\item for each $w\in W$
 there is some valuation of $\bar y_w$ that matches $\sigma_w$ with
  starting node $G'(w)$.
\end{itemize}

\begin{claim}\label{claim:bijection}
    Valuations $\tup a$ that match $\sigma$ with starting node $v$
    are in bijective correspondence with tuples $(G,G',(\tup a_w)_{w\in W})$, such that $G,G'\from W\to \str A_v$ and:
    \begin{enumerate}
        \item $G$ is satisfiable,
        \item $G'(w)\succeq G(w)$ and $[G'(w)\mapsto L(w)]\vdash [G(w)\mapsto L(\ast w)]$ for all $w\in W$, and
        \item  $\tup a_w$ matches $\sigma_w$ with starting node $G'(w)$
        for all $w\in W$.
    \end{enumerate}
\end{claim}
\begin{proof}
    Fix a valuation $\tup a\from \tup x\to V(T)$ that matches $\sigma$ with starting node $v$ and let $\rho$ be the run of $\aut A$ on~$T\otimes \tup a$. As remarked above, the function $F\from V(S)\to V(T)$ witnessing that $\tup a$ matches
    $\sigma$ is unique.
    Define  $(G,G',(\tup a_w)_{w\in W}))$ as follows, by setting for $w\in W$:
    \begin{itemize}
        \item $G(w)$ is the $F(w)$-directed child of $F(\ast)$, 
        \item $G'(w)=F(w)$,
        \item $\tup a_w$ is the restriction of $\tup a$ to $\tup y_w\subset \tup y$.
    \end{itemize}

    As $\tup a$ matches $\sigma$ with starting node $v$, we have that
    $\rho(v)=q$ and
    the image of $\tup a$ is contained in the union of $\set v$ and the set of descendants of either of the nodes $G(w)$, for $w\in W$.
    Since the run $\rho$ is computed deterministically in a bottom-up fashion, it follows that the runs $\rho$ and  $\rho_0$ differ only on those vertices of $\str A_v$
    that are in the image of $G$.
Moreover, $\rho(G(w))=L(\ast w)$ by definition of $S$. As $\rho(v)=q$ it follows that the formula $\delta_{q,(a,\tup z)}$ holds in $\str A_v$ labeled by $\rho$.

Fix $w\in W$. We show that
  $\tup a_w$ matches $\sigma_w$ with starting node $G'(w)$.
To this end, it is enough to
observe that the restriction $F_w$ of $F$ to $V(S_w)$ witnesses this.
In particular, this
shows that  $G$ is satisfiable.

Vice versa, from a tuple $(G,G',(\tup a_w)_{w\in W})$ satisfying
the properties of the claim we uniquely obtain the valuation $\bar a$
matching $\sigma$ with starting node $v$: $\bar a=\bigcup_{w\in W}\tup a_w\cup \bar b$, where $\bar b$ is the valuation that maps~$\bar z$
(the set of variables $y\in \bar y$ such that $\bar s(y)=\ast$) to $v$.

Hence, the assignment $\tup a\mapsto (G,G',(\tup a_w)_{w\in W})$ defined above yields the bijective correspondence as stated in the claim.
\end{proof}

The enumerator for the set of all valuations $\tup a$ that match $\sigma$ with starting node $v$
is obtained by enumerating all triples $(G,G',(\tup a_w)_{w\in W})$ that satisfy
the properties stated in the claim.
We now argue that all components
can be enumerated with constant delay.

For $w\in W$,
let $X_w$ denote the set of nodes $u'$ such that there is a valuation that matches $\sigma_w$ with starting node $u'$.
By inductive assumption, in particular we can compute $X_w$ in time linear in $|T|$.
Apply \cref{lem:state entailment} to
$T$ and $X_w$ and states $q_2=L(w)$ and $q_1=L(\ast w)$, to obtain, for each $v\in V(T)$, an enumerator for the set $C_w(v)$ of all children $u$ of $v$ such that there exists $u'\succeq u$ in $X_w$ satisfying \[[u'\mapsto L(w)]\vdash
[u\mapsto L(\ast w)].\]

For each child $u$ of $v$,
label $u$ by the set of all $w\in W$ such that $u\in C_w(v)$. This yields a labeling $\lambda_v$ of the whorl $\str A_v$.

Let $\tup w$ be the set $W$, where now each $w\in \tup w$ is treated as a variable.
Then there is a first-order formula~$\psi(\tup w)$ with free variables $\tup w$
that holds in $\str A_v$ with the  labeling $\lambda_v$ for a valuation $\tup b\from \tup w\to \str A_v$ if~$\tup b$
induces a function $G\from W\to \str A_v$ that is injective and satisfiable.
By assumption on the class~$\CC$,
we can compute an enumerator for the formula $\psi(\str A_v)$, for each $v\in V(T)$. 
Each enumerator has delay bounded by some constant, and all those enumerators can be computed in time linear in $|T|$.
In particular, in this way we obtain for each $v\in V(T)$ an enumerator with constant delay that enumerates all satisfiable functions $G\from W\to\str A_v$.

Fix $v\in V(T)$ and a satisfiable function $G\from W\to \str A_v$. Our next task is to enumerate all functions $G'\from W\to\str A_v$ that certify satisfiability of $G$.
The key observation is that this can be done by enumerating, for each $w\in W$ independently, all nodes  $u'$ of $V(T)$ with $u'\succeq G(w)$ and \mbox{$[u'\mapsto L(w)]\vdash[G(w)\mapsto L(\ast w)]$},
and such that there is some valuation of $\tup y_w$ that matches $\sigma_w$ with starting node $u'$.
Note that this last condition is equivalent to $u'\in X_w$. Hence, for a fixed $w\in W$, it suffices to enumerate all descendants $u'\in X_w$
of $G(w)$ with  $[u'\mapsto L(w)]\vdash[G(w)\mapsto L(\ast w)]$.
This is done using the enumerator obtained from~\cref{lem:state entailment}.

More precisely,
for each $w\in W$, apply \cref{lem:state entailment} to
$T$ and $X_w\subset V(T)$ and states $q_2=L(w)$ and $q_1=L(\ast w)$, obtaining, for each $u\in V(T)$
an enumerator for the set $Y_w(u)$ all descendants  $u'\in X_w$ of $u$ that satisfy \[[u'\mapsto L(w)]\vdash
[u\mapsto L(\ast w)].\]

In particular, for $u=G(w)$ we obtain an enumerator for the set $Y_w(G(w))$ of all descendants~$u'$ of~$G(w)$ that belong to $X_w$ and satisfy the above. Now, functions $G'\from W\to V(T)$ that certify satisfiability of $G$ are exactly all functions $G'\from W\to V(T)$
such that $G'(w)\in Y_w(G(w))$, for all $w\in W$. Therefore, all such functions can be enumerated with constant delay, by taking the product of the enumerators of the sets $Y_w(G(w))$, for $w\in W$.

Hence, we may compute in time linear in $|T|$ a data structure that, given $v\in V(T)$, enumerates all pairs $(G,G')$ such that $G\from W\to \str A_v$ is satisfiable, and
$G'\from W \to V(T)$ certifies that.
Fix such a pair $(G,G')$.
It remains to construct a data structure that allows to enumerate all tuples $(\tup a_w)_{w\in W}$ such that $\tup a_w$ matches $\sigma_w$ with starting node $G'(w)$, for all $w\in W$. Again, this can be done by enumerating, independently for each $w\in W$, all valuations $\tup a_w$ that match $\sigma_w$ with starting node $G'(w)$.
Apply the inductive assumption to $\sigma_w$ for each $w\in W$,
yielding in time linear in $|T|$ for every node $u\in V(T)$
an enumerator with constant delay for the set $Z_w(u)$ of the set of all valuations $\tup a_w$ that match $\sigma_w$ with starting node $u$.
Then tuples $(\tup a_w)_{w\in W}$ as above are exactly the tuples in $\prod_{w\in W}Z_w(G'(w))$. Therefore, all such tuples may be enumerated with constant delay, by taking the product of the enumerators for the sets $Z_w(G'(w))$ for $w\in W$.

Reassuming, we may compute in time linear in $|T|$ a data structure that, given $v\in V(T)$, enumerates all triples $(G,G',(\tup a_w)_{w\in W})$ that satisfy the condition of~\cref{claim:bijection}. Those correspond precisely to tuples~$\tup a$ that match $\sigma$ with starting node $v$,
as described in the claim. This yields the conclusion of the inductive step, and finishes the proof of~\cref{lem:shape enumeration}.
\end{proof}

%% file: model-checking.tex

\section{Separator logic on topological-minor-free classes}\label{sec:model-checking}

In this section, we combine all the pieces together, and show that the
model-checking, querying, and enumeration problems for the logic $\FOsep$ can be
solved efficiently on graphs excluding a fixed topological minor.
Precisely, we prove the following.

\mainub*

The following is a generalization to formulas with free variables.

\begin{theorem}\label{thm:main-query-answering}
  Let $\Cc$ be a class of graphs that exclude a fixed graph as a
  topological minor. Then given $G\in \Cc$ and an $\FOsep$ formula
  $\varphi(\tup x)$, one can compute in time $f(\phi)\cdot \|G\|^3$ a data structure that allows to answer whether a given tuple $\tup a$ satisfies $\phi(\tup x)$ in $G$ in time
  $f(\phi)$, where $f$ is a computable function depending on $\Cc$.
\end{theorem}

Finally, the following result computes an enumerator for all the answers to a given query.

\begin{theorem}\label{thm:main-enumeration}
  Let $\Cc$ be a class of graphs that exclude a fixed graph as a
  topological minor. Then given $G\in \Cc$ and an $\FOsep$ formula
  $\varphi(\tup x)$, one can enumerate all answers to $\phi(\tup x)$
  in $G$ with delay $f(\phi)$ after preprocessing in time
  $f(\phi)\cdot \|G\|^3$, where $f$ is a computable function depending on $\Cc$.
\end{theorem}

The two results above follow from \Cref{thm:formula
  evaluation,thm:answering,thm:formula enumeration}, respectively, and the following lemma. 
  The \emph{Gaifman graph} of a $\Sigma$-structure $\str A$ is a graph with
vertex set $V(\str A)$ where two distinct vertices $u,v$ are connected
by an edge if $u$ and $v$ appear together in a relation of $\Sigma$.
We say that a $\Sigma$-structure excludes a topological minor $H$ if
its Gaifman graph excludes $H$ as a topological minor.

\begin{lemma}\label{lem:fo-sep-translation}
  For every graph $H$ and \FOsep~formula $\phi(\tup x)$ there exists a
  graph $H'$, an alphabet $A$, a signature $\Sigma$ consisting of
  binary relation symbols including the edge relation symbol $E$, and
  a formula $\psi(\tup x)$ of $\FO(\MSO(\preccurlyeq, A), \Sigma)$
  such that the following holds. Given a graph $G$ that excludes $H$
  as a topological minor, one can in time $f(\phi,H)\cdot \|G\|^3$
  compute a tree $T$ of size $f(\phi,H)\cdot \|G\|$, vertex-labeled with $A$,
  and augmented with
  $H'$-topological-minor-free $\Sigma$-structures, such that
  $V(G)\subseteq V(T)$ and $\phi(G)=\psi(T)$; here, $f$ is a
  computable function. Moreover, $H',A,\Sigma,\psi(\tup x)$
    can be computed from $H$ and $\phi(\tup x)$.
\end{lemma}

We remark that the logic $\FO(\MSO(\preccurlyeq, A), \Sigma)$
in the statement could be replaced by the weaker logic $\FO(\preceq,A,\Sigma)$.
This can be deduced a posteriori from the above statement,
using
\cref{rem:to-FO}. The labels from $A$ can be further eliminated by
encoding them as gadgets in the tree, so the logic $\FO(\preceq,\Sigma)$ would suffice, too.
However, we do not need this stronger version of the lemma here.

\begin{proof}[Proof of \Cref{thm:main-ub,thm:main-query-answering,thm:main-enumeration}]
  Let $\CC$ be a class of graphs that exclude a graph $H$ as a topological minor
  and let  $\phi(\tup x)$ be an {\FOsep} formula.
Let $H'$, $A$, and $\psi(\tup x)$ be as in~\cref{lem:fo-sep-translation}.
Let $\CC'$ be the class of all graphs that exclude $H'$ as a topological minor.
In particular, $\CC'$ has bounded expansion, and hence has efficient model-checking (by the main result of~\cite{DvorakKT13}), query-answering, and query enumeration (by the main result of~\cite{DBLP:journals/lmcs/KazanaS19}, see also~\cite{10.1145/3375395.3387660}). 

Applying \Cref{thm:formula
evaluation,thm:answering,thm:formula enumeration} respectively to the class $\CC'$,
we conclude that model-checking, query-answering and query enumeration is efficient
for $\FOMSO$ on $\CC'$-augmented trees that are vertex-labeled with $A$.

Now,
given a graph $G\in\CC$ compute in time $f(\phi,H)\cdot \|G\|^3$ a tree $T$, augmented with graphs from $\CC'$ and labeled with $A$, as in~\cref{lem:fo-sep-translation}.
Apply the appropriate (model-checking/query-answering/enumeration) algorithm to the tree $T$ and the formula $\psi(\tup x)$.
This runs in time linear in the size of $|T|$. As $|T|$ has size at most  $f(\phi,H)\cdot \|G\|$, this time is linear in $\|G\|$.
The algorithm yields the required data structure for $\psi(\tup x)$
(that is, gives a boolean answer if $\psi$ is a sentence, or a query-answering data structure, or an enumerator). As $\psi(T)=\phi(G)$, this is also the required data structure for $\phi$.
\end{proof}

\medskip
The remainder of this section is devoted to the proof of
\cref{lem:fo-sep-translation}.

\bigskip

Let $k$ be such that every $\conn$ predicate used in $\phi(\tup x)$
has arity at most $k+2$, and let $q\coloneqq q(k)$ be the parameter
provided by \cref{thm:strong-unbreakability}. Define
$$H'\coloneqq K_{\max(|H|,2q+2)}.$$

\paragraph*{Construction of $T$.}
We now show how, given a graph $G$ that excludes $H$ as a topological
minor, to construct a suitable augmented tree $T$. The alphabet $A$
will be defined within this construction.

First, apply the algorithm of \cref{thm:strong-unbreakability} to the
graph $G$, obtaining a strongly $(q,k)$-unbreakable tree
decomposition $\Tt=(T_0,\bag)$ of adhesion at most $q$. Construct $T$
from $T_0$ as follows: for each $x\in V(T_0)$ and $u\in \bag(x)$,
create a copy $u_x$ of $u$ and attach $u_x$ as a (leaf) child of $x$.
When $x$ is the unique node satisfying $u\in \mrg(x)$, we identify the
copy $u_x$ with $u$ itself, so that we have $V(G)\subseteq V(T)$. Next,
for every node $x\in V(T_0)$, we construct the whorl at $x$ to be the
graph $\bgraph(x)$; for this we use the binary edge relation $E$ on
the children of $x$. By~\cref{lem:bag graph}, each graph $\bgraph(x)$
excludes $H'$ as a topological minor, hence $T$ is augmented with
$H'$-topological-minor-free graphs.

It remains to decorate the nodes of $T$ using a finite alphabet $A$
and to decorate the edges of whorls $\{\bgraph(x)\colon x\in V(T_0)\}$
using a finite signature $\Sigma'$ consisting of binary symbols so
that the original structure of the tree decomposition $\Tt$ can be
recovered from this decoration. This way, the whorls become
$\Sigma$-structures rather than just graphs, where
$\Sigma=\Sigma'\cup \{E\}$.

Apply \cref{lem:translation} to the tree decomposition $\Tt$. This we, in
time $q^{\Oh(1)}\cdot |T|\|G\|$ we construct suitable injective
colorings $\lambda_x\colon \adh(x)\to [2q]$ for $x\in V(T_0)\}$ and a
labeling $\sigma\colon E(T_0)\to \Basic$, where $\Basic=\Basic_{2q}$
is the set of all basic bi-interface graphs of arity $2q$, with
properties asserted by \cref{lem:translation}. We now define the alphabet
$A$ and signature $\Sigma'$ as follows: \begin{itemize}
 \item $A$ consists of unary predicates $W$, $\{C_i\colon i\in [2q]\}$, and $\{R_{\Bb}\colon \Bb \in \Basic\}$.
 \item $\Sigma'$ consists of binary predicates $\{D_i\colon i\in [2q]\}$.
\end{itemize}
We decorate $T$ with those predicates as follows:
\begin{itemize}
\item Predicate $W$ selects all nodes of $T_0$.
\item For $i\in [2q]$, predicate $C_i$ selects all nodes $x\in V(T_0)$
  such that $i\in \lambda_x(\adh(x))$. Furthermore, for each such node
  $x$, $C_i$ also selects the unique copy of the form $u_x$, where $u$
  is such that $\lambda_x(u)=i$.
\item For every $\Bb\in \Basic$, $R_{\Bb}$ selects all nodes
  $y\in V(T_0)$ that have a parent, say $x$, and such that
 $$\sigma(xy)=\Bb.$$
\item For $i\in [2q]$ and each $x\in V(T_0)$, predicate $D_i$ selects
  all edges of the form $yu_x$, where $y$ is a child of $x$ in $T_0$
  and $u\in \adh(y)$ is such that $\lambda_y(u)=i$.
\end{itemize}

Note that in the last point, edges decorated with predicates $D_i$ are
already present in $\bgraph(x)$, hence adding those predicates does
not alter the Gaifman graph of the whorl. We now show that the above
information is enough to recover the structure of $G$ using the logic
$\FO(\MSO(\preccurlyeq, A), \Sigma)$. Given colorings $\lambda_x$ and
labeling $\sigma$, it is straightforward to compute those decorations,
and therefore the final augmented tree $T$, in time
$\Oh(|B|\cdot |T|)=2^{\Oh(q^2)}\cdot |T|$.

A node $s$ of $T$ is called a {\em{representative}} of a vertex
$u\in V(G)$ if $s=u_x$ for the unique node $x\in V(T_0)$ satisfying
$u\in \mrg(x)$. Note that every vertex $u$ of $G$ has exactly one
representative, which is exactly the node of $T$ identified with $u$.
We note that representatives can be conveniently distinguished and
manipulated.

\newcommand{\rep}{\rho}
\newcommand{\adj}{\alpha}
\newcommand{\cpy}{\iota}

\begin{claim}\label{cl:reps}
  There is an $\MSO(\preccurlyeq, A)$ formula $\rep(s)$ that is
  satisfied exactly by the representatives. \end{claim}
\begin{claimproof}
  It suffices to write
 $$\rep(s)\coloneqq \neg\, W(s)\wedge \bigwedge_{i=1}^{2q} \neg\, C_i(s).\qedhere$$
\end{claimproof}

\begin{claim}\label{cl:copy}
  There is an $\FO(\MSO(\preccurlyeq, A), \Sigma)$ formula $\cpy(x,y)$
  such that $\cpy(u,t)$ holds if and only if $u$ is a representative
  and $t$ is a copy of $u$.
\end{claim}
\begin{claimproof}
  We first check that $\neg W(u)$, $\neg W(t)$, and that $u$ is a
  representative using formula~$\rep(u)$ provided by \cref{cl:reps}.
  Next, from the properties of colorings $\lambda_x$ provided by
  \cref{lem:translation} it follows that~$t$ is a copy of $u$ if and
  only if the following condition holds: either $u=t$, or $u$ has a
  sibling $z$ that is an ancestor of $t$ and there is $i\in [2q]$ such
  that $D_i(u,z)$ holds and $C_i(w)$ holds for all
  $z\preccurlyeq w\preccurlyeq t$. It is straightforward to express
  this check in $\FO(\MSO(\preccurlyeq, A), \Sigma)$.
\end{claimproof}

Next, we observe that the original adjacency relation in $G$ can be
defined between the representatives.

\begin{claim}\label{cl:adj}
  There is an $\FO(\MSO(\preccurlyeq, A), \Sigma)$ formula $\adj(x,y)$
  such that $\adj(u,v)$ holds in $T$ if and only if $u$ and $v$ are
  representatives and $uv\in E(G)$.
\end{claim}
\begin{claimproof}
  By the properties of tree decompositions, if $u$ and $v$ are
  adjacent, then either $u$ belongs to the unique bag where $v$
  belongs to the margin, or vice versa. Therefore,
 it suffices to write
 $$\adj(u,v)\coloneqq \rep(u)\wedge \rep(v)\wedge \left(\exists t.(E(u,t)\wedge \cpy(v,t))\vee \exists t.(E(v,t)\wedge \cpy(u,t))\right).$$
 Here, $\cpy(x,y)$ is the formula provided by \cref{cl:copy}.
\end{claimproof}

Finally, the information encoded through predicates $R_{\Bb}$ can be
used to express torso queries.

\begin{claim}\label{cl:torso}
  For every $\Ab\in B$ there is an $\MSO(\preccurlyeq,A)$ formula
  $\tau_\Bb(x,y)$ that selects all pairs of nodes $(x,y)$ of\, $T$
  such that $x,y\in V(T_0)$, $x$ is a strict ancestor of $y$ in $T_0$,
  and $\sigma(x,y)=\Ab$.
\end{claim}
\begin{claimproof}
  Checking the first two assertions boils down to verifying that
  $W(x)$ and $W(y)$ hold and that $x\prec y$. For the third assertion,
  the idea is to compose the edge labels under $\sigma$, encoded
  through predicates $R_{\Bb}$, along the $x$-$y$ path using monadic
  existential quantification.

  Let $P$ be the set of all nodes on the $x$-$y$ path in $T$,
  including $x$ but excluding $y$.
  For $\Bb\in \Basic$, let $P_\Bb$ be the set of those $w\in P$ for
  which
 $$\sigma(w,y)=\Bb.$$
 Clearly, $\{P_{\Bb}\colon \Bb\in \Basic\}$ is a partition of $P$ into
 $2^{\Oh(q^2)}$ sets.
 We guess this partition using existential quantification and verify
 the correctness of the guess. This verification boils down to
 checking the following assertions:
 \begin{itemize}[nosep]
 \item Sets $P_{\Bb}$ form a partition of $P$.
 \item If $y'$ is the parent of $y$, then $y'\in P_{\Bb}$ for $\Bb$
   such that $R_\Bb(y)$ holds.
  \item For every $w\in P\setminus \{x\}$, if $w'$ is the parent of $w$, $w\in P_{\Bb}$, $w'\in P_{\Bb'}$, and $R_{\Cb}(w)$ holds, then
 $$\Cb\odot \Bb=\Bb'.$$
 \end{itemize}
 It is straightforward to express these assertions in
 $\MSO(\preccurlyeq,A)$.

 Now $\tau_{\Ab}(x,y)$ should hold if and only if $x\in P_{\Ab}$.
\end{claimproof}

\paragraph*{Rewriting the query.} The remainder of the argument is
devoted to the proof of the following claim, which boils down to the
statement of the theorem restricted to a single $\conn$ query.

\begin{claim}\label{cl:conn}
  There is an $\FO(\MSO(\preccurlyeq, A), \Sigma)$ formula
  $\chi(x,y,z_1,\ldots,z_k)$ such that
 $$\conn_k(G)=\chi(T).$$
\end{claim}

Before we proceed to the proof of \cref{cl:conn}, let us see why the
theorem statement follows from it. Indeed, given a formula
$\phi(\tup x)$ of $\FOsep$, we can rewrite it to a formula
$\psi(\tup x)$ of $\FO(\MSO(\preccurlyeq, A), \Sigma)$ satisfying
$\phi(G)=\psi(T)$ as follows:
\begin{itemize}
\item Quantification is relativized to the representatives by guarding
  each quantifier with formula $\rep$ provided by \cref{cl:reps}:
  $\exists x.\delta$ becomes $\exists x.(\rep(x)\wedge \delta)$ and
  $\forall x.\delta$ becomes $\forall x.(\rep(x)\rightarrow \delta)$.
\item Every atomic edge relation check $E(x,y)$ is replaced by the
  formula $\adj(x,y)$ provided by \cref{cl:adj}.
 \item Every occurrence of the predicate $\conn_\ell(x,y,z_1,\ldots,z_\ell)$, where $\ell\leq k$, is replaced with the formula $\chi(x,y,z_1,\ldots,z_\ell,z_\ell,\ldots,z_\ell)$ provided by \cref{cl:conn}. Here $z_\ell$ is repeated $k-\ell+1$ times.
\end{itemize}
Hence, it remains to prove \cref{cl:conn}. The argument boils down to
turning the algorithmic procedure presented in
\cref{sec:general-queries} into a formula of
$\FO(\MSO(\preccurlyeq, A), \Sigma)$.

\bigskip

Suppose then that we are given $s,t,u_1,\ldots,u_k\in V(G)$ and we
would like to verify whether $G\models \conn(s,t,u_1,\ldots,u_k)$. We
describe a procedure that performs this verification using the data
provided in $T$, and along the way we describe how the procedure is
turned into a suitable formula $\chi$ of
$\FO(\MSO(\preccurlyeq, A), \Sigma)$. We first check that
$s,t,u_1,\ldots,u_k$ are representatives, so we may proceed under this
assumption.

Similarly as in \cref{sec:general-queries}, let $X$ consist of the
root of $T_0$ and all nodes of $T_0$ whose margins intersect
$\wh{S}\coloneqq \{s,t,u_1,\ldots,u_k\}$. Further, let $Y$ be the
least common ancestor closure of $X$ in $T$. Note that $X$ and $Y$ can
be easily defined in $\MSO(\preccurlyeq, A)$ and recall that
$|Y|\leq 2k+5$. Recall that a node $z$ of $T$ is {\em{affected}} if
$\cmp(z)$ intersects $\wh{S}$. Construct $Y'$ by taking $Y$ and
adding, for each $y\in Y$, all affected children of $Y$. It is
straightforward to see that $|Y'|\leq 2|Y|-1\leq 4k+9$. Note that $X$,
$Y$, and $Y'$ can be easily defined from $s,t,u_1,\ldots,u_k$ in
$\MSO(\preccurlyeq,A)$.

As in \cref{sec:general-queries}, for $x\in V(T_0)$ we define
$$D(x)\coloneqq \adh(x)\cup (\{s,t\}\cap \cone(x))$$
and
$$\profile(x)\coloneqq \left\{\ \{a,b\}\in \binom{D(x)}{2}\quad \Big|\quad \textrm{in }G[\cone(x)]\textrm{ there is an }a\textrm{-}b\textrm{ path avoiding }S\ \right\}.$$

Recall that the procedure described in \cref{sec:general-queries}
computed sets $\profile(y)$ for all $y\in Y'$ in a bottom-up manner.
Each set $\profile(y)$ can be encoded using
$\binom{|D(y)|}{2}\leq \binom{q+2}{2}$ bits of information, denoting
which pairs are included in $\profile(y)$. Observe that for a given
$y\in Y'$, the set $D(y)$ (encoded through a sequence of
representatives) can be defined in
$\FO(\MSO(\preccurlyeq, A), \Sigma)$ with the help of
formula~$\cpy(x,y)$ provided by \cref{cl:copy}. Therefore, we
construct the formula $\chi$ by making a disjunction over all
$\prod_{y\in Y'} 2^{\binom{|D(y)|}{2}}\leq 2^{(4k+9)\binom{q+2}{2}}$
possible guesses on how the profiles of nodes of $Y'$ exactly look
like. For each guess we verify that it encodes the actual profiles of
nodes of $Y'$, and that $\{s,t\}$ belongs to the profile of the root
of $T$.

Therefore, it remains to show, for a given guess of the profiles, that
this guess is correct. For $y\in Y'$, let
$\Pi_y\subseteq \binom{D(y)}{2}$ be the guess on the profile of $y$.
We need to verify the following two assertions:
\begin{enumerate}[(C1)]
\item\label{a:warp} For every $z\in Y'\setminus Y$, if $y$ is the
  unique descendant of $z$ that belongs to $Y$ and such that there is
  no $w\in Y'$ with $z\prec w\prec y$, then the profile $\Pi_z$ is
  obtained from the profile $\Pi_y$ as described in \cref{lem:warp}.
 \item\label{a:aggregate} For every $x\in Y$, if $Z\subseteq Y'$ is the set of affected children of $x$, then the profile $\Pi_x$ is obtained from profiles $\{\Pi_z\colon z\in Z\}$ as described in \cref{lem:aggregate}.
\end{enumerate}
We show how to implement those checks in
$\FO(\MSO(\preccurlyeq, A), \Sigma)$ separately. These implementations
correspond to the proofs of \cref{lem:warp} and of
\cref{lem:aggregate}, respectively.

\paragraph*{Implementing \ref{a:warp}.} Note that for a given
$z\in Y'\setminus Y$, it can be distinguished in
$\MSO(\preccurlyeq,A)$ which element of $Y$ is the corresponding node
$y$ as in the statement of assertion \ref{a:warp}. As in the proof of
\cref{lem:warp}, let $J$ be the graph on vertex set $D(y)\cup \adh(z)$
whose edge set is $\Pi_y\cup E(\torso(z,y))$. As argued in
\cref{cl:warp-J}, we need to verify that $\Pi_y$ is exactly the
connectivity relation in the graph $J-S$, restricted to $D(z)$. Since
$J$ is a graph on at most $2q+2$ vertices, it suffices to define the
edge relation in $J$, since from it the connectivity relation of $J-S$
can be deduced in first-order logic. The set $\Pi_y$ is assumed to be
given through the guess. On the other hand, using formulas
$\tau_\Bb(z,y)$ for $\Bb\in \Basic$ we can deduce $\sigma(z,y)$, which
is the graph $\torso(z,y)$ labeled using colorings $\lambda_z$ and
$\lambda_y$. We may now use predicates $C_i$ on the children of $z$
and $y$, and formula $\cpy(\cdot,\cdot)$ provided by \cref{cl:copy},
to infer the edge relation of $\torso(z,y)$, thereby completing the
construction of the edge set of $J$.

\paragraph*{Implementing \ref{a:aggregate}.} Now we are given
$x\in Y$, the guess $\Pi_x$ for the profile of $x$, and, for each
affected child $z$ of $y$, the guess $\Pi_z$ for the profile of $z$.
We first observe that we can define the graph $\bgraph_S(x)$ in
$\FO(\MSO(\preccurlyeq, A), \Sigma)$ as follows:
\begin{itemize}
\item The vertex set of $\bgraph_S(x)$ comprises all children of $x$
  in $T$ that do not satisfy predicate $W$ nor are contained in $S$.
\item Two such vertices $u$ and $v$ are adjacent in $\bgraph_S(x)$ if
  and only if they are adjacent in $\bgraph(x)$, or there is a child
  $z$ of $y$ that satisfies $W$, is adjacent to both $u$ and $v$, and
  either:
 \begin{itemize}
 \item $z$ is not affected; or
 \item $z$ is affected and $\{u,v\}\in \Pi_z$.
\end{itemize}
\end{itemize}
It it straightforward to express these checks in
$\FO(\MSO(\preccurlyeq, A), \Sigma)$.

Once the graph $\bgraph_S(x)$ is defined, we can use
\cref{lem:short-BFS-combinatorial} to observe that the following
formula correctly decides whether two vertices
$a,b\in \bag(x)\setminus S$ are connected in $G[\cone(x)]-S$: check
whether $a$ and $b$ can be connected in $\bgraph_S(x)$ by a path of
length at most $q+1$, or that in $\bgraph_S(x)$ there exist connected
sets of vertices $K_a$ and $K_b$, each of size $q+1$ and containing
$a$ and $b$, respectively. This can be clearly expressed in
first-order logic.

Using the formula constructed above we can verify the correctness of
the guess on the profile~$\Pi_x$ on pairs
$\{a,b\}\in \binom{D(x)\cap \bag(x)}{2}$. It remains to verify it also
on pairs where either $a$ or $b$ belong to
$D(x)\setminus \bag(x)\subseteq \{u,v\}$. This can be done analogously
as in the proof of \cref{lem:aggregate} using the formula~$\tau$
provided by \cref{cl:torso}; we leave the details to the reader.

This concludes the proof of \cref{lem:fo-sep-translation}. $\hfill\square$

%% file: first-levels.tex
\section{Lower bounds}\label{sec:lbs}

We proved that the model-checking problem for separator logic is
fixed-parameter tractable on classes of graphs that exclude a fixed
topological minor. The next natural question is whether this result
can be improved in terms of generality. In this section we prove
\cref{thm:main-lb} that shows that, under technical assumptions, this
is not the case.


The reduction is based on the following lemma.

\begin{lemma}\label{lem:transduce-all-graphs}
  Let $\Cc$ be a class of graphs that is closed under taking subgraphs
  and is not topological-minor-free. Then for every sentence $\phi$ of
  $\FO$ over graphs, there is an $\FOsep$ sentence $\psi$ such that
  for every graph $G$, there is a graph $H\in \Cc$ such that
  $G\models\phi$ iff $H\models \psi$.
\end{lemma}
\begin{proof}
  The sentence $\psi$ is obtained from $\phi$ by the following syntactic
  modifications. First, each quantifier is restricted to the set of
  vertices of degree at least $3$.  Second, we replace each atom of
  the form $E(x,y)$ by the following $\FOsep$ formula, which states
  that $x$ and $y$ are connected by a path of vertices of degree $2$.
  $$\phi_E(x,y) \coloneqq E(x,y) \vee \exists z. \Big( d^{=2}(z) \wedge \conn(x,z) \wedge \conn(y,z) \wedge \forall w.  \big( \conn(z,w,x,y) \rightarrow d^{=2}(w) \big)\Big),$$
  Here $d^{=2}(x)$ is a formula checking that $x$ has degree two. Note
  that $\psi$ is derived from $\phi$ only, it does not depend on
  $\Cc$.

  Now let $G$ be any graph, we define $H\in \Cc$ as follows. We first
  add two private neighbors to each vertex of $G$. This yields a graph
  $G'$. Since $\Cc$ is not topological-minor-free, $G'$ is a
  topological minor of some graph $H\in \Cc$. As $\Cc$ is closed under
  taking subgraphs, we can assume that $H$ is a subdivision of~$G'$.

  We finally have $G\models\phi$ iff $H\models \psi$. To see this,
  note that all original vertices of $G$ have degree at least~$3$ is
  $G'$; and that, in $H$, two vertices of degree at least three
  satisfy $\phi_E$ if and only if they are adjacent or connected by a
  path of nodes of degree two.
\end{proof}

Unfortunately, \cref{lem:transduce-all-graphs} in general does not
provide an fpt reduction from $\FO$ model-checking on general graphs
to $\FOsep$ model-checking on classes $\Cc$ as in the lemma
statement. This is because the graph $H$ whose existence is guaranteed
by the assumptions might not be of size polynomial in the size of $G$,
making the reduction non-fpt.

We get around this obstacle by adding an appropriate
assumption. Recall that a class $\Cc$ admits \emph{efficient encoding
  of topological minors} if for every graph $G$ there exists
$H\in \Cc$ such that $G$ is a topological minor of $H$, and, given
$G$, such $H$ together with a suitable topological minor model can be
computed in time polynomial in $|G|$. Note that in particular such $H$
can be only polynomially larger than $G$. We may now
prove~\cref{thm:main-lb}.

\mainlb*
\begin{proof}
  Recalling that $\FO$ model-checking on general graphs is
  $\AWstar$-complete~\cite{FlumG06}, it suffices to apply the
  reduction presented in the proof of \cref{lem:transduce-all-graphs}
  and note that it can be performed in polynomial time due to the
  assumption of polynomial embedding of topological minors.
\end{proof}

It is well-known that the subgraph closure of the class of
$n$-subdivisions of $n$-vertex complete graphs has bounded
expansion. As this class admits efficient encoding of topological
minors, we obtain the following.

\begin{corollary}\label{cor:hard-be}
  There is a class $\Cc$ with bounded expansion such that
  $\FOsep$ model-checking on $\Cc$ is $\AWstar$-hard.
\end{corollary}

\Cref{cor:hard-be} provides a separation between the parameterized
complexity of $\FO$ and $\FOsep$ model-checking, as the former problem
is fixed-parameter tractable on every class of bounded
expansion~\cite{DvorakKT13}.

%% file: conclusions.tex
\section{Discussion}\label{sec:conclusions}

In this work we have added computational problems related to $\conn$
queries in graphs and the logic $\FOsep$ to the growing list of
applications of the decomposition theorem of Cygan et
al.~\cite{CyganLPPS19}. Several questions can be asked about possible
improvements of the obtained complexity bounds, especially regarding
the data structure of \cref{thm:general-queries}.

First, in both the model-checking algorithm of \cref{thm:main-ub} and
the construction algorithm of the data structure of
\cref{thm:general-queries}, the main bottleneck for the polynomial
factor is the running time of the algorithm to compute the tree
decomposition of Cygan et al. (see \cref{thm:strong-unbreakability}).
This running time is cubic in the size of the graph, while in both
applications presented in this work, the remainder of the construction
takes fpt time with at most quadratic dependence on the graph size.
Therefore, it seems imperative to revisit
\cref{thm:strong-unbreakability} with the purpose of improving the
time complexity of the construction algorithm. We note that polynomial
factor of the running time of the algorithm of
\cref{thm:weak-unbreakability} is not specified
in~\cite{CyganKLPPSW21}, but it is also at least cubic.

Second, in \cref{sec:tradeoff} we showed how to improve the time
complexity of queries to polynomial in $k$ by using the weakly
unbreakable tree decomposition provided by
\cref{thm:weak-unbreakability}. The drawback is that the space usage
of the data structure becomes $k^{\Oh(1)}\cdot |G|^2$. However, the
quadratic dependence on $|G|$ is caused only by the brute-force
implementation of the data structure for torso queries provided by
\cref{lem:torso-ds-trivial}. So now we have two implementations of
this data structure: \begin{itemize}[nosep]
\item \cref{lem:torso-ds} offers query time $2^{\Oh(q^2)}$ and space
  usage $2^{\Oh(q^2)}\cdot |T|$.
 \item \cref{lem:torso-ds-trivial} offers query time $q^{\Oh(1)}$ and space usage $q^{\Oh(1)}\cdot |T|^2$.
\end{itemize}
Is it possible to design a data structure that would offer query time
$q^{\Oh(1)}$ while keeping the space usage at $q^{\Oh(1)}\cdot |T|$?
In our proof of \cref{lem:torso-ds}, we apply the more general
\cref{thm:colcombet-ds} to the semigroup of bi-interface graphs of
arity $2q$. It is possible that a better understanding of the
structure of this semigroup might lead to an improvement proposed
above.

%% file: main.bbl
\newcommand{\etalchar}[1]{$^{#1}$}
\begin{thebibliography}{AKP{\etalchar{+}}21b}

\bibitem[ABMN19]{AmarilliBMN19}
Antoine Amarilli, Pierre Bourhis, Stefan Mengel, and Matthias Niewerth.
\newblock Enumeration on trees with tractable combined complexity and efficient
  updates.
\newblock In {\em Proceedings of the 38th {ACM} {SIGMOD-SIGACT-SIGAI} Symposium
  on Principles of Database Systems, {PODS} 2019}, pages 89--103. {ACM}, 2019.

\bibitem[AKP{\etalchar{+}}21a]{DBLP:conf/stacs/AgrawalKP0021}
Akanksha Agrawal, Lawqueen Kanesh, Fahad Panolan, M.~S. Ramanujan, and Saket
  Saurabh.
\newblock An {FPT} algorithm for elimination distance to bounded degree graphs.
\newblock In {\em Proceedings of the 38th International Symposium on
  Theoretical Aspects of Computer Science, {STACS} 2021}, volume 187 of {\em
  LIPIcs}, pages 5:1--5:11. Schloss Dagstuhl --- Leibniz-Zentrum f{\"{u}}r
  Informatik, 2021.

\bibitem[AKP{\etalchar{+}}21b]{agrawal2021fpt}
Akanksha Agrawal, Lawqueen Kanesh, Fahad Panolan, MS~Ramanujan, and Saket
  Saurabh.
\newblock An fpt algorithm for elimination distance to bounded degree graphs.
\newblock In {\em 38th International Symposium on Theoretical Aspects of
  Computer Science (STACS 2021)}. Schloss Dagstuhl-Leibniz-Zentrum f{\"u}r
  Informatik, 2021.

\bibitem[Bag06]{Bagan06}
Guillaume Bagan.
\newblock {MSO} queries on tree decomposable structures are computable with
  linear delay.
\newblock In {\em Proceedings of the 15th International Conference on Computer
  Science Logic, {CSL} 2006}, volume 4207 of {\em Lecture Notes in Computer
  Science}, pages 167--181. Springer, 2006.

\bibitem[BD16]{bulian2016graph}
Jannis Bulian and Anuj Dawar.
\newblock Graph isomorphism parameterized by elimination distance to bounded
  degree.
\newblock {\em Algorithmica}, 75(2):363--382, 2016.

\bibitem[BD17]{bulian2017fixed}
Jannis Bulian and Anuj Dawar.
\newblock Fixed-parameter tractable distances to sparse graph classes.
\newblock {\em Algorithmica}, 79(1):139--158, 2017.

\bibitem[Boj09]{10.1007/978-3-642-02737-6_1}
Miko\l{}aj Boja\'nczyk.
\newblock Factorization forests.
\newblock In {\em Proceedings of the 13th International Conference on
  Developments in Language Theory, {DLT} 2009}, volume 5583 of {\em Lecture
  Notes in Computer Science}, pages 1--17. Springer, 2009.

\bibitem[Boj21]{Bojanczyk21}
Miko\l{}aj Boja\'nczyk.
\newblock Separator logic and star-free expressions for graphs.
\newblock {\em CoRR}, abs/2107.13953, 2021.

\bibitem[BP16]{BojanczykP16}
Miko\l{}aj Boja\'n{}czyk and Micha\l{} Pilipczuk.
\newblock Definability equals recognizability for graphs of bounded treewidth.
\newblock In {\em Proceedings of the 31st Annual {ACM/IEEE} Symposium on Logic
  in Computer Science, {LICS} 2016}, pages 407--416. {ACM}, 2016.

\bibitem[CCH{\etalchar{+}}16]{ChitnisCHPP16}
Rajesh Chitnis, Marek Cygan, MohammadTaghi Hajiaghayi, Marcin Pilipczuk, and
  Micha\l{} Pilipczuk.
\newblock Designing {FPT} algorithms for cut problems using randomized
  contractions.
\newblock {\em {SIAM} Journal on Computing}, 45(4):1171--1229, 2016.

\bibitem[CKL{\etalchar{+}}21]{CyganKLPPSW21}
Marek Cygan, Pawe\l{} Komosa, Daniel Lokshtanov, Marcin Pilipczuk, Micha\l{}
  Pilipczuk, Saket Saurabh, and Magnus Wahlstr{\"{o}}m.
\newblock Randomized contractions meet lean decompositions.
\newblock {\em {ACM} Transactions on Algorithms}, 17(1):6:1--6:30, 2021.

\bibitem[CLP{\etalchar{+}}19]{CyganLPPS19}
Marek Cygan, Daniel Lokshtanov, Marcin Pilipczuk, Micha\l{} Pilipczuk, and
  Saket Saurabh.
\newblock Minimum {B}isection is fixed-parameter tractable.
\newblock {\em {SIAM} Journal on Computing}, 48(2):417--450, 2019.

\bibitem[Col07]{Colcombet07}
Thomas Colcombet.
\newblock A combinatorial theorem for trees.
\newblock In {\em Proceedings of the 34th International Colloquium Automata,
  Languages and Programming, {ICALP} 2007}, volume 4596 of {\em Lecture Notes
  in Computer Science}, pages 901--912. Springer, 2007.

\bibitem[Cou90]{Courcelle90}
Bruno Courcelle.
\newblock The {M}onadic {S}econd-{O}rder logic of graphs. {I}. {R}ecognizable
  sets of finite graphs.
\newblock {\em Information and Computation}, 85(1):12--75, 1990.

\bibitem[DKT13]{DvorakKT13}
Zden\v{e}k Dvo\v{r}{\'{a}}k, Daniel Kr{\'{a}}l, and Robin Thomas.
\newblock Testing first-order properties for subclasses of sparse graphs.
\newblock {\em Journal of the {ACM}}, 60(5):36:1--36:24, 2013.

\bibitem[DP10]{DuanP10}
Ran Duan and Seth Pettie.
\newblock Connectivity oracles for failure prone graphs.
\newblock In {\em Proceedings of the 42nd {ACM} Symposium on Theory of
  Computing, {STOC} 2010}, pages 465--474. {ACM}, 2010.

\bibitem[DP20]{DuanP20}
Ran Duan and Seth Pettie.
\newblock Connectivity oracles for graphs subject to vertex failures.
\newblock {\em {SIAM} Journal on Computing}, 49(6):1363--1396, 2020.

\bibitem[FG06]{FlumG06}
J{\"{o}}rg Flum and Martin Grohe.
\newblock {\em Parameterized Complexity Theory}.
\newblock Texts in Theoretical Computer Science. An {EATCS} Series. Springer,
  2006.

\bibitem[FGT21]{DBLP:conf/lics/FominGT21}
Fedor~V. Fomin, Petr~A. Golovach, and Dimitrios~M. Thilikos.
\newblock Parameterized complexity of elimination distance to first-order logic
  properties.
\newblock In {\em Proceedings of the 36th Annual {ACM/IEEE} Symposium on Logic
  in Computer Science, {LICS} 2021}, pages 1--13. {IEEE}, 2021.

\bibitem[FI97]{frigioni1997dynamically}
Daniele Frigioni and Giuseppe~F. Italiano.
\newblock Dynamically switching vertices in planar graphs.
\newblock In {\em Proceedings of the 5th Annual European Symposium on
  Algorithms, {ESA} 1997}, volume 1284 of {\em Lecture Notes in Computer
  Science}, pages 186--199. Springer, 1997.

\bibitem[GHL{\etalchar{+}}14]{ganian2014lower}
Robert Ganian, Petr Hlin{\v{e}}n{\'y}, Alexander Langer, Jan
  Obdr{\v{z}}{\'a}lek, Peter Rossmanith, and Somnath Sikdar.
\newblock Lower bounds on the complexity of mso1 model-checking.
\newblock {\em Journal of Computer and System Sciences}, 80(1):180--194, 2014.

\bibitem[GKS17]{GroheKS17}
Martin Grohe, Stephan Kreutzer, and Sebastian Siebertz.
\newblock Deciding first-order properties of nowhere dense graphs.
\newblock {\em Journal of the {ACM}}, 64(3):17:1--17:32, 2017.

\bibitem[HT84]{HarelT84}
Dov Harel and Robert~Endre Tarjan.
\newblock Fast algorithms for finding nearest common ancestors.
\newblock {\em {SIAM} Journal on Computing}, 13(2):338--355, 1984.

\bibitem[Jd21]{JansenK21}
Bart M.~P. Jansen and Jari J.~H. {de Kroon}.
\newblock {FPT} algorithms to compute the elimination distance to bipartite
  graphs and more.
\newblock In {\em 47th International Workshop on Graph-Theoretic Concepts in
  Computer Science, {WG} 2021}, volume 12911 of {\em Lecture Notes in Computer
  Science}, pages 80--93. Springer, 2021.

\bibitem[KS13]{KazanaS13}
Wojciech Kazana and Luc Segoufin.
\newblock Enumeration of monadic second-order queries on trees.
\newblock {\em {ACM} Transactions on Computational Logic}, 14(4):25:1--25:12,
  2013.

\bibitem[KS20]{DBLP:journals/lmcs/KazanaS19}
Wojciech Kazana and Luc Segoufin.
\newblock First-order queries on classes of structures with bounded expansion.
\newblock {\em Log. Methods Comput. Sci.}, 16(1), 2020.

\bibitem[KT10]{kreutzer2010lower}
Stephan Kreutzer and Siamak Tazari.
\newblock Lower bounds for the complexity of monadic second-order logic.
\newblock In {\em Proceedings of the 25th Annual {IEEE} Symposium on Logic in
  Computer Science, {LICS} 2010}, pages 189--198. {IEEE} Computer Society,
  2010.

\bibitem[KT11]{KawarabayashiT11}
Ken{-}ichi Kawarabayashi and Mikkel Thorup.
\newblock The {M}inimum $k$-way {C}ut of bounded size is fixed-parameter
  tractable.
\newblock In {\em Proceedings of the {IEEE} 52nd Annual Symposium on
  Foundations of Computer Science, {FOCS} 2011}, pages 160--169. {IEEE}
  Computer Society, 2011.

\bibitem[L{\"o}d12]{loding2012basics}
Christof L{\"o}ding.
\newblock Basics on tree automata.
\newblock In {\em Modern Applications of Automata Theory}, pages 79--109. World
  Scientific, 2012.

\bibitem[LPPS21]{LoksthanovPPS21}
Daniel Lokshtanov, Marcin Pilipczuk, Micha\l{} Pilipczuk, and Saket Saurabh.
\newblock Fixed-parameter tractability of {G}raph {I}somorphism in graphs with
  an excluded minor, 2021.
\newblock Manuscript.

\bibitem[LRSZ18]{LokshtanovR0Z18}
Daniel Lokshtanov, M.~S. Ramanujan, Saket Saurabh, and Meirav Zehavi.
\newblock Reducing {CMSO} model checking to highly connected graphs.
\newblock In {\em Proceedings of the 45th International Colloquium on Automata,
  Languages, and Programming, {ICALP} 2018}, volume 107 of {\em LIPIcs}, pages
  135:1--135:14. Schloss Dagstuhl --- Leibniz-Zentrum f{\"{u}}r Informatik,
  2018.

\bibitem[LSS20]{Lokshtanov0S20}
Daniel Lokshtanov, Saket Saurabh, and Vaishali Surianarayanan.
\newblock A parameterized approximation scheme for {M}in $k$-{C}ut.
\newblock In {\em Proceedings of the 61st {IEEE} Annual Symposium on
  Foundations of Computer Science, {FOCS} 2020}, pages 798--809. {IEEE}, 2020.

\bibitem[LSV20]{LindermayrSV20}
Alexander Lindermayr, Sebastian Siebertz, and Alexandre Vigny.
\newblock Elimination distance to bounded degree on planar graphs.
\newblock In {\em Proceedings of the 45th International Symposium on
  Mathematical Foundations of Computer Science, {MFCS} 2020}, volume 170 of
  {\em LIPIcs}, pages 65:1--65:12. Schloss Dagstuhl --- Leibniz-Zentrum
  f{\"{u}}r Informatik, 2020.

\bibitem[PT07]{PatrascuT07}
Mihai P\u{a}tra\c{s}cu and Mikkel Thorup.
\newblock Planning for fast connectivity updates.
\newblock In {\em Proceedings of the 48th Annual {IEEE} Symposium on
  Foundations of Computer Science, FOCS 2007}, pages 263--271. {IEEE} Computer
  Society, 2007.

\bibitem[SSV18]{DBLP:conf/pods/SchweikardtSV18}
Nicole Schweikardt, Luc Segoufin, and Alexandre Vigny.
\newblock Enumeration for {FO} queries over nowhere dense graphs.
\newblock In {\em Proceedings of the 37th {ACM} {SIGMOD-SIGACT-SIGAI} Symposium
  on Principles of Database Systems, PODS 2018}, pages 151--163. {ACM}, 2018.

\bibitem[SSV21]{SchraderSV21}
Nicole Schirrmacher, Sebastian Siebertz, and Alexandre Vigny.
\newblock First-order logic with connectivity operators.
\newblock {\em CoRR}, abs/2107.05928, 2021.

\bibitem[Str19]{DBLP:journals/eatcs/Strozecki19}
Yann Strozecki.
\newblock Enumeration complexity.
\newblock {\em Bulletin of the {EATCS}}, 129, 2019.

\bibitem[SZ18]{SaurabhZ18}
Saket Saurabh and Meirav Zehavi.
\newblock Parameterized complexity of multi-node hubs.
\newblock In {\em Proceedings of the 13th International Symposium on
  Parameterized and Exact Computation, {IPEC} 2018}, volume 115 of {\em
  LIPIcs}, pages 8:1--8:14. Schloss Dagstuhl --- Leibniz-Zentrum f{\"{u}}r
  Informatik, 2018.

\bibitem[Tor20]{10.1145/3375395.3387660}
Szymon Toru\'{n}czyk.
\newblock Aggregate queries on sparse databases.
\newblock In {\em Proceedings of the 39th ACM SIGMOD-SIGACT-SIGAI Symposium on
  Principles of Database Systems, PODS 2020}, pages 427--443. ACM, 2020.

\bibitem[vdBS19]{BrandS19}
Jan van~den Brand and Thatchaphol Saranurak.
\newblock Sensitive distance and reachability oracles for large batch updates.
\newblock In {\em 60th {IEEE} Annual Symposium on Foundations of Computer
  Science, {FOCS} 2019}, pages 424--435. {IEEE} Computer Society, 2019.

\end{thebibliography}
